\documentclass{lmcs}
\pdfoutput=1

\usepackage{lastpage}
\lmcsdoi{17}{4}{23}
\lmcsheading{}{\pageref{LastPage}}{}{}%
{Jan.~14,~2021}{Dec.~23,~2021}{}


\usepackage[utf8]{inputenc}
\usepackage{verbatim}
\usepackage{adjustbox}
\usepackage{marvosym}
\usepackage{graphicx}
\usepackage{amsmath}
\usepackage{amsthm}
\usepackage{amscd}
\usepackage{mathpartir}
\usepackage{xcolor}

\usepackage{bbold}
\usepackage{url}
\usepackage{upgreek}
\usepackage{stmaryrd}

\usepackage{lipsum}
\usepackage{tikz-cd}
\usetikzlibrary{cd}
\usetikzlibrary{calc}
\usetikzlibrary{arrows}

\usepackage{bussproofs}
\EnableBpAbbreviations

\DeclareMathAlphabet{\mathpzc}{OT1}{pzc}{m}{it}

\newcommand{\mcF}{\mathcal{F}}

\newcommand{\id}{id}

\newcommand{\sem}[1]{\llbracket{#1}\rrbracket}
\newcommand{\setsem}[1]{\llbracket{#1}\rrbracket^{\set}}
\newcommand{\relsem}[1]{\llbracket{#1}\rrbracket^{\rel}}
\newcommand{\dsem}[1]{\llbracket{#1}\rrbracket^{\mathsf D}}
\newcommand{\setenv}{\mathsf{SetEnv}}
\newcommand{\relenv}{\mathsf{RelEnv}}

\newcommand{\rel}{\mathsf{Rel}}

\newcommand{\colim}[2]{{{\underrightarrow{\lim}}_{#1}{#2}}}

\newcommand{\graph}[1]{\langle {#1} \rangle}

\newcommand{\nat}{\mathpzc{Nat}}

\newcommand{\df}{\; := \;}

\newcommand{\tvars}{\mathbb{T}}
\newcommand{\fvars}{\mathbb{F}}

\newcommand{\Lan}{\mathsf{Lan}}
\newcommand{\zerot}{\mathbb{0}}
\newcommand{\onet}{\mathbb{1}}

\renewcommand{\nat}{\mathbb{N}}

\newcommand{\Set}{\mathsf{Set}}
\newcommand{\Nat}{\mathsf{Nat}}
\newcommand{\Homrel}{\mathsf{Hom_{Rel}}}

\newcommand{\cse}[3]{\mathsf{case}\,{#1}\,\mathsf{of}\,\{{#2};\,{#3}\}}
\newcommand{\tin}{\mathsf{in}}
\newcommand{\Eq}{\mathsf{Eq}}

\newcommand{\curry}{\mathsf{curry}}

\newcommand{\eval}{\mathsf{eval}}

\newcommand{\ar}[1]{\##1}

\newcommand{\ol}[1]{\overline{#1}}

\theoremstyle{plain}

\newcommand{\inl}{\mathsf{inl}}
\newcommand{\inr}{\mathsf{inr}}
\newcommand{\fold}{\mathsf{fold}}
\newcommand{\ininv}[2]{(\tin^{-1}_{\onet +
  \beta \times \phi (\phi\beta)})_{#1}\, #2}

\newcommand{\cal}{\mathcal}
\newcommand{\F}{\mathcal{F}}

\newcommand{\set}{\mathsf{Set}}

\newcommand{\C}{\mathcal{C}}
\newcommand{\D}{\mathcal{D}}

\renewcommand{\id}{\mathit{id}}
\newcommand{\map}{\mathsf{map}}

\newcommand{\semmap}{\mathit{map}}

\newcommand{\filtype}{\Nat^\emptyset
 (\Nat^\emptyset \, \alpha \, \mathit{Bool})\, (\Nat^\emptyset
  (List \, \alpha) \, (List \, \alpha))}
\newcommand{\filtypeGRose}{\Nat^\emptyset
 (\Nat^\emptyset \, \alpha \, \mathit{Bool})\, (\Nat^\emptyset
  (\mathit{GRose}\,\psi \, \alpha) \, (\mathit{GRose}\,\psi \, (\alpha
  + \onet)))}

\begin{document}

\title[Parametricity for Primitive Nested Types and
  GADTs]{Parametricity for Primitive Nested Types and GADTs}

\author[P.~Johann]{Patricia Johann\rsuper{a}}
\author[E.~Ghiorzi]{Enrico Ghiorzi\rsuper{b}}
\address{Appalachian State University}
\email{johannp@appstate.edu, enrico.ghiorzi@iit.it}
\thanks{Supported by NSF awards CCR-1906388 and CCR-1420175.}


\begin{abstract}
This paper considers parametricity and its consequent free theorems
for nested data types. Rather than representing nested types via their
Church encodings in a higher-kinded or dependently typed extension of
System F, we adopt a functional programming perspective and design a
Hindley-Milner-style calculus with primitives for constructing nested
types directly as fixpoints. Our calculus can express all nested types
appearing in the literature, including truly nested types. At the
level of terms, it supports primitive pattern matching, map functions,
and fold combinators for nested types. Our main contribution is the
construction of a parametric model for our calculus. This is both
delicate and challenging. In particular, to ensure the existence of
semantic fixpoints interpreting nested types, and thus to establish a
suitable Identity Extension Lemma for our calculus, our type system
must explicitly track functoriality of types, and cocontinuity
conditions on the functors interpreting them must be appropriately
threaded throughout the model construction. We also prove that our
model satisfies an appropriate Abstraction Theorem, as well as that it
verifies all standard consequences of parametricity in the presence of
primitive nested types. We give several concrete examples illustrating
how our model can be used to derive useful free theorems, including a
short cut fusion transformation, for programs over nested
types. Finally, we consider generalizing our results to GADTs, and
argue that no extension of our parametric model for nested types can
give a functorial interpretation of GADTs in terms of left Kan
extensions and still be parametric.
\end{abstract}

\maketitle

\section{Introduction}\label{sec:intro}

\emph{Algebraic data types} (ADTs), both built-in and user-defined,
have long been at the core of functional languages such as Haskell,
ML, Agda, Epigram, and Idris. ADTs, such as that of natural numbers,
can be unindexed. But they can also be indexed over other types. For
example, the ADT of lists (here coded in Agda)

\[\begin{array}{l}
\mathtt{data\; List \;(A : Set)\;:\;Set\;where}\\
\hspace*{0.4in}\mathtt{nil\;:\; List\;A}\\
\hspace*{0.4in}\mathtt{cons\;:\;A \rightarrow List\;A \rightarrow List\;A}
\end{array}\]

\noindent
is indexed over its element type $\mathtt{A}$.  The instance of
$\mathtt{List}$ at index $\mathtt{A}$ depends only on itself, and so is
independent of $\mathtt{List\,B}$ for any other index $\mathtt{B}$.
That is, $\mathtt{List}$, like all other ADTs, defines a \emph{family
  of inductive types}, one for each index type.

Over time, there has been a notable trend toward data types whose
non-regular indexing can capture invariants and other sophisticated
properties that can be used for program verification and other
applications.  A simple example of such a type is given by the type
\[\begin{array}{l}
\mathtt{data\; PTree\;(A : Set)\;:\;Set\;where}\\
\hspace*{0.4in}\mathtt{pleaf\;:\;A \rightarrow PTree\;A}\\
\hspace*{0.4in}\mathtt{pnode\;:\;PTree\;(A \times A) \rightarrow PTree\;A}
\end{array}\]
\noindent
of perfect trees, which can be thought of as constraining lists to
have lengths that are powers of 2.  A similar data type
$\mathtt{Nest}$ is given in the canonical paper~\cite{bm98} on nested
types. The above code makes clear that perfect trees at index type
$\mathtt{A}$ are defined in terms of perfect trees at index type
$\mathtt{A \times A}$. This is typical of nested types, one type
instance of which can depend on others, so that the entire family of
types must actually be defined at once. A nested type thus defines not
a family of inductive types, but rather an \emph{inductive family of
  types}.  Nested types include simple nested types, like perfect
trees, none of whose recursive occurrences occur below another type
constructor; ``deep'' nested types~\cite{jp20}, such as the nested
type
\[\begin{array}{l}
\mathtt{data\; PForest\;(A : Set)\;:\;Set\;where}\\
\hspace*{0.4in}\mathtt{fempty\;:\;PForest\;A}\\
\hspace*{0.4in}\mathtt{fnode\;:\; A \rightarrow PTree\;(PForest\;A) \to
PForest\;A}
\end{array}\]
\hspace{-0.04in}of perfect forests, whose recursive occurrences appear
below type constructors for other nested types; and truly nested
types, such as the nested type
\[\begin{array}{l}
\mathtt{data\; Bush\;(A : Set)\;:\;Set\;where}\\
\hspace*{0.4in}\mathtt{bnil\;:\; Bush\;A}\\
\hspace*{0.4in}\mathtt{bcons\;:\;A \rightarrow Bush\;(Bush \; A)
  \rightarrow Bush\;A}
\end{array}\]
\hspace{-0.04in}of bushes, whose recursive occurrences appear below
their own type constructors.

\emph{Relational parametricity} encodes a powerful notion of
type-uniformity, or representation independence, for data types in
polymorphic languages. It formalizes the intuition that a polymorphic
program must act uniformly on all of its possible type instantiations
by requiring that every such program preserves all relations between
pairs of types at which it is instantiated. Parametricity was
originally put forth by Reynolds~\cite{rey83} for System
F~\cite{gir72}, the calculus at the core of all polymorphic functional
languages. It was later popularized as Wadler's ``theorems for
free''~\cite{wad89}, so called because it can deduce properties of
programs in such languages solely from their types, i.e., with no
knowledge whatsoever of the text of the programs involved.  Most of
Wadler's free theorems are consequences of naturality for polymorphic
list-processing functions. However, parametricity can also derive
results that go beyond just naturality, such as correctness for ADTs
of the program optimization known as \emph{short cut
  fusion}~\cite{glp93,joh02}.

But what about nested types? Does parametricity still hold if such
types are added to polymorphic calculi? More practically, can we
justifiably reason type-independently about (functions over) nested
types in functional languages?

Type-independent reasoning about ADTs in functional languages is
usually justified by first representing ADTs by their Church
encodings, and then reasoning type-independently about these
encodings. This is typically justified by constructing a parametric
model --- i.e, a model in which polymorphic functions preserve
relations \emph{\'a la} Reynolds --- for a suitable fragment of System
F, demonstrating that an initial algebra exists for the positive type
constructor corresponding to the functor underlying an ADT of
interest, and showing that each such initial algebra is suitably
isomorphic to its corresponding Church encoding. In fact, this
isomorphism of initial algebras and their Church encodings is one of
the ``litmus tests'' for the goodness of a parametric model.

This approach works well for ADTs, which are always fixpoints of \emph{first-order} functors, and whose Church encodings, which involve
quantification over only type variables, are always expressible in
System F. For example, $\mathtt{List\,A}$ is the fixpoint of the
first-order functor $F\,X = 1 + A \times X$ and has Church encoding
$\forall \alpha. \, \alpha \to (\mathtt{A} \to \alpha \to \alpha) \to
\alpha$. But despite Cardelli's~\cite{car97} claim that ``virtually
any basic type of interest can be encoded within F$_2$'' --- i.e.,
within System F --- non-ADT nested types cannot.  Not even our
prototypical nested type of perfect trees has a Church encoding
expressible in System F\@!  Indeed, $\mathtt{PTree\,A}$ cannot be
represented as the fixpoint of any \emph{first-order} functor. However,
it can be seen as the instance at index $\mathtt{A}$ of the fixpoint
of the \emph{higher-order} functor $H\,F\,A\,=\, (A \to F\,A) \to (F
\,(A \times A) \to F\,A) \to F\,A$. It thus has Church encoding
$\forall f.\, (\forall \alpha.\,\alpha \to f\alpha) \to (\forall
\alpha. \,f (\alpha \times \alpha) \to f\alpha) \to \forall \alpha.\,
f\alpha$, which requires quantification at the higher kind $* \to *$
for $f$. A similar situation obtains for any (non-ADT) nested
type. Unfortunately, higher-kinded quantification is not available in
System F, so if we want to reason type-independently about nested
types in a language based on it we have only two options: (\emph{i})\,move to an extension of System F, such as the higher-kinded
calculus F$_\omega$ or a dependent type theory, and reason via their
Church encodings in a known parametric model for that extension, or
(\emph{ii})\, add nested types to System F as primitives --- i.e., as
primitive type-level fixpoints --- and construct a parametric model
for the result.

Since the type systems of F$_\omega$ and dependent type theories are
designed to extend System F with far more than non-ADT data types, it
seems like serious overkill to pass to their parametric models to
reason about nested types in System F. Indeed, such calculi support
fundamentally new features --- e.g., the full hierarchy of
higher-kinds or term-indexed types --- adding complexity to their
models that is entirely unnecessary for reasoning about nested types.
In addition, to our knowledge no such models represent data types as
primitive fixpoints, so they cannot be used to obtain the results of
this paper.  (See the last paragraph of this section for a more
detailed discussion in the case of F$_\omega$.) Whether or not these
results even hold for such calculi is an open question.

This paper therefore pursues the second option above.  We first design
a Hindley-Milner-style calculus supporting primitive nested types,
together with primitive types of natural transformations representing
morphisms between them. Our calculus can express all nested types
appearing in the literature, including truly nested types.  At the
term-level, it supports primitive pattern matching, map functions, and
fold combinators for nested types.\footnote{We leave incorporating
  general term-level recursion to future work because, as
  Pitts~\cite{pit00} reminds us, ``it is hard to construct models of
  both impredicative polymorphism and fixpoint recursion''. In fact,
  as the development in this paper shows, constructing a parametric
  model even for our predicative calculus with primitive nested types
  --- and even without term-level fixpoints --- is already rather
  involved. On the other hand, our calculus is strongly normalizing,
  so it perhaps edges us toward the kind of provably total practical
  programming language proposed in~\cite{wad89}.}  Our main
contribution is the construction of a parametric model for our
calculus. This is both delicate and challenging. To ensure the
existence of semantic fixpoints interpreting nested types, and thus to
establish a suitable Identity Extension Lemma, our type system must
explicitly track functoriality of types, and cocontinuity conditions
on the functors interpreting them must be appropriately threaded
throughout the model construction. Our model validates all standard
consequences of parametricity in the presence of primitive nested
types, including the isomorphism of primitive ADTs and their Church
encodings, and correctness of short cut fusion for nested types. The
relationship between naturality and parametricity has long been of
interest, and our inclusion of a primitive type of natural
transformations allows us to clearly delineate those consequences of
parametricity that follow from naturality, from those, such as short
cut fusion for nested types, that require the full power of
parametricity.

\vspace*{0.1in}

\noindent
{\bf Structure of this Paper} We introduce our calculus in
Section~\ref{sec:calculus}.  Its type system is based on the
level-2-truncation of the higher-kinded grammar from~\cite{jp19},
augmented with a primitive type of natural
transformations. (Since~\cite{jp19} contains no term calculus, the
issue of parametricity could not even be raised there.)  In
Section~\ref{sec:type-interp} we give set and relational
interpretations of our types. Set interpretations are possible
precisely because our calculus is predicative --- as ensured by our
primitive natural transformation types --- and~\cite{jp19} guarantees
that local finite presentability of $\set$ makes it suitable for
interpreting nested types.  As is standard in categorical models,
types are interpreted as functors from environments interpreting their
type variable contexts to sets or relations, as appropriate. To ensure
that these functors satisfy the cocontinuity properties needed for the
semantic fixpoints interpreting nested types to exist, set environments
must map $k$-ary type constructor variables to appropriately
cocontinuous $k$-ary functors on sets, relation environments must map
$k$-ary type constructor variables to appropriately cocontinuous
$k$-ary relation transformers, and these cocontinuity conditions must
be threaded through our type interpretations in such a way that an
Identity Extension Lemma (Theorem~\ref{thm:iel}) can be
proved. Properly propagating the cocontinuity conditions requires
considerable care, and Section~\ref{sec:iel}, where it is done, is
(apart from tracking functoriality in the calculus so that it is
actually possible) where the bulk of the work in constructing our
model lies.

In Section~\ref{sec:term-interp}, we give set and relational
interpretations for the terms of our calculus. As usual in categorical
models, terms are interpreted as natural transformations from
interpretations of their term contexts to interpretations of their
types, and these must cohere in what is essentially a fibred way.  In
Section~\ref{sec:Nat-type-terms} we prove a scheme deriving free
theorems that are consequences of naturality of polymorphic functions
over nested types. This scheme is very general, and is parameterized
over both the data type and the type of the polymorphic function at
hand. It has, for example, analogues for nested types of Wadler's
map-rearrangement free theorems as instances. In
Section~\ref{sec:abs-and-nat} we prove that our model satisfies an
Abstraction Theorem (Theorem~\ref{thm:abstraction}), which we use to
derive, in Section~\ref{sec:ftnt}, other parametricity results that go
beyond mere naturality. In Section~\ref{sec:GADTs} we show that the
results of this paper do not extend to arbitrary GADTs, and argue that
no extension of the parametric model we construct here for nested
types can give a functorial interpretation of GADTs in terms of left
Kan extensions as proposed in~\cite{jp19} and still be parametric. We
conclude and offer some directions for future work in
Section~\ref{sec:conclusion}.

This paper extends the conference paper~\cite{jgj21} by including more
exposition, proofs of all theorems, and more examples. In addition,
the discussion of parametricity for GADTs in Section~\ref{sec:GADTs}
is entirely new.

\vspace*{0.05in}

\noindent
{\bf Related Work} There is a long line of work on categorical
models of parametricity for System F;\@ see,
e.g.,~\cite{bfss90,bm05,dr04,gjfor15,has94,jac99,mr92,rr94}.  To our
knowledge, all such models treat ADTs via their Church encodings,
verifying in the just-constructed parametric model that each Church
encoding is a solution for the fixpoint equation associated with its
corresponding data type. The present paper draws on this rich
tradition of categorical models of parametricity for System F, but
modifies them to treat nested types (and therefore ADTs) as primitive
data types.

The only other extensions we know of System F with primitive data
types are those in~\cite{mat11,mg01,pit98,pit00,wad89}.
Wadler~\cite{wad89} treats full System F, and sketches parametricity
for its extension with lists. Martin and Gibbons~\cite{mg01} outline a
semantics for a grammar of primitive nested types similar to that
in~\cite{jp19}, but treat only polynomial nested types, i.e., nested
types that are fixpoints of polynomial higher-order
functors. Unfortunately, the model suggested in~\cite{mg01} is not
entirely correct (see~\cite{jp19}), and parametricity is nowhere
mentioned.  Matthes~\cite{mat11} treats System F with non-polynomial
ADTs and nested types, but his focus is on expressivity of generalized
Mendler iteration for them. He gives no semantics whatsoever.

In~\cite{pit00}, Pitts adds list ADTs to full System F with a
term-level fixpoint primitive. Other ADTs are included
in~\cite{pit98}, but nested types are not expressible in either
syntax. Pitts constructs parametric models for his calculi based on
operational, rather than categorical, semantics. A benefit of using
operational semantics to build parametric models is that it avoids
needing to work in a suitable metatheory to accommodate System F's
impredicativity. It is well-known that there are no set-based
parametric models of System F~\cite{rey84}, so parametric models for
it and its extensions are often constructed in a syntactic metatheory
such as the impredicative Calculus of Inductive Constructions (iCIC).
By adding primitive nested types to a Hindley-Milner-style calculus
and working in a categorical setting we side-step such metatheoretic
distractions.

Atkey~\cite{atk12} treats parametricity for arbitrary higher kinds,
constructing a parametric model for System F$_\omega$ within iCIC,
rather than in a semantic category. His construction is in some ways
similar to ours, but he represents (now higher-kinded) data types
using Church encodings rather than as primitives. Moreover, the
$\mathit{fmap}$ functions associated to Atkey's functors must be \emph{given}, presumably by the programmer, together with their underlying
type constructors. This absolves him of imposing cocontinuity
conditions on his model to ensure that fixpoints of his functors
exist, but, unfortunately, he does not indicate which type
constructors support $\mathit{fmap}$ functions. We suspect explicitly
spelling out which types can be interpreted as strictly positive
functors would result in a full higher-kinded extension of a calculus
akin to that presented here.

\section{The Calculus}\label{sec:calculus}

\subsection{Types}
For each $k \ge 0$, we assume countable sets $\tvars^k$ of \emph{type
  constructor variables of arity $k$} and $\fvars^k$ of
\emph{functorial variables of arity $k$}, all mutually disjoint.  The
sets of all type constructor variables and functorial variables are
$\tvars = \bigcup_{k \ge 0} \tvars^k$ and $\fvars = \bigcup_{k \ge 0}
\fvars^k$, respectively, and a \emph{type variable} is any element of
$\tvars \cup \fvars$.  We use lower case Greek letters for type
variables, writing $\phi^k$ to indicate that $\phi \in \tvars^k \cup
\fvars^k$, and omitting the arity indicator $k$ when convenient,
unimportant, or clear from context. Letters from the beginning of the
alphabet denote type variables of arity $0$, i.e., elements of
$\tvars^0 \cup \fvars^0$. We write $\overline{\phi}$ for either a set
$\{\phi_1,\dots,\phi_n\}$ of type constructor variables or a set of
functorial variables when the cardinality $n$ of the set is
unimportant or clear from context. If $V$ is a set of type variables
we write $V, \overline{\phi}$ for $V \cup \overline{\phi}$ when $V
\cap \overline{\phi} = \emptyset$.  We omit the vector notation for a
singleton set, thus writing $\phi$, instead of $\overline{\phi}$, for
$\{\phi\}$.

If $\Gamma$ is a finite subset of\, $\tvars$, $\Phi$ is a finite
subset of\, $\fvars$, $\overline{\alpha}$ is a finite subset of\,
$\fvars^0$ disjoint from $\Phi$, and $\phi^k \in \fvars^k \setminus
\Phi$, then the set $\mcF$ of well-formed types is given in
Definition~\ref{def:wftypes}. The notation there entails that an
application $F F_1 \dots F_k$ is allowed only when $F$ is a type variable
of arity $k$, or $F$ is a subexpression of the form $\mu
\phi^{k}.\lambda \alpha_1 \dots \alpha_k.F'$. Accordingly, an overbar
indicates a sequence of types whose length matches the arity of the
type applied to it.  Requiring that types are always in such
\emph{$\eta$-long normal form} avoids having to consider
$\beta$-conversion at the level of types. In a type
$\Nat^{\ol{\alpha}}F\,G$, the $\Nat$ operator binds all occurrences of
the variables in $\ol{\alpha}$ in $F$ and $G$; intuitively,
$\Nat^{\ol{\alpha}}F\,G$ represents the type of a natural
transformation in $\ol{\alpha}$ from the functor $F$ to the functor
$G$.  In a subexpression $\mu \phi^k.\lambda \ol{\alpha}.F$, the $\mu$
operator binds all occurrences of the variable $\phi$, and the
$\lambda$ operator binds all occurrences of the variables in
$\ol{\alpha}$, in the body $F$.

A \emph{type constructor context}, or \emph{non-functorial context}, is
a finite set $\Gamma$ of type constructor variables, and a \emph{functorial context} is a finite set $\Phi$ of functorial
variables. In Definition~\ref{def:wftypes}, a judgment of the form
$\Gamma;\Phi \vdash F$ indicates that the type $F$ is intended to be
functorial in the variables in $\Phi$ but not necessarily in those in
$\Gamma$.

\begin{defi}\label{def:wftypes}
The formation rules for the set $\F$ of\, \emph{(well-formed) types}
are

\begin{mathpar}
\AXC{\phantom{$\Gamma,\Phi$}}
\UIC{$\Gamma;\Phi \vdash \zerot$}
\DisplayProof
\quad
\AXC{\phantom{$\Gamma,\Phi$}}
\UIC{$\Gamma;\Phi \vdash \onet$}
\DisplayProof
\\
\AXC{$\Gamma;\ol{\alpha^0} \vdash F$}
\AXC{$\Gamma;\ol{\alpha^0}  \vdash G$}
\BIC{$\Gamma;\Phi \vdash \Nat^{\ol{\alpha^0}}F \,G$}
\DisplayProof
\\
\AXC{$\phi^k \in \Gamma \cup \Phi$}
\AXC{$\quad\quad\ol{\Gamma;\Phi \vdash F}$}
\BIC{$\Gamma;\Phi \vdash \phi^k \ol{F}$}
\DisplayProof
\\
\AXC{$\Gamma;\ol{\alpha^0},\phi^k \vdash F$}
\AXC{$\quad\quad\ol{\Gamma;\Phi \vdash G}$}
\BIC{$\Gamma;\Phi \vdash (\mu \phi^k.\lambda
  \ol{\alpha^0}. \,F)\,\ol{G}$}
\DisplayProof
\\
\AXC{$\Gamma;\Phi \vdash F$}
\AXC{$\Gamma;\Phi \vdash G$}
\BIC{$\Gamma; \Phi \vdash F + G$}
\DisplayProof
\\
\AXC{$\Gamma;\Phi \vdash F$}
\AXC{$\Gamma;\Phi \vdash G$}
\BIC{$\Gamma; \Phi \vdash F \times G$}
\DisplayProof
\end{mathpar}
\end{defi}
We write $\vdash F$ for $\emptyset;\emptyset \vdash F$.
Definition~\ref{def:wftypes} ensures that the expected weakening rules
for well-formed types hold, i.e., if $\Gamma; \Phi \vdash F$ is
well-formed, then both $\Gamma, \ol{\psi}; \Phi \vdash F$ and $\Gamma;
\Phi, \ol{\phi} \vdash F$ are also well-formed. Note that weakening
does not change the contexts in which types can be formed.  For
example, a $\mathsf{Nat}$-type --- i.e., a type of the form
$\Nat^{\ol{\alpha}} F \,G$ --- can be formed in any functorial
context. But since $\Nat$-types contain no functorial variables, we
will form them in empty functorial contexts below whenever
convenient. Of course, we cannot \emph{first} weaken the functorial
contexts of $F$ and $G$ to $\Gamma; \ol{\alpha}, \beta \vdash F$ and
$\Gamma; \ol{\alpha}, \beta \vdash G$, and \emph{then} form
$\Nat^{\ol{\alpha}} F \,G$ in the weakened context $\Gamma; \beta$.
If $\Gamma;\emptyset \vdash F$ and $\Gamma;\emptyset \vdash G$, then
our rules allow formation of the type $\Gamma;\emptyset \vdash
\Nat^\emptyset F \,G$, which represents the System F type $\Gamma
\vdash F \to G$. Similarly, if $\Gamma;\ol{\alpha} \vdash F$, then our
rules allow formation of the type $\Gamma; \emptyset \vdash
\Nat^{\ol\alpha} \,\onet \,F$, which represents the System F type
$\Gamma; \emptyset \vdash \forall \ol\alpha . F$. However, some System
F types, such as $\forall \alpha. (\alpha \to \alpha) \to \alpha$, are
not representable in our calculus. Note that, in the rule for
$\mu$-types, no variables in $\Phi$ appear in the body $F$ of $\mu
\phi. \lambda \ol\alpha.F$. This will be critical to proving the
Identity Extension Lemma for our calculus.

Definition~\ref{def:wftypes} allows the formation of all of the types
from Section~\ref{sec:intro}:
\[\begin{array}{lll}
\mathit{List}\, \alpha & = & \mu \beta. \,\onet + \alpha \times
\beta\; \mbox{ or } \; (\mu \phi. \lambda \beta.\,\onet + \beta \times \phi
\beta)\,\alpha\\
\mathit{PTree}\,\alpha & = & (\mu \phi. \lambda \beta.\,\beta +
\phi\,(\beta \times \beta))\,\alpha\\
\mathit{Forest}\,\alpha & = & (\mu \phi. \lambda \beta. \,\onet +
\beta \times \mathit{PTree}\,(\phi \beta))\,\alpha\\
\mathit{Bush}\,\alpha & = & (\mu \phi.\lambda \beta. \,\onet + \beta
\times \phi\,(\phi\beta))\,\alpha
\end{array}\]

Note that since the body $F$ of a type $(\mu \phi. \lambda
\ol\alpha. F)\ol{G}$ can only be functorial in $\phi$ and the
variables in $\ol{\alpha}$, the representation of
$\mathit{List}\,\alpha$ as the ADT $\mu \beta. \,\onet + \alpha \times
\beta$ cannot be functorial in $\alpha$. By contrast, if
$\mathit{List}\,\alpha$ is represented as the nested type $(\mu
\phi. \lambda \beta.\,\onet + \beta \times \phi \beta)\,\alpha$ then
we can choose $\alpha$ to be a functorial variable or not when forming
the type. This observation holds for other ADTs as well; for example,
if $\mathit{Tree}\,\alpha\,\gamma = \mu \beta. \alpha + \beta \times
\gamma \times \beta$, then $\alpha, \gamma; \emptyset \vdash
\mathit{Tree}\,\alpha\,\gamma$ is well-formed, but $\emptyset; \alpha,
\gamma \vdash \mathit{Tree}\,\alpha\,\gamma$ is not. And it also
applies to some non-ADT types, such as $\mathit{GRose}\,\phi\,\alpha =
\mu \beta. \onet + \alpha \times \phi\beta$, in which $\phi$ and
$\alpha$ must both be non-functorial variables.  It is in fact
possible to allow ``extra'' $0$-ary functorial variables in the body
of $\mu$-types (functorial variables of higher arity are the real
problem), which would allow the first-order representations of ADTs to
be functorial. However, since doing this requires some changes to the
formation rule for $\mu$-types, as well as the delicate threading of
some additional conditions throughout our model construction, we do
not pursue this line of investigation here.

Definition~\ref{def:wftypes} allows well-formed types to be functorial
in no variables. Functorial variables can also be demoted to
non-functorial status: if\,$F[\phi :== \psi]$ is the textual
replacement of $\phi$ in $F$, then $\Gamma, \psi^k; \Phi \vdash
F[\phi^k :== \psi^k]$ is derivable whenever $\Gamma; \Phi, \phi^k
\vdash F$ is. The proof is by induction on the structure of $F$. In
addition to textual replacement, we also have a substitution operation
on types.

If $\Gamma; \Phi \vdash F$ is a type, if $\Gamma$ and $\Phi$ contain
only type variables of arity $0$, and if $k=0$ for every occurrence of
$\phi^k$ bound by $\mu$ in $F$, then we say that $F$ is \emph{first-order}; otherwise we say that $F$ is \emph{second-order}. Substitution for first-order types is the usual
capture-avoiding textual substitution. We write $H[\alpha := F]$
for the result of substituting $F$ for $\alpha$ in $H$, and
$H[\alpha_1 := F_1, \ldots ,\alpha_k := F_k]$, or $H[\ol{\alpha := F}]$
when convenient, for $H[\alpha_1 := F_1][\alpha_2 := F_2, \ldots ,\alpha_k
  := F_k]$. The operation of \emph{second-order type substitution along
  $\ol\alpha$} is given in Definition~\ref{def:second-order-subst},
where we adopt a similar notational convention for vectors of types.
Of course, $(\cdot)[\phi^0 :=_\emptyset F]$ coincides with first-order
substitution. We omit $\ol\alpha$ when convenient, but note that it
is not correct to substitute along non-functorial variables.

\begin{defi}\label{def:second-order-subst}
If \,$\Gamma; \Phi,\phi^k \vdash H$ and $\Gamma;\Phi, \ol{\alpha}
\vdash F$ with $|\ol\alpha| = k$, then the operation $H[\phi :=_{\ol
    \alpha} F]$ of \emph{second-order type substitution along
  $\ol\alpha$} is defined by induction on $H$ as follows:
\[\begin{array}{lll}
\zerot[\phi :=_{\ol{\alpha}} F] & = & \zerot\\[0.5ex]
\onet[\phi :=_{\ol{\alpha}} F] & = & \onet\\[0.25ex]
(\Nat^{\ol\beta} G \,K)[\phi :=_{\ol{\alpha}} F]
& = & \Nat^{\ol\beta}\, G \,K\\
(\psi\ol{G})[\phi :=_{\ol{\alpha}} F] & = &
\left\{\begin{array}{ll}
\psi \,\ol{G[\phi :=_{\ol{\alpha}} F]} & \mbox{if } \psi \not = \phi\\
  F[\ol{\alpha  := G[\phi :=_{\ol{\alpha}} F]}]
  & \mbox{if } \psi = \phi
\end{array}\right.\\[2.8ex]
(G + K)[\phi :=_{\ol{\alpha}} F] & = & G[\phi
  :=_{\ol{\alpha}} F] + K[\phi :=_{\ol{\alpha}} F]\\[0.5ex]
(G \times K)[\phi :=_{\ol{\alpha}} F] & = &
G[\phi :=_{\ol{\alpha}} F] \times K[\phi
  :=_{\ol{\alpha}} F]\\[0.5ex]
((\mu \psi. \lambda \ol{\beta}. G)\ol{K})[\phi :=_{\ol{\alpha}}
  F] & = & (\mu \psi. \lambda \ol{\beta}. \,G)\, \ol{K[\phi :=_{\ol{\alpha}} F]}
\end{array}\]
\end{defi}
\noindent
The idea is that the arguments to $\phi$ get substituted for the
variables in $\ol\alpha$ in each $F$ replacing an occurrence of
$\phi$. It is not hard to see that $\Gamma;\Phi \vdash H[\phi
  :=_{\ol{\alpha}} F]$.

\subsection{Terms}\label{sec:terms}

\begin{figure*}

  \begin{adjustbox}{varwidth=5.8in, max width=5.5in, fbox, center}
       \[\begin{array}{ccc}
       \AXC{$\Gamma;\Phi \vdash F$}
       \UIC{$\Gamma;\Phi \,|\, \Delta,x :F \vdash x : F$}
       \hspace*{0.2in}\DisplayProof\hspace*{0.05in}
       &
       \AXC{$\Gamma;\Phi \,|\, \Delta \vdash t : \zerot$}
       \AXC{$\Gamma;\Phi \vdash F$}
       \BIC{$\Gamma;\Phi \,|\, \Delta \vdash \bot_F t  : F$}
       \DisplayProof\hspace*{0.05in}
       &
       \AXC{$\phantom{\Gamma;\Phi}$}
       \UIC{$\Gamma;\Phi \,|\, \Delta \vdash \top : \onet$}
       \DisplayProof\\\\
       \end{array}\]

       \vspace*{-0.15in}

       \[\begin{array}{cc}
       \AXC{$\Gamma;\Phi \,|\, \Delta \vdash s: F$}
       \UIC{$\Gamma;\Phi \,|\, \Delta \vdash \inl \,s: F + G$}
       \hspace*{0.8in}\DisplayProof\hspace*{0.05in}
       &
       \AXC{$\Gamma;\Phi \,|\, \Delta \vdash t : G$}
       \UIC{$\Gamma;\Phi \,|\, \Delta \vdash \inr \,t: F + G$}
       \DisplayProof\\\\
       \end{array}\]

       \vspace*{-0.1in}

       \[\begin{array}{c}
       \;\;\AXC{$\Gamma; \Phi \vdash F,G$}
       \AXC{$\Gamma;\Phi \,|\, \Delta \vdash t : F+G$}
       \AXC{$\Gamma;\Phi \,|\, \Delta, x : F \vdash l : K \hspace{0.2in} \Gamma;\Phi \,|\, \Delta, y : G \vdash r : K$}
       \TIC{$\Gamma;\Phi~|~\Delta \vdash \cse{t}{x \mapsto l}{y \mapsto r} : K$}
       \hspace*{-0.2in}\DisplayProof
       \end{array}\]

       \vspace*{0.05in}

       \[\begin{array}{lll}
       \AXC{$\Gamma;\Phi \,|\, \Delta \vdash s: F$}
       \AXC{$\Gamma;\Phi \,|\, \Delta \vdash t : G$}
       \BIC{$\Gamma;\Phi \,|\, \Delta \vdash (s,t) : F \times G$}
       \DisplayProof\hspace*{0.05in}
       &
       \AXC{$\Gamma;\Phi \,|\, \Delta \vdash t : F \times G$}
       \UIC{$\Gamma;\Phi \,|\, \Delta \vdash \pi_1 t : F$}
       \DisplayProof\hspace*{0.05in}
       &
       \AXC{$\Gamma;\Phi \,|\, \Delta \vdash t : F \times G$}
       \UIC{$\Gamma;\Phi \,|\, \Delta \vdash \pi_2 t : G$}
       \DisplayProof
       \end{array}\]

       \[\begin{array}{c}
       \AXC{$\Gamma; \ol{\alpha} \vdash F$}
       \AXC{$\Gamma; \ol{\alpha} \vdash G$}
       \AXC{$\Gamma; \ol{\alpha} \,|\, \Delta, x : F \vdash t: G$}
       \TIC{$\Gamma; \Phi
         \,|\, \Delta \vdash L_{\ol{\alpha}} x.t : \Nat^{\ol{\alpha}} \,F \,G$}
       \DisplayProof\vspace*{-0.05in}
       \\\\
       \AXC{$\ol{\Gamma;\Phi \vdash K}$}
       \AXC{$\Gamma; \emptyset
         \,|\, \Delta \vdash t : \Nat^{\ol{\alpha}} \,F \,G$}
       \AXC{$\Gamma;\Phi \,|\, \Delta \vdash s: F[\overline{\alpha := K}]$}
       \TIC{$\Gamma;\Phi\,|\, \Delta \vdash t_{\ol K} s:
         G[\overline{\alpha := K}]$}
       \DisplayProof\vspace*{-0.1in}
       \\\\
       \AXC{$\Gamma; \ol{\phi}, \ol{\gamma} \vdash H$}
       \AXC{$\ol{\Gamma; \ol{\beta},\ol{\gamma} \vdash F}$}
       \AXC{$\ol{\Gamma; \ol{\beta},\ol{\gamma} \vdash
           G}$}
       \TIC{$\Gamma; \Phi         ~|~\Delta 
         \vdash \map^{\ol{F},\ol{G}}_H :
         \Nat^\emptyset\;(\ol{\Nat^{\ol{\beta},\ol{\gamma}}\,F\,G})\;
         (\Nat^{\ol{\gamma}}\,H[\ol{\phi :=_{\ol{\beta}} F}]\;H[\ol{\phi
             :=_{\ol{\beta}} G}])$}
       \DisplayProof\vspace*{-0.1in}
       \\\\
       \AXC{$\Gamma; \phi, \ol{\alpha} \vdash H$}
       \UIC{$\Gamma; \Phi  \,|\, \Delta \vdash \tin_H :
         \Nat^{\ol{\beta}} H[\phi :=_{\ol{\beta}} (\mu
           \phi.\lambda \ol{\alpha}.H)\ol{\beta}][\ol{\alpha := \beta}]\,(\mu
         \phi.\lambda \ol{\alpha}.H)\ol{\beta}$}
       \hspace*{0.2in}\DisplayProof\vspace*{-0.1in}
       \\\\
       \AXC{$\Gamma; \phi,\ol{\alpha} \vdash H$}
       \AXC{$\Gamma; \ol\beta \vdash F$}
       \BIC{$\Gamma; \Phi  \,|\, \Delta \vdash \fold^F_H :
         \Nat^\emptyset\; (\Nat^{\ol{\beta}}\,H[\phi
           :=_{\ol{\beta}} F][\ol{\alpha := \beta}]\,F)\;
         (\Nat^{\ol{\beta}}\,(\mu \phi.\lambda
         \ol{\alpha}.H)\ol{\beta}\,F)$}
       \hspace*{0.2in}\DisplayProof\vspace*{-0.1in}
       \end{array}\]

       \vspace*{0.05in}

       \caption{Well-formed terms}\label{fig:terms} \vspace*{-0.00in}
\end{adjustbox}
\end{figure*}

To define our term calculus we assume an infinite set $\cal V$ of term
variables disjoint from $\tvars$ and $\fvars$. If $\Gamma$ is a type
constructor context and $\Phi$ is a functorial context, then a \emph{term context for $\Gamma$ and $\Phi$} is a finite set of bindings of
the form $x : F$, where $x \in {\cal V}$ and $\Gamma; \Phi \vdash
F$. We adopt the above conventions for denoting disjoint unions and
vectors in term contexts. If $\Delta$ is a term context for $\Gamma$
and $\Phi$, then the formation rules for the set of \emph{well-formed
  terms over $\Delta$} are as in Figure~\ref{fig:terms}. In the rule
there for $L_{\ol{\alpha}}x.t$, the $L$ operator binds all occurrences
of the type variables in $\ol{\alpha}$ in the types of $x$ and $t$, as
well as all occurrences of $x$ in $t$. In the rule for $t_{\ol K} s$
there is one type in $\ol K$ for every functorial variable in $\ol
\alpha$. In the rule for $\map^{\ol{F},\ol{G}}_H$ there is one type
$F$ and one type $G$ for each functorial variable in
$\ol\phi$. Moreover, for each $\phi^k$ in $\ol\phi$ the number of
functorial variables in $\ol\beta$ in the judgments for its
corresponding type $F$ and $G$ is $k$. In the rules for $\tin_H$ and
$\fold^F_H$, the functorial variables in $\ol{\beta}$ are fresh with
respect to $H$, and there is one $\beta$ for every
$\alpha$. Substitution for terms is the obvious extension of the usual
capture-avoiding textual substitution, and the rules of
Figure~\ref{fig:terms} ensure that weakening is respected.  Below we
form $L$-terms, which contain no free functorial variables, in empty
functorial contexts whenever convenient. We similarly form $\map$-,
$\tin$-, and $\fold$-terms, which contain no free functorial variables
or free term variables, in empty functorial contexts and empty term
contexts whenever convenient.

The ``extra'' functorial variables $\ol{\gamma}$ in the rules for
$\map^{\ol{F},\ol{G}}_H$ (i.e., those variables not affected by the
substitution of $\phi$) deserve comment. They allow us to map
polymorphic functions over nested types. Suppose, for example, we want
to map the polymorphic function $\mathit{flatten} : \Nat^\beta
(\mathit{PTree}\,\beta)\,(\mathit{List}\,\beta)$ over a list.  Even in
the absence of extra variables the instance of $\map$ required to map
each non-functorial monomorphic instantiation of $\mathit{flatten}$
over a list of perfect trees is well-formed:
\[\begin{array}{l}
\hspace*{-0.1in}\AXC{$\Gamma;\alpha \vdash \mathit{List} \, \alpha$}
\AXC{$\Gamma;\emptyset \vdash \mathit{PTree}\,F$}
\AXC{$\Gamma;\emptyset \vdash \mathit{List}\,G$}
\TIC{$\Gamma;\emptyset \,|\, \emptyset \vdash \map^{\mathit{PTree}\,F,
    \,\mathit{List}\,G}_{\mathit{List}\,\alpha} :
  \Nat^\emptyset (\Nat^\emptyset\,(\mathit{PTree}\, F)\,(\mathit{List}\, G))\,
  (\Nat^\emptyset\,(\mathit{List}\, (\mathit{PTree}\,
  F))\,(\mathit{List}\, (\mathit{List}\, G)))$} \DisplayProof
\end{array}\]
But in the absence of $\ol \gamma$, the instance
\[\Gamma;\emptyset~|~\emptyset \vdash
\map^{\mathit{PTree}\,\beta,\mathit{List}\,\beta}_{\mathit{List}\,\alpha}
: \Nat^\emptyset ( \Nat^\beta(\mathit{PTree}\,
\beta)\,(\mathit{List}\, \beta))\, (\Nat^\beta\,(\mathit{List}\,
(\mathit{PTree}\, \beta))\,(\mathit{List}\, (\mathit{List}\,
\beta)))\] required to map the \emph{polymorphic} $\mathit{flatten}$
function over a list of perfect trees is not: indeed, the functorial
contexts for $F$ and $G$ in the rule for $\map^{F,G}_H$ would have to
be empty, but because the polymorphic $\mathit{flatten}$ function is
natural in $\beta$ it cannot possibly have a type of the form
$\Nat^\emptyset F\, G$ as would be required for it to be the function
input to $\map$. Untypeability of this instance of $\map$ is
unsatisfactory in a polymorphic calculus, where we naturally expect to
be able to manipulate entire polymorphic functions rather than just
their monomorphic instances, but the ``extra'' variables $\ol \gamma$
remedy the situation, ensuring that the instance of $\map$ needed to
map the polymorphic $\mathit{flatten}$ function is typeable as
follows:

\vspace*{0.1in}

\begin{adjustbox}{varwidth=6.4in, max width=\linewidth, center}
\[\begin{array}{l}
\AXC{$\Gamma;\alpha,\gamma \vdash \mathit{List} \, \alpha$}
\AXC{$\Gamma;\gamma \vdash \mathit{PTree}\,\gamma \hspace*{0.3in}
  \Gamma;\gamma \vdash \mathit{List}\,\gamma$}
\BIC{$\Gamma;\emptyset~|~\emptyset \vdash
  \map^{\mathit{PTree}\,\gamma,\mathit{List}\,\gamma}_{\mathit{List}\,\alpha}  : \Nat^\emptyset
  (\Nat^\gamma(\mathit{PTree}\,
  \gamma)\,(\mathit{List}\, \gamma))\,
 (\Nat^\gamma\,(\mathit{List}\,
  (\mathit{PTree}\, \gamma))\,(\mathit{List}\, (\mathit{List}\,
  \gamma)))$}
\DisplayProof
  \end{array}\]
\end{adjustbox}

\vspace*{0.1in}

Our calculus is expressive enough to define a function
$\mathit{reversePTree :
  \Nat^\alpha\,(\mathit{PTree}\,\alpha)\,(\mathit{PTree}\,\alpha)}$
that reverses the order of the leaves in a perfect tree. This function
maps, e.g., the perfect tree
\[\mathit{pnode}\,(\mathit{pnode}\,(\mathit{pleaf}\,
((1, 2), (3, 4))))\]
to the perfect tree
\[\mathit{pnode}\,(\mathit{pnode}\,(\mathit{pleaf}\,
((4, 3), (2, 1))))\]
i.e., maps $((1,2),(3,4))$
to $((4,3),(2,1))$.  It can be defined as
\[\vdash (\fold_{\beta + \phi(\beta \times \beta)}^{\mathit{PTree}\,
  \alpha})_\emptyset \, s : \Nat^{\alpha}
(\mathit{PTree}\,\alpha)\,(\mathit{PTree}\,\alpha)\] where
\[\begin{array}{lll}
\fold_{\beta + \phi(\beta \times \beta)}^{\mathit{PTree}\,\alpha} & :
& \Nat^{\emptyset} (\Nat^{\alpha} (\alpha + \mathit{PTree}\,(\alpha
\times \alpha)) \; (\mathit{PTree}\,\alpha))\; (\Nat^{\alpha}
(\mathit{PTree}\,\alpha) \; (\mathit{PTree}\,\alpha))\\
\tin_{\beta + \phi(\beta \times \beta)} & : & \Nat^{\alpha} (\alpha +
\mathit{PTree}\,(\alpha \times \alpha)) \; (\mathit{PTree}\, \alpha)\\
\map_{\mathit{PTree}\,\alpha}^{\alpha \times \alpha, \alpha \times
  \alpha} & : & \Nat^\emptyset
(\Nat^{\alpha} (\alpha \times \alpha)\, (\alpha \times
\alpha))\,
(\Nat^{\alpha} (\mathit{PTree}\,(\alpha \times
\alpha))\, (\mathit{PTree}\,(\alpha \times \alpha)))
\end{array}\]
and
$\mathit{pleaf}$, $\mathit{pnode}$, $\mathit{swap}$, and $s$ are the terms
\[\begin{array}{l}
\vdash \tin_{\beta + \phi (\beta \times \beta)} \circ (
L_{\alpha} x.\, \inl\, x) \; : \; \Nat^{\alpha}\, \alpha\,
(\mathit{PTree}\,\alpha)\\
\vdash \tin_{\beta + \phi (\beta \times \beta)} \circ (
L_{\alpha} x.\, \inr\, x) \; : \; \Nat^{\alpha}\, \mathit{PTree}
(\alpha \times \alpha)\,(\mathit{PTree}\,\alpha)\\
\vdash L_{\alpha} p.\, (\pi_2 p, \pi_1 p) :
\Nat^{\alpha} (\alpha \times \alpha)\, (\alpha \times \alpha)\\
\vdash L_{\alpha} t. \,\cse{t}{b \mapsto
 \mathit{pleaf}\, b}{t' \mapsto \mathit{pnode}\,
    (((\map_{\mathit{PTree}\,\alpha}^{\alpha \times \alpha, \alpha \times
      \alpha})_\emptyset\, \mathit{swap})_\alpha\, t')}\\
\hspace*{0.5in}:  \Nat^{\alpha} (\alpha + \mathit{PTree}\,(\alpha \times
\alpha))\; \mathit{PTree}\,\alpha
\end{array}\]
respectively. Here, if $\Gamma; \emptyset \,|\, \Delta \vdash t:
\Nat^{\overline{\alpha}} F\,G$ and $\Gamma; \emptyset \,|\, \Delta
\vdash s: \Nat^{\overline{\alpha}} G\,H$ are terms then the \emph{composition} $s \circ t$ of $t$ and $s$ is defined by $s \circ t =
\Gamma; \emptyset\,|\, \Delta \vdash L_{\overline{\alpha}}
x. s_{\overline{\alpha}}(t_{\overline{\alpha}}x):
\Nat^{\overline{\alpha}} F\,H$. Our calculus can similarly define a
$\mathtt{rustle}$ function that changes the placement of the data in a
bush. This function maps, e.g., the bush
\[\begin{array}{l}
\mathit{bcons}\,0 \,(\mathit{bcons}\, (\mathit{bcons}\, 1\,
(\mathit{bcons}\, (\mathit{bcons}\, 2\,\mathit{bnil}) \,
\mathit{bnil}))\\ \hspace*{0.93in}(\mathit{bcons}\, (\mathit{bcons}\,
(\mathit{bcons}\, 3 \, (\mathit{bcons}\, (\mathit{bcons}\, 4\,
\mathit{bnil})\,  \mathit{bnil}))\, \mathit{bnil})\, \mathit{bnil}))
\end{array}\]
to the bush
\[\begin{array}{c}
\hspace*{-1.64in}\mathit{bcons}\,4\, (\mathit{bcons}\, (\mathit{bcons}\, 0
\,(\mathit{bcons}\, (\mathit{bcons}\, 3 \, \mathit{bnil}) \,
\mathit{bnil}))\\
\hspace*{0.93in}(\mathit{bcons}\, (\mathit{bcons}\, (\mathit{bcons}\, 2
\,(\mathit{bcons}\, (\mathit{bcons}\, 1 \, \mathit{bnil}) \,
\mathit{bnil})) \, \mathit{bnil})\, \mathit{bnil}))
\end{array}\]
It can be defined as
\[\vdash (\fold^{\mathit{Bush}\,\alpha}_{\onet + \beta \times \phi (\phi \beta)})_\emptyset\, balg
: \Nat^\alpha (\mathit{Bush}\,\alpha)\,(\mathit{Bush}\,\alpha)\]
where
\[\begin{array}{lll}
\vdash \fold_{\onet + \beta \times \phi
  (\phi\beta)}^{\mathit{Bush}\,\alpha} : \Nat^{\emptyset}\, (\Nat^{\alpha}\,
(\onet + \alpha \times \mathit{Bush}\, (\mathit{Bush}\, \alpha))) \;
(\mathit{Bush}\,\alpha))\; (\Nat^{\alpha} \,(\mathit{Bush}\,\alpha) \;
(\mathit{Bush}\,\alpha))\\
\end{array}\]
and $\mathit{bnil}$, $\mathit{bcons}$, $\tin^{-1}_{\onet + \beta
  \times \phi (\phi\beta)}$, $\mathit{balg}$, and $\mathit{consalg}$
are the terms
\[\begin{array}{l}
\vdash \tin_{\onet + \beta \times \phi (\phi\beta)} \circ (
L_{\alpha}\, x.\, \inl\, x) \; : \; \Nat^{\alpha}\, \onet\;
(\mathit{Bush}\,\alpha)\\
\vdash \tin_{\onet + \beta \times \phi (\phi\beta)} \circ (
L_{\alpha}\, x.\, \inr\, x ) \; : \; \Nat^{\alpha}\, (\alpha \times
\mathit{Bush}\,(\mathit{Bush}\,\alpha))\; (\mathit{Bush}\,\alpha)\\[0.05in]
\vdash (\fold_{\onet + \beta \times \phi (\phi\beta)}^{(\onet + \beta
  \times \phi (\phi\beta))[\phi := \mathit{Bush}\,\alpha]})_\emptyset\,
((\map_{\onet + \beta \times \phi (\phi\beta)}^{(\onet + \beta \times
  \phi (\phi\beta))[\phi := \mathit{Bush}\,\alpha][\beta := \alpha],
  \mathit{Bush}\,\alpha})_\emptyset\, \tin_{\onet + \beta \times \phi
  (\phi\beta)})\\[0.05in]
\hspace*{0.2in} : \Nat^{\alpha}\, (\mathit{Bush}\,\alpha)\; (\onet +
\alpha \times \mathit{Bush}\,(\mathit{Bush}\,\alpha))\\
  \vdash L_{\alpha} \, s.
  \cse{s}{* \mapsto \mathit{bnil}_\alpha *}{(a, bba) \mapsto consalg_\alpha (a, bba)}
  : \Nat^\alpha \, (\onet + \alpha \times \mathit{Bush} (\mathit{Bush} \,\alpha)) \, (\mathit{Bush}\,\alpha) \\
  \vdash L_{\alpha} \, (a, bba).
  \mathsf{case}\, (\ininv{\mathit{Bush}\,\alpha}{bba})\, \mathsf{of}\,\{
  \\
  \hspace{8em}\, * \mapsto \mathit{bcons}_{\alpha} (a,\mathit{bnil}_{\mathit{Bush}\,\alpha} *);  \\
  \hspace{8em} (ba, bbba) \mapsto \mathsf{case}\, (\ininv{\alpha}{ba})
  \, \mathsf{of}\,\{  \\
  \hspace{16em}* \mapsto \mathit{bcons}_\alpha (a, \mathit{bcons}_{\mathit{Bush}\,\alpha}
        (\mathit{bnil}_\alpha *, bbba)) ;  \\
  \hspace{16em}(a', bba') \mapsto \mathit{bcons}_\alpha (a', \mathit{bcons}_{\mathit{Bush}\,\alpha}
  (\mathit{bcons}_\alpha (a, bba'), bbba)) \} \! \}\\
\hspace*{0.2in} : \Nat^\alpha \, (\alpha \times \mathit{Bush} (\mathit{Bush} \,\alpha)) \, (\mathit{Bush}\,\alpha) \\
\end{array}\]
respectively.

Unfortunately, our calculus cannot express types of recursive
functions --- such as a concatenation function for perfect trees or a
zip function for bushes --- that take as inputs a nested type and an
argument of another type, both of which are parameterized over the
same variable. The fundamental issue is that recursion is expressible
only via $\fold$, which produces natural transformations in some
variables $\ol\alpha$ from $\mu$-types to other functors $F$. The
restrictions on $\Nat$-types entail that $F$ cannot itself be a
$\Nat$-type containing $\ol{\alpha}$, so, e.g., $\Nat^\alpha
(\mathit{PTree}\,\alpha) \,(\Nat^\emptyset (\mathit{PTree}\,\alpha)\,
(\mathit{PTree}\,(\alpha \times \alpha)))$ is not well-formed.
Uncurrying gives $\Nat^\alpha (\mathit{PTree}\,\alpha \times
\mathit{PTree}\,\alpha)\, (\mathit{PTree}\,(\alpha \times \alpha))$,
which is well-formed, but $\mathsf{fold}$ cannot produce a term of this
type because $\mathit{PTree}\,\alpha \times \mathit{PTree}\,\alpha$ is
not a $\mu$-type. Our calculus can, however, express types of
recursive functions that take multiple nested types as arguments,
provided they are parameterized over disjoint sets of type variables
and the return type of the function is parameterized over only the
variables occurring in the type of its final argument. Even for ADTs
there is a difference between which folds over them we can type when
they are viewed as ADTs (i.e., as fixpoints of first-order functors)
versus as proper nested types (i.e., as fixpoints of higher-order
functors). This is because, in the return type of $\mathsf{fold}$, the
arguments of the $\mu$-type must be variables bound by $\Nat$.  For
ADTs, the $\mu$-type takes no arguments, making it possible to write
recursive functions, such as a concatenation function for lists of
type $\alpha ; \emptyset \vdash \Nat^\emptyset \ (\mu \beta. \onet +
\alpha \times \beta)\, (\Nat^\emptyset (\mu \beta. \onet + \alpha
\times \beta) \, (\mu \beta. \onet + \alpha \times \beta))$.  This is
not possible for nested types --- even when they are semantically
equivalent to ADTs.

Interestingly, even some recursive functions of a single proper nested
type --- e.g., a reverse function for bushes that is a true involution
--- cannot be expressed as folds because the algebra arguments needed
to define them are again recursive functions with types of the same
problematic form as the type of, e.g., a zip function for perfect
trees.  Expressivity of folds for nested types has long been a vexing
issue, and this is naturally inherited by our calculus. Adding more
expressive recursion combinators could help, but since this is
orthogonal to the issue of parametricity in the presence of primitive
nested types we do not consider it further here.

\section{Interpreting Types}\label{sec:type-interp}

We denote the category of sets and functions by $\set$. The category
$\rel$ has as its objects triples $(A,B,R)$ where $R$ is a relation
between the objects $A$ and $B$ in $\set$, i.e., a subset of $A \times
B$, and has as its morphisms from $(A,B,R)$ to $(A',B',R')$ pairs $(f
: A \to A',g : B \to B')$ of morphisms in $\set$ such that $(f a,g\,b)
\in R'$ whenever $(a,b) \in R$. We write $R : \rel(A,B)$ in place of
$(A,B,R)$ when convenient.  If $R : \rel(A,B)$ we write $\pi_1 R$ and
$\pi_2 R$ for the \emph{domain} $A$ of $R$ and the \emph{codomain} $B$
of $R$, respectively.  If $A : \set$, then we write $\Eq_A =
(A,A,\{(x,x)~|~ x \in A\})$ for the \emph{equality relation} on $A$.

The key idea underlying Reynolds' parametricity is to give each type
$F(\alpha)$ with one free variable $\alpha$ both an \emph{object
  interpretation} $F_0$ taking sets to sets and a \emph{relational
  interpretation} $F_1$ taking relations $R : \rel(A,B)$ to relations
$F_1 (R) : \rel(F_0 (A), F_0 (B))$, and to interpret each term
$t(\alpha,x) : F(\alpha)$ with one free term variable $x : G(\alpha)$
as a map $t_0$ associating to each set $A$ a function $t_0(A) : G_0(A)
\to F_0(A)$. These interpretations are to be given inductively on the
structures of $F$ and $t$ in such a way that they imply two
fundamental theorems. The first is an \emph{Identity Extension Lemma},
which states that $F_1(\Eq_A) = \Eq_{F_0(A)}$, and is the essential
property that makes a model relationally parametric rather than just
induced by a logical relation. The second is an \emph{Abstraction
  Theorem}, which states that, for any $R :\rel(A, B)$,
$(t_0(A),t_0(B))$ is a morphism in $\rel$ from
$(G_0(A),G_0(B),G_1(R))$ to $(F_0(A),F_0(B),F_1(R))$. The Identity
Extension Lemma is similar to the Abstraction Theorem except that it
holds for \emph{all} elements of a type's interpretation, not just
those that are interpretations of terms. Similar theorems are
expected to hold for types and terms with any number of free
variables.

The key to proving the Identity Extension Lemma in our setting
(Theorem~\ref{thm:iel}) is a familiar ``cutting down'' of the
interpretations of universally quantified types to include only the
``parametric'' elements; the relevant types in our calculus are the
$\Nat$-types.  This cutting down requires, as usual, that the set
interpretations of types (Section~\ref{sec:set-interp}) are defined
simultaneously with their relational interpretations
(Section~\ref{sec:rel-interp}). While the set interpretations are
relatively straightforward, their relation interpretations are less
so, mainly because of the cocontinuity conditions required to ensure
that they are well-defined. We develop these conditions in
Sections~\ref{sec:set-interp} and~\ref{sec:rel-interp}. This separates
our set and relational interpretations in space, but otherwise has no
impact on the fact that they are given by mutual induction.

\subsection{Interpreting Types as Sets}\label{sec:set-interp}

We interpret types in our calculus as $\omega$-cocontinuous functors
on locally finitely presentable categories~\cite{ar94}. Both $\set$
and $\rel$ are locally finitely presentable categories. Since functor
categories of locally finitely presentable categories are again
locally finitely presentable, the fixpoints interpreting $\mu$-types
in $\set$ and $\rel$ must all exist, and thus the set and relational
interpretations of all of the types in Definition~\ref{def:wftypes},
are well-defined~\cite{jp19}. To bootstrap this process, we interpret
type variables as $\omega$-cocontinuous functors in
Definitions~\ref{def:set-env} and~\ref{def:reln-env}. If $\C$ and $\D$
are locally finitely presentable categories, we write $[\C,\D]$ for
the category of $\omega$-cocontinuous functors from $\C$ to $\D$.

\begin{defi}\label{def:set-env}
A \emph{set environment} maps each type variable in $\tvars^k \cup
\fvars^k$ to an element of $[\set^k,\set]$.  A morphism $f : \rho \to
\rho'$ for set environments $\rho$ and $\rho'$ with $\rho|_\tvars =
\rho'|_\tvars$ maps each type constructor variable $\psi^k \in \tvars$
to the identity natural transformation on $\rho \psi^k = \rho'\psi^k$
and each functorial variable $\phi^k \in \fvars$ to a natural
transformation from the $k$-ary functor $\rho \phi^k$ on $\set$ to the
$k$-ary functor $\rho' \phi^k$ on $\set$.  Composition of morphisms on
set environments is given componentwise, with the identity morphism
mapping each set environment to itself. This gives a category of set
environments and morphisms between them, denoted $\setenv$.
\end{defi}
When convenient we identify a functor in $[\set^0, \set]$ with its
value on $\ast$ and consider a set environment to map a type variable
of arity $0$ to a set.  If $\ol{\alpha} = \{\alpha_1,\dots,\alpha_k\}$
and $\ol{A} = \{A_1,\dots,A_k\}$, then we write $\rho[\ol{\alpha := A}]$
for the set environment $\rho'$ such that $\rho' \alpha_i = A_i$ for
$i = 1,\dots,k$ and $\rho' \alpha = \rho \alpha$ if $\alpha \not \in
\{\alpha_1,\dots,\alpha_k\}$. We note that $\lambda \ol{A}.\,
\rho[\ol{\alpha  := A}]$ is a functor in $\ol{A}$.

We can now define our set interpretation. Its action on objects of
$\setenv$ is given in Definition~\ref{def:set-sem}, and its action on
morphisms of $\setenv$ is given in Definition~\ref{def:set-sem-funcs}.
If $\rho$ is a set environment we write $\Eq_\rho$ for the equality
relation environment such that $\Eq_\rho \phi = \Eq_{\rho \phi}$ for
every type variable $\phi$; see Definitions~\ref{def:rel-transf}
and~\ref{def:reln-env} for the definitions of a relation transformer
and a relation environment, and Equation~\ref{def:eq-transf} for the
definition of the relation transformer $\Eq_F$ on a functor $F$.
Equality relation environments appear in the third clause of
Definition~\ref{def:set-sem}.  The relational interpretation also
appearing in the third clause of Definition~\ref{def:set-sem} is
given in Definition~\ref{def:rel-sem}.

\begin{defi}\label{def:set-sem}
The \emph{set interpretation} $\setsem{\cdot} : \F \to [\setenv, \set]$
is defined by:
\begin{align*}
  \setsem{\Gamma;\Phi \vdash \zerot}\rho &= 0\\
  \setsem{\Gamma;\Phi \vdash \onet}\rho &= 1\\
  \setsem{\Gamma; \emptyset
    \vdash \Nat^{\ol{\alpha}}
    \,F\,G}\rho &= \{\eta : \lambda \ol{A}. \,\setsem{\Gamma;
    \ol{\alpha} \vdash
    F}\rho[\ol{\alpha := A}]
      \Rightarrow \lambda \ol{A}.\,\setsem{\Gamma;
        \ol{\alpha} \vdash G}\rho[\ol{\alpha := A}] \\
      &\hspace{0.3in}|~\forall \overline{A}, \overline{B} :
      \set. \forall \overline{R : \rel(A, B)}.\\
      &\hspace{0.4in}(\eta_{\overline{A}}, \eta_{\overline{B}})
      : \relsem{\Gamma; \ol{\alpha} \vdash F}\Eq_{\rho}[\ol{\alpha := R}]
      \rightarrow \relsem{\Gamma; \ol{\alpha} \vdash
        G}\Eq_{\rho}[\ol{\alpha := R}] \} \\
  \setsem{\Gamma;\Phi \vdash \phi\ol{F}}\rho &=
  (\rho\phi)\,\ol{\setsem{\Gamma;\Phi \vdash
    F}\rho}\\
  \setsem{\Gamma;\Phi \vdash F+G}\rho &=
  \setsem{\Gamma;\Phi \vdash F}\rho +
  \setsem{\Gamma;\Phi \vdash G}\rho\\
  \setsem{\Gamma;\Phi \vdash F\times G}\rho &=
  \setsem{\Gamma;\Phi \vdash F}\rho \times
  \setsem{\Gamma;\Phi \vdash G}\rho\\
  \setsem{\Gamma;\Phi \vdash (\mu \phi.\lambda
    \ol{\alpha}. H)\ol{G}}\rho &= (\mu
    T^\set_{H,\rho})\ol{\setsem{\Gamma;\Phi \vdash G}\rho}\\
    \text{where } T^\set_{H,\rho}\,F & = \lambda
  \ol{A}. \setsem{\Gamma;\phi, \ol{\alpha} \vdash
    H}\rho[\phi :=  F][\ol{\alpha := A}]\\
  \text{and } T^\set_{H,\rho}\,\eta &= \lambda
  \ol{A}. \setsem{\Gamma;\phi, \ol{\alpha} \vdash
    H}\id_\rho[\phi := \eta][\ol{\alpha := \id_{A}}]
\end{align*}
\end{defi}
If $\rho \in \setenv$ and $\vdash F$ then we write $\setsem{\vdash F}$
instead of $\setsem{\vdash F}\rho$ since the environment is
immaterial. The third clause of Definition~\ref{def:set-sem} does
indeed define a set: local finite presentability of $\set$ and
$\omega$-cocontinuity of $\setsem{\Gamma;\ol{\alpha} \vdash F}$
ensure that $\{\eta : \lambda \ol{A}.\,\setsem{\Gamma;\ol{\alpha}
  \vdash F}\rho[\ol{\alpha := A}]$ $\Rightarrow \lambda
\ol{A}.\,\setsem{\Gamma;\ol{\alpha} \vdash G}\rho[\ol{\alpha := A}]\}$
(containing $\setsem{\Gamma;\emptyset \vdash
  \Nat^{\ol{\alpha}}\,F\,G}\rho$) can be embedded in the set
\[ \prod_{\substack{\ol{S} =
    (S_1,\dots,S_{|\ol{\alpha}|})\\ S_1,\dots,S_{|\ol{\alpha}|} \mbox{ are
      finite cardinals}}} \hspace*{-0.4in}(\setsem{\Gamma;\ol{\alpha}
  \vdash G}\rho[\ol{\alpha := S}])^{(\setsem{\Gamma;\ol{\alpha} \vdash
    F}\rho[\ol{\alpha := S}])}\] Also, $\setsem{\Gamma; \emptyset
  \vdash \Nat^{\ol\alpha} F\,G}$ is $\omega$-cocontinuous since it is
constant (in particular, on $\omega$-directed sets). This is because
Definition~\ref{def:set-env} ensures that restrictions to
$\mathbb{T}^k$ of morphisms between set environments are identities.
Interpretations of $\Nat$-types ensure that the interpretations
$\setsem{\Gamma \vdash F \to G}$ and $\setsem{\Gamma \vdash \forall
  \ol\alpha. F}$ of the System F arrow types and $\forall$-types
representable in our calculus are as expected in any parametric model.

To make sense of the last clause in Definition~\ref{def:set-sem}, we
need to know that, for each $\rho \in \setenv$, $T^\set_{H,\rho}$ is
an $\omega$-cocontinuous endofunctor on $[\set^k, \set]$, and thus
admits a fixpoint.  Since $T_{H,\rho}^\set$ is defined in terms of
$\setsem{\Gamma;\phi, \ol{\alpha} \vdash H}$, this means that
interpretations of types must be such functors, which in turn means
that the actions of set interpretations of types on objects and on
morphisms in $\setenv$ are intertwined. Fortunately, we know
from~\cite{jp19} that, for every $\Gamma; \ol{\alpha} \vdash G$,
$\setsem{\Gamma; \ol{\alpha} \vdash G}$ is actually in $[\set^k,\set]$
where $k = |\ol \alpha|$. This means that for each $\setsem{\Gamma;
  \phi^k, \ol{\alpha} \vdash H}$, the corresponding operator
$T^\set_{H}$ can be extended to a \emph{functor} from $\setenv$ to
$[[\set^k,\set],[\set^k,\set]]$. The action of $T^\set_H$ on an object
$\rho \in \setenv$ is given by the higher-order functor
$T_{H,\rho}^\set$, whose actions on objects (functors in $[\set^k,
  \set]$) and morphisms (natural transformations) between them are
given in Definition~\ref{def:set-sem}. Its action on a morphism $f :
\rho \to \rho'$ is the higher-order natural transformation
$T^\set_{H,f} : T^\set_{H,\rho} \to T^\set_{H,\rho'}$ whose action on
$F : [\set^k,\set]$ is the natural transformation $T^\set_{H,f}\, F :
T^\set_{H,\rho}\,F \to T^\set_{H,\rho'}\,F$ whose component at
$\ol{A}$ is $(T^\set_{H,f}\, F)_{\ol{A}} = \setsem{\Gamma;
  \phi,\ol{\alpha} \vdash H}f[\phi := \id_F][\ol{\alpha := \id_A}]$.
Note that this action of $T_{H,\rho}^\set$ on morphisms appears in the
final clause of Definition~\ref{def:set-sem-funcs} below. It will
indeed be well-defined there because the functorial action of
$\setsem{\Gamma; \phi,\ol{\alpha} \vdash H}$ will already be available
from the induction hypothesis on $H$.

Using $T^\set_H$, we can define the functorial action of set
interpretation.
\begin{defi}\label{def:set-sem-funcs}
Let $f: \rho \to \rho'$ be a morphism between set environments $\rho$
and $\rho'$ (so that $\rho|_\tvars = \rho'|_\tvars$). The action
$\setsem{\Gamma;\Phi \vdash F}f$ of\, $\setsem{\Gamma;\Phi \vdash F}$
on $f$ is given by:
\begin{itemize}
\item If \,$\Gamma;\Phi \vdash \zerot$ then $\setsem{\Gamma;\Phi \vdash
  \zerot}f = \id_0$
\item If \,$\Gamma;\Phi \vdash \onet$ then $\setsem{\Gamma;\Phi \vdash
  \onet}f = \id_1$
\item If \,$\Gamma; \emptyset
  \vdash \Nat^{\ol{\alpha}}\,F\,G$ then
  $\setsem{\Gamma; \emptyset
    \vdash \Nat^{\ol{\alpha}}\,F\,G} f =
  \id_{\setsem{\Gamma; \emptyset
      \vdash \Nat^{\ol{\alpha}}\,F\,G}\rho}$
\item If \,$\Gamma;\Phi \vdash \phi \ol{F}$ then
\[\setsem{\Gamma;\Phi \vdash \phi \ol{F}} f : \setsem{\Gamma;\Phi
  \vdash \phi \ol{F}}\rho \to \setsem{\Gamma;\Phi \vdash
  \phi\ol{F}}\rho' = (\rho\phi) \ol{\setsem{\Gamma;\Phi \vdash
    F}\rho} \to (\rho'\phi) \ol{\setsem{\Gamma;\Phi \vdash
    F}\rho'}\] is defined by \[\setsem{\Gamma;\Phi \vdash \phi
  \ol{F}} f = (f\phi)_{\ol{\setsem{\Gamma;\Phi \vdash
      F}\rho'}}\, \circ\, (\rho\phi) {\ol{\setsem{\Gamma;\Phi
      \vdash F}f}} = (\rho'\phi) {\ol{\setsem{\Gamma;\Phi \vdash
      F}f}}\, \circ\, (f \phi)_{\ol{\setsem{\Gamma;\Phi \vdash
      F}\rho}}\]  The latter equality holds because $\rho\phi$ and
  $\rho'\phi$ are functors and $f\phi : \rho\phi \to \rho'\phi$ is a
  natural transformation, so the following naturality square commutes:
{\footnotesize\begin{equation}\label{eq:cd2}
\begin{CD}
  (\rho\phi) \ol{\setsem{\Gamma;\Phi \vdash F}\rho} @> (f\phi)_{
    \ol{\setsem{\Gamma;\Phi \vdash F}\rho}} >> (\rho'\phi)
  \ol{\setsem{\Gamma;\Phi \vdash F}\rho} \\ @V(\rho\phi)
  \ol{\setsem{\Gamma;\Phi \vdash F}f}VV @V (\rho'\phi)
  \ol{\setsem{\Gamma;\Phi \vdash F}f} VV \\ (\rho\phi)
  \ol{\setsem{\Gamma;\Phi \vdash F}\rho'} @>(f\phi)_{
    \ol{\setsem{\Gamma;\Phi \vdash F}\rho'}}>> (\rho'\phi)
  \ol{\setsem{\Gamma;\Phi \vdash F}\rho'}
\end{CD}
\end{equation}}
\item If\, $\Gamma;\Phi \vdash F + G$ then $\setsem{\Gamma;\Phi
  \vdash F + G}f$ is defined by \[\setsem{\Gamma;\Phi \vdash
  F + G}f(\inl\,x) = \inl\,(\setsem{\Gamma;\Phi \vdash
  F}f x)\] and $\setsem{\Gamma;\Phi \vdash F +
  G}f(\inr\,y) = \inr\,(\setsem{\Gamma;\Phi \vdash G}f y)$
\item If \,$\Gamma;\Phi\vdash F \times G$ then
  $\setsem{\Gamma;\Phi \vdash F \times G}f =
  \setsem{\Gamma;\Phi \vdash F}f \times \setsem{\Gamma;\Phi \vdash
    G}f$
\item If \,$\Gamma;\Phi \vdash (\mu \phi.\lambda
  \ol{\alpha}. H)\ol{G}$ then
  \[\begin{array}{lll}
  \setsem{\Gamma;\Phi \vdash (\mu  \phi.\lambda
    \ol{\alpha}. H)\ol{G}} f &: &\setsem{\Gamma;\Phi
    \vdash (\mu \phi.\lambda \ol{\alpha}. H)\ol{G}} \rho \to
  \setsem{\Gamma;\Phi \vdash (\mu
    \phi.\lambda\ol{\alpha}. H)\ol{G}} \rho'\\
  &= &(\mu
  T^\set_{H,\rho})\ol{\setsem{\Gamma;\Phi \vdash G}\rho} \to (\mu
  T^\set_{H,\rho'})\ol{\setsem{\Gamma;\Phi \vdash G}\rho'}
  \end{array}\]
  is
  defined by
  \[\begin{array}{ll} & (\mu T^\set_{H,f})_{\ol{\setsem{\Gamma;\Phi \vdash
      G}\rho'}} \circ (\mu T^\set_{H,\rho})\ol{\setsem{\Gamma;\Phi
      \vdash G}f}\\ = & (\mu T^\set_{H,\rho'})\ol{\setsem{\Gamma;\Phi
      \vdash G}f} \circ (\mu T^\set_{H,f})_{\ol{\setsem{\Gamma;\Phi
      \vdash G}\rho}}\end{array}\]  The latter equality holds because $\mu
  T^\set_{H,\rho}$ and $\mu T^\set_{H,\rho'}$ are functors and $\mu
  T_{H,f}^\set : \mu T_{H,\rho}^\set \to \mu T_{H,\rho'}^\set$ is a
  natural transformation, so the following naturality square commutes:
{\footnotesize\begin{equation}\label{eq:cd3}
\begin{CD}
 (\mu T^\set_{H,\rho}) \ol{\setsem{\Gamma;\Phi \vdash G}\rho} @> (\mu
  T^\set_{H,f})_{\ol{\setsem{\Gamma;\Phi \vdash G}\rho}} >> (\mu
  T^\set_{H,\rho'}) \ol{\setsem{\Gamma;\Phi \vdash G}\rho} \\
 @V(\mu T^\set_{H,\rho}) \ol{\setsem{\Gamma;\Phi \vdash G}f}VV @V  (\mu
 T^\set_{H,\rho'}) \ol{\setsem{\Gamma;\Phi \vdash G}f} VV \\
(\mu T^\set_{H,\rho}) \ol{\setsem{\Gamma;\Phi \vdash G}\rho'} @>(\mu
 T^\set_{H,f})_{\ol{\setsem{\Gamma;\Phi \vdash G}\rho'}}>> (\mu
 T^\set_{H,\rho'}) \ol{\setsem{\Gamma;\Phi \vdash G}\rho'}
\end{CD}
\end{equation}}
\end{itemize}
\end{defi}

\noindent
Definitions~\ref{def:set-sem} and~\ref{def:set-sem-funcs} respect
weakening, i.e., ensure that a type and its weakenings have the same
set interpretations.

\subsection{Interpreting Types as Relations}\label{sec:rel-interp}

\begin{defi}\label{def:rel-transf}
A \emph{$k$-ary relation transformer} $F$ is a triple $(F^1, F^2,F^*)$,
where
\begin{itemize}
\item $F^1,F^2 : [\set^k,\set]$ and $F^* : [\rel^k, \rel]$ are
  functors
\item If $R_1:\rel(A_1,B_1),\dots,R_k:\rel(A_k,B_k)$, then $F^* \ol{R} :
  \rel(F^1 \ol{A}, F^2 \ol{B})$
\item If $(\alpha_1, \beta_1) \in \Homrel(R_1,S_1),\dots, (\alpha_k,
  \beta_k) \in \Homrel(R_k,S_k)$ then $F^* \ol{(\alpha, \beta)} = (F^1
  \ol{\alpha}, F^2 \ol{\beta})$
\end{itemize}
We define $F\ol{R}$ to be $F^*\overline{R}$ and
$F\overline{(\alpha,\beta)}$ to be $F^*\overline{(\alpha,\beta)}$.
\end{defi}
The last clause of Definition~\ref{def:rel-transf} expands to: if
$\ol{(a,b) \in R}$ implies $\ol{(\alpha\,a,\beta\,b) \in S}$ then
$(c,d) \in F^*\ol{R}$ implies $(F^1 \ol{\alpha}\,c,F^2 \ol{\beta}\,d)
\in F^*\ol{S}$. When convenient we identify a $0$-ary relation
transformer $(A,B,R)$ with $R : \rel(A,B)$, and write $\pi_1 F$ for
$F^1$ and $\pi_2 F$ for $F^2$. Below we extend these conventions to
relation environments in the obvious ways.

\begin{defi}
The category $RT_k$ of $k$-ary relation transformers is given by the
following data:
\begin{itemize}
\item An object of $RT_k$ is a relation transformer
\item A morphism $\delta : (F^1,F^2,F^*) \to (G^1,G^2,G^*)$ in $RT_k$
  is a pair of natural transformations $(\delta^1, \delta^2)$, where
  $\delta^1 : F^1 \to G^1$ and $\delta^2 : F^2 \to G^2$, such that,
  for all $\ol{R : \rel(A, B)}$, if $(x, y) \in F^*\ol{R}$ then
  $(\delta^1_{\ol{A}}x, \delta^2_{\ol{B}}y) \in G^*\ol{R}$
\item Identity morphisms and composition are inherited from the
  category of functors on $\set$
\end{itemize}
\end{defi}

\begin{defi}\label{def:RT-functor}
A \emph{higher-order relation transformer} $H$ on $RT_k$ is a triple $H
= (H^1,H^2,H^*)$, where
\begin{itemize}
\item $H^1$ and $H^2$ are functors from $[\set^k,\set]$ to $[\set^k,\set]$
\item $H^*$ is a functor from $RT_k$ to $[\rel^k,\rel]$
\item $H^*(F^1,F^2,F^*) \ol{R} : \rel(H^1F^1 \ol{A}, H^2F^2 \ol{B})$
  whenever $R_1:\rel(A_1,B_1),\dots,R_k:\rel(A_k,B_k)$
\item $H^*(F^1,F^2,F^*)\, \ol{(\alpha, \beta)} = (H^1F^1\ol{\alpha},
  H^2F^2 \ol{\beta})$ whenever $(\alpha_1, \beta_1) \in
  \Homrel(R_1,S_1),\dots, (\alpha_k, \beta_k)$ $\in \Homrel(R_k,S_k)$
\item For all $\overline{R : \rel(A,B)}$,
  $\pi_1((H^*(\delta^1,\delta^2))_{\overline{R}}) = (H^1
  \delta^1)_{\overline{A}}$ and
  $\pi_2((H^*(\delta^1,\delta^2))_{\overline{R}}) = (H^2
  \delta^2)_{\overline{B}}$
\end{itemize}
\end{defi}
\noindent
Definition~\ref{def:RT-functor} entails that every higher-order
relation transformer $H$ is an endofunctor on $RT_k$, where
\begin{itemize}
\item The action of $H$ on objects is given by $H\,(F^1,F^2,F^*) =
  (H^1F^1,\,H^2F^2,\,H^*(F^1,F^2,F^*))$
\item The action of $H$ on morphisms is given by
  $H\,(\delta^1,\delta^2) = (H^1\delta^1,H^2\delta^2)$ for
  $(\delta^1,\delta^2) : (F^1,F^2,F^*)\to (G^1,G^2,G^*)$
\end{itemize}
Note that the last condition of Definition~\ref{def:RT-functor}
entails that if $(\delta^1,\delta^2) : (F^1,F^2,F^*)\to (G^1,G^2,G^*)$
in $RT_k$, if $R_1:\rel(A_1,B_1),\dots,R_k:\rel(A_k,B_k)$, and if $(x,
y) \in H^*(F^1,F^2,F^*)\ol{R}$, then
\[((H^1\delta^1)_{\ol{A}}x, (H^2\delta^2)_{\ol{B}}y) \in
H^*(G^1,G^2,G^*)\ol{R}\]

\begin{defi}\label{def:RT-nat-trans}
A \emph{morphism} $\sigma : H \to K$ between higher-order relation
transformers $H$ and $K$ on $RT_k$ is a pair $\sigma = (\sigma^1,
\sigma^2)$, where $\sigma^1 : H^1 \to K^1$ and $\sigma^2 : H^2 \to
K^2$ are natural transformations between endofunctors on
$[\set^k,\set]$ such that $((\sigma^1_{F^1})_{\overline{A}},
(\sigma^2_{F^2})_{\overline{B}})$ is a morphism from $H^*
F\overline{R}$ to $K^* F\overline{R}$ in $\rel$ for any $F \in RT_k$
and any $k$-tuple of relations $\overline{R : \rel(A, B)}$.
\end{defi}

Definition~\ref{def:RT-nat-trans} entails that a morphism $\sigma$
between higher-order relation transformers is a natural transformation
between endofunctors on $RT_k$ whose component at $F = (F^1,F^2,F^*)
\in RT_k$ is given by $\sigma_F = (\sigma^1_{F^1}, \sigma^2_{F^2})$.
Moreover, $\sigma^i_{F^i}$ is natural in $F^i : [\set^k,\set]$, and,
for every $F = (F^1,F^2,F^*) \in RT_k$, both
$(\sigma^1_{F^1})_{\overline{A}}$ and
$(\sigma^2_{F^2})_{\overline{A}}$ are natural in $\overline{A}$.

Critically, we can compute $\omega$-directed colimits in
$RT_k$. Indeed, if $\cal D$ is an $\omega$-directed set then $\colim{d
  \in {\cal D}}{(F^1_d, F^2_d,F^*_d)} = (\colim{d \in {\cal
    D}}{F^1_d}, \colim{d \in {\cal D}}{F^2_d}, \colim{d \in {\cal
    D}}{F^*_d})$.  We define a higher-order relation transformer $T =
(T^1,T^2,T^*)$ on $RT_k$ to be \emph{$\omega$-cocontinuous} if $T^1$
and $T^2$ are $\omega$-cocontinuous endofunctors on $[\set^k,\set]$
and $T^*$ is an $\omega$-cocontinuous functor from $RT_k$ to
$[\rel^k,\rel]$, i.e., is in $[RT_k,[\rel^k,\rel]]$.
Now, for any $k$, any $A : \set$, and any $R : \rel(A, B)$, let
$K^\set_A$ be the constantly $A$-valued functor from $\set^k$ to
$\set$ and $K^\rel_R$ be the constantly $R$-valued functor from
$\rel^k$ to $\rel$.  Also let $0$ denote either the initial object of
either $\set$ or $\rel$, as appropriate.  Observing that, for every
$k$, $K^\set_0$ is initial in $[\set^k,\set]$, and $K^\rel_0$ is
initial in $[\rel^k,\rel]$, we have that, for each $k$, $K_0 =
(K^\set_0,K^\set_0,K^\rel_0)$ is initial in $RT_k$. Thus, if $T =
(T^1,T^2,T^*) : RT_k \to RT_k$ is a higher-order relation transformer
on $RT_k$ then we can define the relation transformer $\mu T$ to be
$\colim{n \in \nat}{T^n K_0}$. It is not hard to see that $\mu T$ is
given explicitly as
\begin{equation}\label{eq:mu}
\mu T = (\mu T^1,\mu T^2, \colim{n \in \nat}{(T^{n}K_0)^*})
\end{equation}
Moreover, $\mu T$ really is a fixpoint for $T$ if $T$ is
$\omega$-cocontinuous:
\begin{lem}\label{lem:fp}
For any $T : [RT_k,RT_k]$, $\mu T \cong T(\mu T)$.
\end{lem}
\noindent
The isomorphism is given by the morphisms $(\mathit{in}_1,
\mathit{in}_2) : T(\mu T) \to \mu T$ and $(in_1^{-1}, in_2^{-1}) : \mu
T \to T(\mu T)$ in $RT_k$. The latter is always a morphism in $RT_k$,
but the former need not be if $T$ is not $\omega$-cocontinuous.

It is worth noting that the third component in Equation~(\ref{eq:mu})
is the colimit in $[\rel^k,\rel]$ of third components of relation
transformers, rather than a fixpoint of an endofunctor on
$[\rel^k,\rel]$. There is thus an asymmetry between the first two
components of $\mu T$ and its third component, which reflects the
important conceptual observation that the third component of a
higher-order relation transformer on $RT_k$ need not be a functor on
all of $[\rel^k,\rel]$. In particular, although we can define
$T_{H,\rho}\, F$ for a relation transformer $F$ in
Definition~\ref{def:rel-sem} below, it is not clear how we could
define it for an arbitrary $F : [\rel^k,\rel]$.

\begin{defi}\label{def:reln-env}
A \emph{relation environment} maps each type variable in $\tvars^k \cup
\fvars^k$ to a $k$-ary relation transformer.  A morphism $f : \rho \to
\rho'$ between relation environments $\rho$ and $\rho'$ with
$\rho|_\tvars = \rho'|_\tvars$ maps each type constructor variable
$\psi^k \in \tvars$ to the identity morphism on $\rho \psi^k = \rho'
\psi^k$ and each functorial variable $\phi^k \in \fvars$ to a morphism
from the $k$-ary relation transformer $\rho \phi$ to the $k$-ary
relation transformer $\rho' \phi$. Composition of morphisms on
relation environments is given componentwise, with the identity
morphism mapping each relation environment to itself. This gives a
category of relation environments and morphisms between them, denoted
$\relenv$.
\end{defi}
We identify a $0$-ary relation transformer with the relation
(transformer) that is its codomain and consider a relation environment
to map a type variable of arity $0$ to a relation.  We write
$\rho[\ol{\alpha := R}]$ for the relation environment $\rho'$ such
that $\rho' \alpha_i \, = R_i$ for $i = 1,\dots,k$ and $\rho' \alpha =
\rho\alpha$ if $\alpha \not \in \{\alpha_1,\dots,\alpha_k\}$.  If $\rho$
is a relation environment, we write $\pi_1 \rho$ and $\pi_2 \rho$ for
the set environments mapping each type variable $\phi$ to the functors
$(\rho\phi)^1$ and $(\rho\phi)^2$, respectively.

\begin{defi}\label{def:relenv-functor}
  For each $k$, an \emph{$\omega$-cocontinuous environment transformer}
  $H$ is a triple $H = (H^1,H^2,H^*)$, where
\begin{itemize}
\item $H^1$ and $H^2$ are objects in $[\setenv,[\set^k,\set]]$
\item $H^*$ is an object in $[\relenv,[\rel^k,\rel]]$
\item $H^*\rho\, \ol{R} : \rel(H^1(\pi_1 \rho)\, \ol{A}, H^2(\pi_2
  \rho)\, \ol{B})$ whenever $R_1 : \rel(A_1,B_1),\dots,R_k :
  \rel(A_k,B_k)$
\item $H^*\rho\, \ol{(\alpha, \beta)} = (H^1(\pi_1 \rho)\,\ol{\alpha},
  H^2(\pi_2 \rho)\, \ol{\beta})$ whenever \[(\alpha_1, \beta_1) \in
  \Homrel(R_1,S_1),\dots, (\alpha_k, \beta_k) \in \Homrel(R_k,S_k)\]
\item For all $\overline{R : \rel(A,B)}$, $\pi_1(H^*f
  \,{\overline{R}}) = H^1 (\pi_1 f)\,{\overline{A}}$ and $\pi_2(H^*f
  \,{\overline{R}}) = H^2 (\pi_2 f)\,{\overline{B}}$
\end{itemize}
\end{defi}
\noindent
Definition~\ref{def:relenv-functor} entails that every
$\omega$-cocontinuous environment transformer $H$ is a
$\omega$-cocontinuous functor from $\relenv$ to $RT_k$, where
\begin{itemize}
\item The action of $H$ on $\rho$ in $\relenv$ is given by $H \rho = (H^1
  (\pi_1 \rho),\,H^2 (\pi_2 \rho),\,H^*\rho)$
\item The action of $H$ on morphisms $f : \rho \to \rho'$ in $\relenv$
  is given by $Hf = (H^1 (\pi_1 f),H^2 (\pi_2 f))$
\end{itemize}
\noindent
Note that the last condition of Definition~\ref{def:relenv-functor}
entails that if $f : \rho \to \rho'$, if
$R_1:\rel(A_1,B_1),\dots,R_k:\rel(A_k,B_k)$, and if $(x, y) \in
H^*\rho\,\ol{R}$, then \[(H^1(\pi_1 f)\,{\ol{A}}\,x, H^2(\pi_2
f)\,{\ol{B}}\,y) \in H^*\rho'\,\ol{R}\]

Considering $\relenv$ as a subcategory of the product $\Pi_{\phi^k \in
  \tvars \cup \fvars} RT_k$, computation of $\omega$-directed colimits
in $RT_k$ extends componentwise to $\relenv$. Recalling from the start
of this subsection that Definition~\ref{def:rel-sem} is given mutually
inductively with Definition~\ref{def:set-sem} we can now define our
relational interpretation. As with the set interpretation, the
relational interpretation is given in two parts, in
Definitions~\ref{def:rel-sem} and~\ref{def:rel-sem-funcs}.

\begin{defi}\label{def:rel-sem}
The \emph{relational interpretation} $\relsem{\cdot} : \F \to [\relenv,
 \rel]$ is defined by
\begin{align*}
  \relsem{\Gamma;\Phi \vdash \zerot}\rho &= 0\\
  \relsem{\Gamma;\Phi \vdash \onet}\rho &= 1\\
  \relsem{\Gamma; \emptyset \vdash \Nat^{\ol{\alpha}} \,F\,G}\rho &= \{\eta
  : \lambda \ol{R}.\,\relsem{\Gamma; \ol{\alpha} \vdash
    F}\rho[\ol{\alpha := R}] \Rightarrow \lambda \ol{R}. \,\relsem{
    \Gamma; \ol{\alpha} \vdash G}\rho[\ol{\alpha := R}]\}\\
  &=
  \{(t,t') \in \setsem{\Gamma; \emptyset
    \vdash \Nat^{\ol{\alpha}}
    \,F\,G} (\pi_1 \rho) \times \setsem{
    \Gamma;\emptyset
    \vdash \Nat^{\ol{\alpha}} \,F\,G} (\pi_2
  \rho)~|~\\
  & \hspace{0.3in} \forall {R_1 : \rel(A_1,B_1)}\,\dots\,{R_k : \rel(A_k,B_k)}.\\
  & \hspace{0.4in} (t_{\ol{A}},t'_{\ol{B}}) \in
  (\relsem{\Gamma; \ol{\alpha} \vdash G}\rho[\ol{\alpha :=
      R}])^{\relsem{\Gamma;\ol{\alpha}\vdash F}\rho[\ol{\alpha := R}]} \}\\
  \relsem{\Gamma;\Phi \vdash \phi \ol{F}}\rho &=
  (\rho\phi)\ol{\relsem{\Gamma;\Phi \vdash
    F}\rho}\\
  \relsem{\Gamma;\Phi \vdash F+G}\rho &=
  \relsem{\Gamma;\Phi \vdash F}\rho +
  \relsem{\Gamma;\Phi \vdash G}\rho\\
  \relsem{\Gamma;\Phi \vdash F\times G}\rho &=
  \relsem{\Gamma;\Phi \vdash F}\rho \times
  \relsem{\Gamma;\Phi \vdash G}\rho\\
   \relsem{\Gamma;\Phi \vdash (\mu \phi.\lambda
    \ol{\alpha}. H)\ol{G}}\rho
  &= (\mu T_{H,\rho})\ol{\relsem{\Gamma;\Phi \vdash G}\rho}\\
  \text{where }	T_{H,\rho}
    &= (T^\set_{H,\pi_1\rho}, T^\set_{H,\pi_2\rho}, T^\rel_{H,\rho}) \\
  \text{and } T^\rel_{H,\rho}\,F
    &= \lambda \ol{R}. \relsem{
      \Gamma;\phi,\ol{\alpha} \vdash H}\rho[\phi :=
    F][\ol{\alpha := R}]\\
  \text{and } T^\rel_{H,\rho}\,\delta
    &= \lambda \ol{R}. \relsem{
      \Gamma;\phi,\ol{\alpha} \vdash H}\id_\rho[\phi :=
    \delta][\ol{\alpha := \id_{\ol{R}}}]
\end{align*}
\end{defi}

The interpretations in Definitions~\ref{def:rel-sem}
and~\ref{def:rel-sem-funcs} below respect weakening, and also ensure
that the interpretations $\relsem{\Gamma \vdash F \to G}$ and
$\relsem{\Gamma \vdash \forall \ol\alpha. F}$ of the System F arrow
types and $\forall$-types representable in our calculus are as
expected in any parametric model.  As for set interpretations,
$\relsem{\Gamma; \emptyset \vdash \Nat^{\ol\alpha} F\,G}$ is
$\omega$-cocontinuous; indeed, it is constant (on $\omega$-directed
sets) because Definition~\ref{def:reln-env} ensures that restrictions
to $\mathbb{T}^k$ of morphisms between relational environments are
identities. If $\rho \in \relenv$ and $\vdash F$, then we write
$\relsem{\vdash F}$ instead of $\relsem{\vdash F}\rho$.  For the last
clause in Definition~\ref{def:rel-sem} to be well-defined we need
$T_{H,\rho}$ to be an $\omega$-cocontinuous higher-order relation
transformer on $RT_k$, where $k$ is the arity of $\phi$, so that, by
Lemma~\ref{lem:fp}, it admits a fixpoint. Since $T_{H,\rho}$ is
defined in terms of $\relsem{\Gamma;\phi, \ol{\alpha} \vdash H}$, this
means that relational interpretations of types must be
$\omega$-cocontinuous environment transformer from $\relenv$ to
$RT_0$, which in turn entails that the actions of relational
interpretations of types on objects and on morphisms in $\relenv$ are
intertwined. As for set interpretations, we know from~\cite{jp19}
that, for every $\Gamma; \ol{\alpha} \vdash F$, $\relsem{\Gamma;
  \ol{\alpha} \vdash F}$ is actually in $[\rel^k,\rel]$ where $k =
|\ol \alpha|$. We first define the actions of each of these functors
on morphisms between environments, and then argue that they are
well-defined and have the required properties. To do this, we show
that $T_H = (T^\set_H, T^\set_H, T^\rel_H) $ is a \emph{higher-order
  $\omega$-cocontinuous environment transformer}, as defined by

\begin{defi}\label{def:howcet}
A \emph{higher-order $\omega$-cocontinuous environment transformer}
is a triple $H = (H^1,H^2,H^*)$, where
\begin{itemize}
\item $H^1$ and $H^2$ are objects in
  $[\setenv,[[\set^k,\set],[\set^k,\set]]]$
\item $H^*$ is an object in $[\relenv,[RT_k,[\rel^k,\rel]]]$
\item $(H^1(\pi_1 \rho) F^1,H^2(\pi_2 \rho) F^2, H^* \rho F) \in
  RT_k$ whenever $F = (F^1,F^2,F^*) \in RT_k$
\item $H^*\rho\, (\delta^1, \delta^2)$ is a morphism in $RT_k$
  whenever $(\delta^1, \delta^2)$ is a morphism in $RT_k$ and,
  moreover, $H^*\rho\, (\delta^1, \delta^2) = (H^1(\pi_1
  \rho)\,\delta^1, H^2(\pi_2 \rho)\, \delta^2)$
\item For all $F = (F^1,F^2,F^*) \in RT_k$ and all $R_1 :
  \rel(A_1,B_1),\dots,R_k : \rel(A_k,B_k)$, $\pi_1(H^* f \,F\,\ol{R})$ $=
  H^1 (\pi_1 f)\,F^1\,\ol{A}$ and $\pi_2(H^*f \,F\,\ol{R}) = H^2
  (\pi_2 f)\,F^2\,\ol{B}$
\end{itemize}
\end{defi}
\noindent
Definition~\ref{def:howcet} entails that every higher-order
$\omega$-cocontinuous environment transformer $H$ is a
$\omega$-cocontinuous functor from $\relenv$ to
the category of $\omega$-cocontinuous higher-order
relation transformers on $RT_k$, where
\begin{itemize}
\item The action of $H$ on $\rho$ in $\relenv$ is given by $H \rho = (H^1
  (\pi_1 \rho),\,H^2 (\pi_2 \rho),\,H^*\rho)$
\item The action of $H$ on morphisms $f : \rho \to \rho'$ in $\relenv$
  is given by $Hf = (H^1 (\pi_1 f),H^2 (\pi_2 f))$
\end{itemize}
\noindent
Note that the last condition of Definition~\ref{def:howcet} entails
that if $f : \rho \to \rho'$, if $F = (F^1,F^2,F^*) \in RT_k$, if $R_1 :
\rel(A_1,B_1),\dots,R_k : \rel(A_k,B_k)$, and if $(x, y) \in
H^*\rho\,F\,\ol{R}$, then \[(H^1(\pi_1 f)\,F^1\,\ol{A}\,x, H^2(\pi_2
f)\,F^2\,\ol{B}\,y) \in H^*\rho'\,F\,\ol{R}\]

Now, the fact that $T_H = (T^\set_H, T^\set_H, T^\rel_H) $ is a
higher-order $\omega$-cocontinuous environment transformer follows
from the analogue for relations of the argument immediately preceding
Definition~\ref{def:set-sem-funcs}, together with the observations
that the action of $T^\rel_{H,\rho}$ on morphisms will be well-defined
by Definition~\ref{def:rel-sem-funcs}, and the functorial action of
$\relsem{\Gamma; \phi,\ol{\alpha} \vdash H}$ will already be available
in the final clause of Definition~\ref{def:howcet} from the induction
hypothesis on $H$ and Lemma~\ref{lem:rel-transf-morph} below. The
action of $T_H$ on an object $\rho \in \relenv$ is given by the
$\omega$-cocontinuous higher-order relation transformer $T_{H,\rho}$
whose actions on objects and morphisms are given in
Definition~\ref{def:howcet}. The action of $T_H$ on a morphism $f :
\rho \to \rho'$ is the morphism $T_{H,f} : T_{H,\rho} \to T_{H,\rho'}$
between higher-order relation transformers whose action on any $F \in
RT_k$ is the morphism of relation transfomers $T_{H,f}\, F :
T_{H,\rho}\, F \to T_{H,\rho'}\, F$ whose component at $\ol{R}$ is
$(T_{H,f}\, F)_{\ol{R}} = \relsem{\Gamma; \phi,\ol{\alpha} \vdash
  H}f[\phi := \id_F][\ol{\alpha := \id_R}]$, i.e., is
$(\setsem{\Gamma; \phi,\ol{\alpha} \vdash H}\,(\pi_1 f)\,[\phi :=
  \id_{F^1}][\ol{\alpha := \id_A}], \setsem{\Gamma; \phi,\ol{\alpha}
  \vdash H}\,(\pi_2 f)[\phi := \id_{F^2}][\ol{\alpha := \id_B}])$.
Using $T_H$, we can define the functorial action of relational
interpretation.
\begin{defi}\label{def:rel-sem-funcs}
Let $f: \rho \to \rho'$ for relation environments $\rho$ and $\rho'$
(so that $\rho|_\tvars = \rho'|_\tvars$). The action
$\relsem{\Gamma;\Phi \vdash F}f$ of $\relsem{\Gamma;\Phi \vdash F}$ on
the morphism $f$ is given exactly as in
Definition~\ref{def:set-sem-funcs}, except that all interpretations
are relational interpretations and all occurrences of $T^\set_{H,f}$
are replaced by $T_{H,f}$.
\end{defi}

If we define $\sem{\Gamma;\Phi \vdash F}$ to be $(\setsem{\Gamma;\Phi
  \vdash F}, \setsem{\Gamma;\Phi \vdash F},\relsem{\Gamma;\Phi \vdash
  F})$, then this is an immediate consequence of
\begin{lem}\label{lem:rel-transf-morph}
For every $\Gamma;\Phi \vdash F$, $\sem{\Gamma;\Phi \vdash F}$ is an
$\omega$-cocontinuous environment transformer.

\end{lem}
\noindent
The proof is a straightforward induction on the structure of $F$,
using an appropriate result from~\cite{jp19} to deduce
$\omega$-cocontinuity of $\sem{\Gamma;\Phi \vdash F}$ in each case,
together with Lemma~\ref{lem:fp} and Equation~\ref{eq:mu} for
$\mu$-types. We note that, for any relation environment $\rho$,
$\sem{\Gamma;\Phi \vdash F}\rho \in RT_0$.

We can prove by simultaneous induction that our interpretations of
types interact well with demotion of functorial variables to
non-functorial ones, along with other useful identities. Indeed, if
$\rho, \rho' : \setenv$, \,$f : \rho \to \rho'$, \,$\rho \phi = \rho
\psi = \rho' \phi = \rho' \psi$, \, $f \phi = f \psi = \id_{\rho
  \phi}$,\, $\Gamma; \Phi, \phi^k \vdash F$,\,
$\Gamma;\Phi,\ol{\alpha} \vdash G$,\, $\Gamma;\Phi,\alpha_1\dots\alpha_k
\vdash H$, and $\ol{\Gamma;\Phi \vdash K}$, then
\begin{gather}
\label{thm:demotion-objects}
\setsem{\Gamma; \Phi, \phi \vdash F} \rho = \setsem{\Gamma, \psi; \Phi
  \vdash F[\phi :== \psi] } \rho\\
\label{thm:demotion-morphisms}
\setsem{\Gamma; \Phi, \phi \vdash F} f = \setsem{\Gamma, \psi; \Phi
  \vdash F[\phi :== \psi]} f\\
\label{eq:subs-var}
\setsem{\Gamma;\Phi \vdash G[\ol{\alpha := K}]}\rho =
\setsem{\Gamma;\Phi,\ol{\alpha} \vdash G}\rho[\ol{\alpha :=
\setsem{\Gamma;\Phi \vdash K}\rho}]\\
\label{eq:subs-var-morph}
\setsem{\Gamma;\Phi \vdash G[\ol{\alpha := K}]}f =
\setsem{\Gamma;\Phi,\ol{\alpha} \vdash G}f[\ol{\alpha :=
\setsem{\Gamma;\Phi \vdash K}f}]\\
\label{eq:subs-const}
\setsem{\Gamma; \Phi \vdash F[\phi := H]}\rho
= \setsem{\Gamma; \Phi, \phi \vdash F}\rho
[\phi := \lambda \ol{A}.\, \setsem{\Gamma;\Phi,\overline{\alpha}\vdash
    H}\rho[\overline{\alpha := A}]] \\
\label{eq:subs-const-morph}
\setsem{\Gamma; \Phi \vdash F[\phi := H]}f
= \setsem{\Gamma; \Phi, \phi \vdash F}f
[\phi := \lambda \ol{A}.\,\setsem{\Gamma;\Phi,\overline{\alpha}\vdash
    H}f[\overline{\alpha := \id_A}]]
\end{gather}
Identities analogous to (\ref{thm:demotion-objects}) through
(\ref{eq:subs-const-morph}) hold for relational interpretations as well.

\section{The Identity Extension Lemma}\label{sec:iel}

In most treatments of parametricity, equality relations on sets are
taken as \emph{given} --- either directly as diagonal relations, or
perhaps via reflexive graphs if kinds are also being tracked --- and
the graph relations used to validate existence of initial algebras are
defined in terms of them. We take a different approach here, giving a
categorical definition of graph relations for morphisms (i.e., natural
transformations) between functors and \emph{constructing} equality
relations as particular graph relations. Our definitions specialize to
the usual ones for the graph relation for morphisms between sets and
equality relations on sets. In light of its novelty, we spell out
our construction in detail.

The standard definition of the graph for a morphism $f : A \to B$ in
$\set$ is the relation $\graph{f} : \rel(A,B)$ defined by $(x,y) \in
\graph{f}$ iff $fx = y$. This definition naturally generalizes to
associate to each natural transformation between $k$-ary functors on
$\set$ a $k$-ary relation transformer as follows:

\begin{defi}\label{dfn:graph-nat-transf}
If $F, G: \Set^k \to \Set$ and $\alpha : F \to G$ is a natural
transformation, then the functor $\graph{\alpha}^*: \rel^k \to \rel$
is defined as follows. Given $R_1 : \rel(A_1, B_1),\dots,R_k :
\rel(A_k,B_k)$, let $\iota_{R_i} : R_i \hookrightarrow A_i \times
B_i$, for $i = 1,\dots,k$, be the inclusion of $R_i$ as a subset of $A_i
\times B_i$,
let $h_{\overline{A \times B}}$ be the unique morphism making the diagram
{\footnotesize\[\begin{tikzcd}[row sep = large]
        F\overline{A}
        &F(\overline{A \times B})
        \ar[l, "{F\overline{\pi_1}}"']
        \ar[r, "{F\overline{\pi_2}}"]
        \ar[d, dashed, "{h_{\overline{A \times B}}}"]
        &F\overline{B}
        \ar[r, "{\alpha_{\ol{B}}}"]
        &G\overline{B} \\
        &F\overline{A} \times G\overline{B}
        \ar[ul, "{\pi_1}"] \ar[urr, "{\pi_2}"']
\end{tikzcd}\]}

\noindent
commute, and let $h_{\overline{R}} : F\overline{R} \to F\overline{A}
\times G\overline{B}$ be $h_{\overline{A \times B}} \circ
F\overline{\iota_R}$. Further, let $\alpha^\wedge\overline{R}$ be the
subobject through which $h_{\overline{R}}$ is factorized by the
mono-epi factorization system in $\set$, as shown in the
following diagram:
{\footnotesize\[\begin{tikzcd}
        F\overline{R}
        \ar[rr, "{h_{\overline{R}}}"]
        \ar[dr, twoheadrightarrow, "{q_{\alpha^\wedge\overline{R}}}"']
        &&F\overline{A} \times G\overline{B} \\
        &\alpha^\wedge\overline{R}
        \ar[ur, hookrightarrow, "{\iota_{\alpha^\wedge\overline{R}}}"']
\end{tikzcd}\]}

\noindent
Then $\alpha^\wedge\overline{R} : \rel(F\overline{A}, G\overline{B})$
by construction, so the action of $\langle \alpha \rangle^*$ on
objects can be given by $\langle \alpha \rangle^* \overline{(A,B,R)} =
(F\overline{A}, G\overline{B}, \iota_{\alpha^\wedge
  \overline{R}}\alpha^\wedge\overline{R})$. Its action on morphisms is
given by $\graph{\alpha}^*\overline{(\beta, \beta')} =
(F\overline\beta, G\overline\beta')$.
\end{defi}

The data in Definition~\ref{dfn:graph-nat-transf} yield the \emph{graph
  relation transformer for $\alpha$}, denoted $\graph{\alpha} = (F, G,
\graph{\alpha}^*)$.

\begin{lem}\label{lem:graph-reln-functors}
If $F,G : [\set^k,\set]$, and if $\alpha : F \to G$ is a natural
transformation, then $\graph{\alpha}$ is in $RT_k$.
\end{lem}
\proof
Clearly, $\graph{\alpha}^*$ is $\omega$-cocontinuous, so
$\graph{\alpha}^* : [\rel^k,\rel]$. Let $\overline{R :
  \rel(A, B)}$, $\overline{S : \rel(C, D)}$, and $\overline{(\beta,
  \beta') : R \to S}$. We want to show that there exists a morphism
$\epsilon : \alpha^\wedge\overline{R} \to \alpha^\wedge\overline{S}$
such that the diagram on the left below commutes. Since
$\ol{(\beta,\beta') : R \to S}$, there exist $\overline{\gamma : R \to
  S}$ such that each diagram in the middle commutes.
Moreover, since both $h_{\overline{C \times D}} \circ
F(\overline{\beta \times \beta'})$ and $(F\overline{\beta} \times
G\overline{\beta'}) \circ h_{\overline{A \times B}}$ make the diagram
on the right commute, they must be equal.
\begin{figure*}[ht]
  \hspace*{-1in}
  \begin{minipage}[b]{0.25\linewidth}
 {\small    \[
    \begin{tikzcd}
        \alpha^\wedge\overline{R}
        \ar[r, hookrightarrow, "{\iota_{\alpha^\wedge\overline{R}}}"]
        \ar[d, "{\epsilon}"']
        & F\overline{A} \times G\overline{B}
        \ar[d, "{F\overline{\beta} \times G\overline{\beta'}}"] \\
        \alpha^\wedge\overline{S}
        \ar[r, hookrightarrow, "{\iota_{\alpha^\wedge\overline{S}}}"']
        & F\overline{C} \times G\overline{D}
    \end{tikzcd}
    \]}
\end{minipage}\hspace*{0.2in}
\begin{minipage}[b]{0.25\linewidth}
{\small    \[
    \begin{tikzcd}
        R_i
        \ar[d, "{\gamma_i}"']
        \ar[r, hookrightarrow, "{\iota_{R_i}}"]
        &A_i \times B_i
        \ar[d, "{\beta_i \times \beta'_i}"] \\
        S_i
        \ar[r, hookrightarrow, "{\iota_{S_i}}"]
        &C_i \times D_i
    \end{tikzcd}
    \]}
\end{minipage}\hspace*{0.2in}
\begin{minipage}[b]{0.25\linewidth}
{\footnotesize \[
      \begin{tikzcd}[row sep = large]
          F\overline{C}
          &F\overline{C} \times G\overline{D}
          \ar[l, "{\pi_1}"'] \ar[r, "{\pi_2}"]
          &G\overline{D}\\
          &F(\overline{A \times B})
          \ar[u, dashed, "{\exists !}"]
          \ar[ul, "{F\ol{\pi_1} \circ F(\overline{\beta \times \beta'})}"]
          \ar[ur, "{\alpha_{\ol{D}} \circ F\ol{\pi_2} \circ F(\overline{\beta \times \beta'})}"']
      \end{tikzcd}
      \]}
\end{minipage}
\end{figure*}


\noindent
We therefore get that the right-hand square in the diagram on the left
below commutes, and thus that the entire diagram does as well.
Finally, by the left-lifting property of $q_{F^\wedge\overline{R}}$
with respect to $\iota_{F^\wedge\overline{S}}$ given by the mono-epi
factorization system, there exists an $\epsilon$ such that the diagram
on the right below commutes as desired.
\begin{figure*}[ht]
  \vspace*{-0.1in}
  \hspace*{-0.5in}
  \begin{minipage}[b]{0.45\linewidth}
{\footnotesize \[
      \begin{tikzcd}
          F\overline{R}
          \ar[d, "{F\overline{\gamma}}"']
          \ar[r, hookrightarrow, "{F\overline{\iota_R}}"]
          \ar[rr, bend left, "{h_{\overline{R}}}"]
          &F(\overline{A \times B})
          \ar[d, "{F(\overline{\beta \times \beta'})}"]
          \ar[r, "{h_{\overline{A \times B}}}"]
          &F\overline{A} \times G\overline{B}
          \ar[d, "{F\overline{\beta} \times G\overline{\beta'}}"] \\
          F\overline{S}
          \ar[r, hookrightarrow, "{F\overline{\iota_S}}"']
          \ar[rr, bend right, "{h_{\overline{S}}}"']
          &F(\overline{C \times D})
          \ar[r, "{h_{\overline{C \times D}}}"']
          &F\overline{C} \times G\overline{D}
      \end{tikzcd}
      \]}
\end{minipage}
  \vspace*{-0.5in}
  \begin{minipage}[b]{0.45\linewidth}
      {\footnotesize
        \[  \vspace*{0.4in}
      \begin{tikzcd}
          F\overline{R}
          \ar[d, "{F\overline{\gamma}}"']
          \ar[r, twoheadrightarrow, "{q_{\alpha^\wedge\overline{R}}}"]
          &\alpha^\wedge\overline{R}
          \ar[d, dashed, "{\epsilon}"]
          \ar[r, hookrightarrow, "{\iota_{\alpha^\wedge\overline{R}}}"]
          &F\overline{A} \times G\overline{B}
          \ar[d, "{F\overline{\beta} \times G\overline{\beta'}}"] \\
          F\overline{S}
          \ar[r, twoheadrightarrow, "{q_{\alpha^\wedge\overline{S}}}"']
          &\alpha^\wedge\overline{S}
          \ar[r, hookrightarrow, "{\iota_{\alpha^\wedge\overline{S}}}"']
          &F\overline{C} \times G\overline{D}
      \end{tikzcd}
      \]}
\end{minipage}
\end{figure*}

\qed

\vspace*{0.1in}

If $f : A \to B$ is a morphism in $\set$ then the definition of the
graph relation transformer $\langle f \rangle$ for $f$ as a natural
transformation between $0$-ary functors $A$ and $B$ coincides with its
standard definition.  Graph relation transformers are thus a
reasonable extension of graph relations to functors.

The action of a graph relation transformer on a graph relation can be
computed explicitly:

\begin{lem}\label{lem:eq-reln-equalities}
If $\alpha : F \to G$ is a morphism in $[\Set^k, \Set]$
and $f_1: A_1 \to B_1, \dots, f_k : A_k \to B_k$,
then $\graph{\alpha}^* \graph{\overline{f}}
= \langle G \ol{f} \circ \alpha_{\ol{A}} \rangle
= \langle \alpha_{\ol{B}} \circ F \ol{f} \rangle$.
\end{lem}
\proof
Since $h_{\overline{A \times B}}$ is the unique morphism making the
bottom triangle of the diagram on the left below commute, and since
$h_{\graph{\overline{f}}} = h_{\overline{A \times B}} \circ F
\,\ol{\iota_{\graph{f}}} = h_{\overline{A \times B}} \circ F
\overline{\langle \id_A, f \rangle}$, the universal property of the
product depicted in the diagram on the right gives
$h_{\graph{\overline{f}}} = \langle \id_{F \ol{A}}, \alpha_{\ol{B}}
\circ F\ol{f} \rangle : F \ol{A} \to F \ol{A} \times G \ol{B}$.

\begin{figure*}[ht]
  \vspace*{-0.15in}
  \begin{minipage}[b]{0.45\linewidth}
{\footnotesize
\[\begin{tikzcd}[row sep = large]
        &F\overline{A}
        \ar[d, "{F \overline{\langle \id_A, f \rangle}}" description]
        \ar[dl, equal]
        \ar[dr, "{F\ol{f}}"]\\
        F\overline{A}
        &F(\overline{A \times B})
        \ar[l, "{F\overline{\pi_1}}"']
        \ar[r, "{F\overline{\pi_2}}"]
        \ar[d, "{h_{\overline{A \times B}}}"]
        &F\overline{B}
        \ar[r, "{\alpha_{\ol{B}}}"]
        &G\overline{B}\\
        &F\overline{A} \times G\overline{B}
        \ar[ul, "{\pi_1}"] \ar[urr, "{\pi_2}"']
\end{tikzcd}\]}
\end{minipage}
  \begin{minipage}[b]{0.45\linewidth}
{\footnotesize
\[
      \begin{tikzcd}[row sep = large]
          F\overline{A}
          &F\overline{A} \times G\overline{B}
          \ar[l, "{\pi_1}"'] \ar[r, "{\pi_2}"]
          &G\overline{B}\\
          &F\overline{A}
          \ar[u, dashed, "{\exists !}"]
          \ar[ul, equal]
          \ar[r, "{F\ol{f}}"']
          &F{\ol{B}}
          \ar[u, "{\alpha_{\ol{B}}}"']
      \end{tikzcd}
      \]}
\end{minipage}
\end{figure*}

\noindent
Moreover, $\langle \id_{F \ol{A}}, \alpha_{\ol{B}} \circ F\ol{f}
\rangle$ is a monomorphism in $\set$ because $\id_{F \ol{A}}$ is, so
its mono-epi factorization gives $\iota_{\alpha^\wedge
  \graph{\overline{f}}} = \langle \id_{F \ol{A}}, \alpha_{\ol{B}}
\circ F\ol{f} \rangle$, and thus $\alpha^\wedge \graph{\overline{f}} =
F\overline{A}$.  Then $\iota_{\alpha^\wedge \graph{\overline{f}}}
\alpha^\wedge \graph{\overline{f}} = \langle \id_{F \ol{A}},
\alpha_{\ol{B}} \circ F\ol{f} \rangle (F \ol{A}) = \graph{
  \alpha_{\ol{B}} \circ F\ol{f} }^*$, so that
\[\graph{ \alpha }^*
\graph{ \overline{f} } = (F\overline{A}, G\overline{B},
\iota_{\alpha^\wedge \graph{\overline{f}}}\, \alpha^\wedge
\graph{\overline{f}}) = (F\overline{A}, G\overline{B}, \graph{
  \alpha_{\ol{B}} \circ F\ol{f} }^*) = \graph{ \alpha_{\ol{B}} \circ
  F\ol{f} }\]  Finally, $\alpha_{\ol{B}} \circ F\ol{f} = G\ol{f} \circ
\alpha_{\ol{A}}$ by naturality of $\alpha$.
\qed

To prove the IEL, we also need to know that the equality relation
transformer preserves equality relations.  The \emph{equality relation
  transformer} on $F : [\set^k,\set]$ is defined to be
\begin{equation}\label{def:eq-transf}
  \Eq_F = \graph{\id_{F}} = (F, F, \graph{\id_{F}}^*)
\end{equation}
Lemma~\ref{lem:eq-reln-equalities} then gives that, for all $\ol{A :
\set}$,
\begin{equation}\label{eq:eq-pres-eq}
\Eq^*_F \ol{\Eq_A}
= \graph{\id_F}^* \graph{\id_{\ol{A}}}
= \graph{F \id_{\ol{A}} \circ (\id_F)_{\ol{A}}}
= \graph{\id_{F\ol{A}} \circ \id_{F\ol{A}}}
= \graph{\id_{F\ol{A}}}
= \Eq_{F\ol{A}}
\end{equation}

Graph relation transformers in general, and equality relation
transformers in particular, naturally extend to relation environments.
Indeed, if $\rho, \rho' : \setenv$ and $f : \rho \to \rho'$, then the
\emph{graph relation environment} $\graph{f}$ is defined pointwise by
$\graph{f} \phi = \graph{f \phi}$ for every $\phi$, which entails that
$\pi_1 \graph{f} = \rho$ and $\pi_2 \graph{f} = \rho'$. In particular,
the \emph{equality relation environment} $\Eq_\rho$ is defined to be
$\graph{\id_{\rho}}$, which entails that $\Eq_\rho \phi = \Eq_{\rho
  \phi}$ for every $\phi$.

With these definitions in hand, we can state and prove both an
Identity Extension Lemma and a Graph Lemma for our calculus.
\begin{thm}[IEL]\label{thm:iel}
If $\rho : \setenv$ and $\Gamma; \Phi \vdash F$ then $\relsem{\Gamma;
  \Phi \vdash F} \Eq_\rho = \Eq_{\setsem{\Gamma; \Phi \vdash F}\rho}$.
\end{thm}
\proof
The proof is by induction on the structure of $F$. Only the $\Nat$,
application, and fixpoint cases are non-routine.
\begin{itemize}
\item $\relsem{\Gamma; \Phi \vdash \zerot} \Eq_{\rho} = 0_\rel =
  \Eq_{0_\set} = \Eq_{\setsem{\Gamma; \Phi \vdash \zerot}\rho}$
\item $\relsem{\Gamma; \Phi \vdash \onet} \Eq_{\rho} = 1_\rel =
  \Eq_{1_\set} = \Eq_{\setsem{\Gamma; \Phi \vdash \onet}\rho}$
\item By definition, $\relsem{\Gamma; \emptyset \vdash
  \Nat^{\overline\alpha} \,F\,G} \Eq_{\rho}$ is the relation on
  $\setsem{\Gamma; \emptyset \vdash \Nat^{\overline\alpha} \,F\,G}
  \rho$ relating $t$ and $t'$ if, for all ${R_1 :
    \rel(A_1,B_1)},\dots,{R_k : \rel(A_k,B_k)}$, $(t_{\overline{A}},
  t'_{\overline{B}})$ is a morphism from $\relsem{\Gamma;
    \overline\alpha \vdash F} \Eq_{\rho}\overline{[\alpha := R]}$ to
  $\relsem{\Gamma ; \overline\alpha \vdash G}
  \Eq_{\rho}\overline{[\alpha := R]}$ in $\rel$.  To prove that this
  relation is $\Eq_{\setsem{\Gamma; \emptyset \vdash
      \Nat^{\overline\alpha} \,F\,G} \rho}$ we must show that
  $(t_{\overline{A}}, t'_{\overline{B}})$ is a morphism from
  $\relsem{\Gamma; \overline\alpha \vdash F}
  \Eq_{\rho}\overline{[\alpha := R]}$ to $\relsem{\Gamma ;
    \overline\alpha \vdash G} \Eq_{\rho}\overline{[\alpha := R]}$ in
  $\rel$ for all ${R_1 : \rel(A_1,B_1)},\dots,{R_k : \rel(A_k,B_k)}$ if
  and only if $t = t'$.  We first show that if $(t_{\overline{A}},
  t'_{\overline{B}})$ is a morphism from $\relsem{\Gamma;
    \overline\alpha \vdash F} \Eq_{\rho}\overline{[\alpha := R]}$ to
  $\relsem{\Gamma ; \overline\alpha \vdash G}
  \Eq_{\rho}\overline{[\alpha := R]}$ for all ${R_1 : \rel(A_1,B_1),}$
  $\dots,{R_k : \rel(A_k,B_k)}$ then $t = t'$.  By hypothesis, for all
  $A_1\,\dots\,A_k : \set$, $(t_{\overline{A}}, t'_{\overline{A}})$ is a
  morphism from $\relsem{\Gamma; \overline\alpha \vdash F}
  \Eq_{\rho}\overline{[\alpha := \Eq_{A}]}$ to $\relsem{\Gamma ;
    \overline\alpha \vdash G} \Eq_{\rho}\overline{[\alpha :=
      \Eq_{A}]}$. By the induction hypothesis, it is therefore a
  morphism from $\Eq_{\setsem{\Gamma; \overline\alpha \vdash F}
    \rho\overline{[\alpha := A]}}$ to $\Eq_{\setsem{\Gamma;
      \overline\alpha \vdash G} \rho\overline{[\alpha := A]}}$ in
  $\rel$. This means that, for all $x : \Eq_{\setsem{\Gamma;
      \overline\alpha \vdash F} \rho\overline{[\alpha := A]}}$,
  $t_{\overline{A}}x = t'_{\overline{A}}x$, so $t = t'$ by
  extensionality. For the converse, we observe that if $t = t'$ then
  the desired conclusion is simply the extra condition in the
  definition of $\setsem{\Gamma;\emptyset \vdash
    \Nat^{\overline\alpha} \,F\,G} \rho$.

\item The application case is proved by the following sequence of
  equalities, where the second equality is by the induction hypothesis
  and the definition of the relation environment $\Eq_\rho$, the third
  is by the definition of application of relation transformers from
  Definition~\ref{def:rel-transf}, and the fourth is by
  Equation~\ref{eq:eq-pres-eq}:
\[
\begin{split}
\relsem{\Gamma; \Phi \vdash \phi\ol{F}}\Eq_{\rho} &=
(\Eq_{\rho}\phi)\ol{\relsem{\Gamma; \Phi \vdash F}
\Eq_{\rho}}\\
&= \Eq_{\rho \phi}\, \ol{\Eq_{\setsem{\Gamma; \Phi \vdash F}
  \rho}}\\
&= (\Eq_{\rho \phi})^* \,\ol{\Eq_{\setsem{\Gamma; \Phi \vdash F}
  \rho}}\\
&= \Eq_{(\rho \phi) \,\ol{\setsem{\Gamma; \Phi \vdash F} \rho}}\\
&= \Eq_{\setsem{\Gamma; \Phi \vdash \phi\ol{F}}\rho}
\end{split}
\]
\item  The fixpoint case is proven by the sequence of equalities
\[
\begin{split}
\relsem{\Gamma; \Phi \vdash (\mu \phi.\lambda
  \ol{\alpha}. H)\ol{F}}\Eq_{\rho}
&=(\mu {T_{H, \Eq_{\rho}}}) \,\ol{\relsem{\Gamma; \Phi \vdash F}\Eq_{\rho}}\\
&= \colim{n \in \nat}{T^n_{H, \Eq_{\rho}} K_0}\, \ol{\relsem{\Gamma; \Phi
  \vdash F}\Eq_{\rho}}\\
&= \colim{n \in \nat}{ T^n_{H, \Eq_{\rho}} K_0 \,\ol{\Eq_{\setsem{\Gamma;
    \Phi \vdash F}\rho}}}\\
&= \colim{n \in \nat}{(\Eq_{(T^\set_{H,\rho})^n K_0})^*
  \ol{\Eq_{\setsem{\Gamma; \Phi \vdash F}\rho}}}\\
&= \colim{n \in \nat}{\Eq_{(T^\set_{H,\rho})^n K_0 \,\ol{\setsem{\Gamma;
        \Phi \vdash F}\rho}}}\\
&= \Eq_{\colim{n \in \nat}{ (T^\set_{H,\rho})^n K_0\,
    \ol{\setsem{\Gamma; \Phi \vdash F}\rho}}}\\
&= \Eq_{\setsem{\Gamma; \Phi \vdash (\mu \phi.\lambda
      \ol{\alpha}. H)\ol{F}}\rho}
\end{split}
\]
Here, the third equality is by induction hypothesis, the fifth is by
Equation~\ref{eq:eq-pres-eq}, and the fourth equality is because, for
every $n \in \nat$, the following two statements can be proved by
simultaneous induction: for any $H$, $\rho$, and $A$,
\begin{equation}\label{eq:iel-fix-point-intermediate1}
T^n_{H, \Eq_{\rho}} K_0\, \ol{\Eq_A} = (\Eq_{(T^\set_{H, \rho})^n K_0})^*
\ol{\Eq_A}
\end{equation}
and for any subformula $J$ of $H$,
\begin{equation}\label{eq:iel-fix-point-intermediate2}
\begin{split}
& \relsem{\Gamma; \phi, \ol{\alpha} \vdash J} \Eq_{\rho} [\phi
    := T^{n}_{H,\Eq_{\rho}} K_0] \overline{[\alpha := \Eq_A]} \\
=\;\; & \relsem{\Gamma; \phi, \ol{\alpha} \vdash J} \Eq_{\rho}
[\phi := \Eq_{(T^\set_{H,\rho})^n K_0}] \overline{[\alpha := \Eq_A]}
\end{split}
\end{equation}
The case $n=0$ is trivial:
Equation~\ref{eq:iel-fix-point-intermediate1} and
Equation~\ref{eq:iel-fix-point-intermediate2} both hold because
$T^0_{H,\Eq_{\rho}} K_0 = K_0$, $(T^\set_{H,\rho})^0 K_0 = K_0$, and
Equation~\ref{eq:eq-pres-eq} holds.

The inductive case is proved as follows.  We prove
Equation~\ref{eq:iel-fix-point-intermediate1} by the following
sequence of equalities:
\[
\begin{split}
T^{n+1}_{H,\Eq_{\rho}} K_0\, \overline{\Eq_A}
&= T^\rel_{H,\Eq_{\rho}} (T^{n}_{H,\Eq_{\rho}} K_0)
\overline{\Eq_A} \\
&= \relsem{\Gamma; \phi, \ol{\alpha} \vdash H} \Eq_{\rho} [\phi
  := T^{n}_{H, \Eq_{\rho}} K_0] \overline{[\alpha :=
    \Eq_A]} \\
&= \relsem{\Gamma; \phi, \ol{\alpha} \vdash H} \Eq_{\rho} [\phi
  := \Eq_{(T^\set_{H,\rho})^{n} K_0}] \overline{[\alpha :=
    \Eq_A]} \\
&= \relsem{\Gamma; \phi, \ol{\alpha} \vdash H} \Eq_{\rho [\phi
    := (T^\set_{H,\rho})^{n} K_0] \overline{[\alpha :=
      A]}} \\
&= \Eq_{\setsem{\Gamma; \phi, \ol{\alpha} \vdash H} \rho [\phi
    := (T^\set_{H,\rho})^{n} K_0] \overline{[\alpha :=
      A]}} \\
&= \Eq_{(T^\set_{H,\rho})^{n+1} K_0 \overline{A}} \\
&= (\Eq_{(T^\set_{H,\rho})^{n+1} K_0})^*\, \overline{\Eq_A}
\end{split}
\]
Here, the third equality is by the induction hypothesis of
Equation~\ref{eq:iel-fix-point-intermediate2} for $J = H$, the fifth
by the induction hypothesis of the IEL on $H$, and the last is by
Equation~\ref{eq:eq-pres-eq}.

\vspace*{0.1in}

We prove Equation~\ref{eq:iel-fix-point-intermediate2} by structural
induction on $J$. The only interesting cases, though, are when $J =
\phi \ol{G}$ and when $J = (\mu \psi.\lambda \ol\beta. G)\ol K$.
\begin{itemize}
\setlength{\itemsep}{0.5em}
\item The case $J = \phi \ol G$ is proved by the sequence of equalities:
\[
\begin{split}
& \relsem{\Gamma; \phi, \ol{\alpha} \vdash \phi
    \ol{G}} \Eq_{\rho} [\phi := T^{n+1}_{H,\Eq_{\rho}} K_0]
  \overline{[\alpha := \Eq_A]}
  \\
&= T^{n+1}_{H,\Eq_{\rho}} K_0\, \overline{\relsem{\Gamma;
      \phi, \ol{\alpha} \vdash G} \Eq_{\rho} [\phi :=
      T^{n+1}_{H,\Eq_{\rho}} K_0] \overline{[\alpha :=
        \Eq_A]}} \\
&= T^{n+1}_{H,\Eq_{\rho}} K_0\, \overline{\relsem{\Gamma;
      \phi, \ol{\alpha} \vdash G} \Eq_{\rho} [\phi :=
      \Eq_{(T^\set_{H,\rho})^{n+1} K_0}] \overline{[\alpha :=
        \Eq_A]}} \\
&= T^{n+1}_{H,\Eq_{\rho}} K_0\, \overline{\relsem{\Gamma;
      \phi, \ol{\alpha} \vdash G} \Eq_{\rho [\phi := (T^\set_{H,\rho})^{n+1}
        K_0] \overline{[\alpha := A]}}} \\
&= T^{n+1}_{H,\Eq_{\rho}} K_0\, \overline{\Eq_{\setsem{\Gamma;
        \phi, \ol{\alpha} \,\vdash G} \rho [\phi :=
        (T^\set_{H,\rho})^{n+1} K_0] \overline{[\alpha :=
          A]}}} \\
&= (\Eq_{(T^\set_{H,\rho})^{n+1} K_0})^* \,\overline{\Eq_{\setsem{\Gamma;
        \phi, \ol{\alpha} \,\vdash G} \rho [\phi :=
        (T^\set_{H,\rho})^{n+1} K_0] \overline{[\alpha :=
          A]}}} \\
&= (\Eq_{(T^\set_{H,\rho})^{n+1} K_0})^* \overline{\relsem{\Gamma;
      \phi, \ol{\alpha} \vdash G} \Eq_{\rho} [\phi :=
      \Eq_{(T^\set_{H,\rho})^{n+1} K_0}] \overline{[\alpha :=
        \Eq_A]}} \\
&= \relsem{\Gamma; \phi, \ol{\alpha} \vdash \phi \ol{G}}
  \Eq_{\rho} [\phi := \Eq_{(T^\set_{H,\rho})^{n+1} K_0}]
  \overline{[\alpha := \Eq_A]}
\end{split}
\]
Here, the second equality is by the induction hypothesis for
Equation~\ref{eq:iel-fix-point-intermediate2} on the $G$s, the fourth
is by the induction hypothesis for the IEL on the $G$s, and the fifth
is by Equation~\ref{eq:iel-fix-point-intermediate1}, which we have
just proved.

\item
  The case $J = (\mu \psi.\, \lambda \ol{\beta}.\, G) \ol{K}$
is proved by the sequence of equalities
\[
\begin{split}
& \relsem{\Gamma; \phi, \ol{\alpha} \vdash (\mu \psi.\, \lambda
    \ol{\beta}.\, G) \ol{K}}
  \Eq_{\rho} [\phi := T^{n+1}_{H,\Eq_{\rho}} K_0]
  [\overline{\alpha := \Eq_A}] \\
&= (\mu {T_{G, \Eq_{\rho}[\phi := T^{n+1}_{H,\Eq_{\rho}} K_0][\overline{\alpha := \Eq_A}]}})
  \,\ol{\relsem{\Gamma; \phi, \ol{\alpha} \vdash K}\Eq_{\rho}
  [\phi := T^{n+1}_{H,\Eq_{\rho}} K_0][\overline{\alpha := \Eq_A}]} \\
&= \colim{m \in \nat} {T^m_{G, \Eq_{\rho}[\phi := T^{n+1}_{H,\Eq_{\rho}} K_0][\overline{\alpha := \Eq_A}]}}\,K_0
  \,(\ol{\relsem{\Gamma; \phi, \ol{\alpha} \vdash K}\Eq_{\rho}
  [\phi := T^{n+1}_{H,\Eq_{\rho}} K_0][\overline{\alpha := \Eq_A}]}) \\
&= \colim{m \in \nat} {T^m_{G, \Eq_{\rho}[\phi := T^{n+1}_{H,\Eq_{\rho}} K_0][\overline{\alpha := \Eq_A}]}}\,K_0
  \,(\ol{\relsem{\Gamma; \phi, \ol{\alpha} \vdash K}\Eq_{\rho}
  [\phi := \Eq_{(T^{\set}_{H,\rho})^{n+1} K_0}][\overline{\alpha := \Eq_A}]}) \\  
&= \colim{m \in \nat} {T^m_{G, \Eq_{\rho}[\phi := \Eq_{(T^{\set}_{H,\rho})^{n+1} K_0}][\overline{\alpha := \Eq_A}]}}\,K_0
  \,(\ol{\relsem{\Gamma; \phi, \ol{\alpha} \vdash K}\Eq_{\rho}
  [\phi := \Eq_{(T^{\set}_{H,\rho})^{n+1} K_0}][\overline{\alpha := \Eq_A}]}) \\
&= (\mu {T_{G, \Eq_{\rho}[\phi := \Eq_{(T^{\set}_{H,\rho})^{n+1} K_0}][\overline{\alpha := \Eq_A}]}})
  \,\ol{\relsem{\Gamma; \phi, \ol{\alpha} \vdash K}\Eq_{\rho}
  [\phi := \Eq_{(T^{\set}_{H,\rho})^{n+1} K_0}][\overline{\alpha := \Eq_A}]} \\
&= \relsem{\Gamma; \phi, \ol{\alpha} \vdash (\mu \psi.\, \lambda \beta.\, G) \ol{K}}
  \Eq_{\rho}[\phi := \Eq_{(T^{\set}_{H,\rho})^{n+1} K_0}][\overline{\alpha := \Eq_{A}}]
\end{split}
\]
Here, the third equality is by the induction hypothesis for
Equation~\ref{eq:iel-fix-point-intermediate2} on the $K$s, and the
fourth equality holds because we can prove that, for all $m \in \nat$,
\begin{equation}\label{eq:helper}
T^m_{G,
\Eq_{\rho}[\phi := T^{n+1}_{H,\Eq_{\rho}} K_0][\overline{\alpha :=
    \Eq_A}]}\,K_0 \ = T^m_{G, \Eq_{\rho}[\phi :=
    \Eq_{(T^{\set}_{H,\rho})^{n+1} K_0}][\overline{\alpha :=
      \Eq_A}]}\,K_0
\end{equation}
Indeed, the base case of Equation~\ref{eq:helper} is trivial because
\[T^0_{G,
  \Eq_{\rho}[\phi := T^{n+1}_{H,\Eq_{\rho}} K_0][\overline{\alpha :=
      \Eq_A}]}\,K_0 \ = K_0 = T^0_{G, \Eq_{\rho}[\phi :=
    \Eq_{(T^{\set}_{H,\rho})^{n+1} K_0}][\overline{\alpha :=
      \Eq_A}]}\,K_0\]
and the inductive case is proved by: \\
\begin{align*}
& T^{m+1}_{G, \Eq_{\rho}[\phi := T^{n+1}_{H,\Eq_{\rho}} K_0][\overline{\alpha := \Eq_A}]}\,K_0 \\
&= T_{G, \Eq_{\rho}[\phi := T^{n+1}_{H,\Eq_{\rho}} K_0][\overline{\alpha
        := \Eq_A}]} (T^{m}_{G, \Eq_{\rho}[\phi := T^{n+1}_{H,\Eq_{\rho}}
      K_0][\overline{\alpha := \Eq_A}]}\,K_0) \\
&= T_{G, \Eq_{\rho}[\phi := T^{n+1}_{H,\Eq_{\rho}} K_0][\overline{\alpha
        := \Eq_A}]} (T^{m}_{G, \Eq_{\rho}[\phi :=
      \Eq_{(T^{\set}_{H,\rho})^{n+1} K_0}][\overline{\alpha :=
        \Eq_A}]}\,K_0) \\
&= \lambda \ol{R}. \relsem{\Gamma; \psi,
    \ol\beta \vdash G} \Eq_{\rho} [\phi := T^{n+1}_{H,\Eq_{\rho}}
    K_0][\overline{\alpha := \Eq_A}][\psi := T^{m}_{G, \Eq_{\rho}[\phi
        := \Eq_{(T^{\set}_{H,\rho})^{n+1} K_0}][\overline{\alpha :=
          \Eq_A}]}\,K_0][\overline{\beta := R}] \\
&= \lambda \ol{R}. \relsem{\Gamma; \psi,
    \ol\beta \vdash G} \Eq_{\rho} [\phi := \Eq_{(T^{\set}_{H,\rho})^{n+1}
      K_0}][\overline{\alpha := \Eq_A}][\psi := T^{m}_{G,
      \Eq_{\rho}[\phi := \Eq_{(T^{\set}_{H,\rho})^{n+1}
          K_0}][\overline{\alpha := \Eq_A}]}\,K_0][\overline{\beta :=
      R}] \\
&= T_{G, \Eq_{\rho}[\phi := \Eq_{(T^{\set}_{H,\rho})^{n+1}
        K_0}][\overline{\alpha := \Eq_A}]} (T^{m}_{G, \Eq_{\rho}[\phi
      := \Eq_{(T^{\set}_{H,\rho})^{n+1} K_0}][\overline{\alpha :=
        \Eq_A}]}\,K_0) \\
&= T^{m+1}_{G, \Eq_{\rho}[\phi := \Eq_{(T^{\set}_{H,\rho})^{n+1}
        K_0}][\overline{\alpha := \Eq_A}]}\,K_0
\end{align*}
Here, the second equality holds by the induction hypothesis for
Equation~\ref{eq:helper} on $m$. The fourth equality holds because
$\phi$ and the variables in $\ol\alpha$ do not appear in $G$. (We note
that this case of the IEL would fail to hold if the functorial context
in which the body $F$ of a $\mu$-type $\mu \phi. \lambda
\ol{\alpha}. F$ were formed in Definition~\ref{def:wftypes} were
extended to contain variables other than $\phi$ and the $\alpha$s.
This justifies our design of the $\mu$-type formation rule there.)
\end{itemize}
This concludes the fixpoint case.
\item $\relsem{\Gamma; \Phi \vdash F + G} \Eq_{\rho} =
  \relsem{\Gamma; \Phi \vdash F} \Eq_{\rho} + \relsem{\Gamma;
    \Phi \vdash G} \Eq_{\rho} = \Eq_{\setsem{\Gamma; \Phi \vdash
      F}\rho} + \Eq_{\setsem{\Gamma; \Phi \vdash G}\rho} =
  \Eq_{\setsem{\Gamma; \Phi \vdash F}\rho + \setsem{\Gamma; \Phi
      \vdash G}\rho} = \Eq_{\setsem{\Gamma; \Phi \vdash F +
      G}\rho}$
\item $\relsem{\Gamma; \Phi \vdash F \times G} \Eq_{\rho} =
  \relsem{\Gamma; \Phi \vdash F}\Eq_{\rho} \times \relsem{\Gamma;
    \Phi \vdash G}\Eq_{\rho} = \Eq_{\setsem{\Gamma; \Phi \vdash
      F}\rho} \times \Eq_{\setsem{\Gamma; \Phi \vdash G}\rho}
  = \Eq_{\setsem{\Gamma; \Phi \vdash F}\rho \times
    \setsem{\Gamma; \Phi \vdash G}\rho} = \Eq_{\setsem{\Gamma;
      \Phi \vdash F \times G}\rho}$\qedhere
\end{itemize}

\noindent
With the IEL in hand we can prove a Graph Lemma appropriate to our
setting:
\begin{lem}[Graph Lemma]\label{lem:graph}
If $\rho, \rho' : \setenv$ and $f : \rho \to \rho'$ then
$\graph{\setsem{\Gamma; \Phi \vdash F} f} = \relsem{\Gamma; \Phi
  \vdash F}\graph{f}$.
\end{lem}
\proof
Applying Lemma~\ref{lem:rel-transf-morph} to the morphisms $(f,
\id_{\rho'}) : \graph{f} \to \Eq_{\rho'}$ and $(\id_{\rho}, f) :
\Eq_{\rho} \to \graph{f}$ of relation environments gives
\[\begin{array}{ll}
  & (\setsem{\Gamma; \Phi \vdash F}f, \setsem{\Gamma; \Phi \vdash
  F}\id_{\rho'})\\
= & \relsem{\Gamma; \Phi \vdash F} (f, \id_{\rho'}) :
\relsem{\Gamma; \Phi \vdash F}\graph{f} \to \relsem{\Gamma; \Phi
  \vdash F}\Eq_{\rho'}
\end{array}\] and
\[\begin{array}{ll}
  & (\setsem{\Gamma; \Phi \vdash F}\id_{\rho}, \setsem{\Gamma; \Phi
  \vdash F}f)\\
= & \relsem{\Gamma; \Phi \vdash F} (\id_{\rho}, f) : \relsem{\Gamma;
  \Phi \vdash F}\Eq_{\rho} \to \relsem{\Gamma; \Phi \vdash
  F}\graph{f}
\end{array}\]

Expanding the first equation gives that if $(x,y) \in \relsem{\Gamma;
  \Phi \vdash F}\graph{f}$ then
\[(\setsem{\Gamma; \Phi \vdash F} f\,
x, \setsem{\Gamma; \Phi \vdash F}\id_{\rho'}\, y) \in \relsem{\Gamma;
  \Phi \vdash F}\Eq_{\rho'}\] So $\setsem{\Gamma; \Phi \vdash
  F}\id_{\rho'}\, y = \id_{\setsem{\Gamma; \Phi \vdash F}\rho'}\, y =
y$ and $\relsem{\Gamma; \Phi \vdash F}\Eq_{\rho'} =
\Eq_{\setsem{\Gamma; \Phi \vdash F}\rho'}$, and if $(x,y) \in
\relsem{\Gamma; \Phi \vdash F}\graph{f}$ then $(\setsem{\Gamma; \Phi
  \vdash F} f\, x, y) \in \Eq_{\setsem{\Gamma; \Phi \vdash F}\rho'}$,
i.e., $\setsem{\Gamma; \Phi \vdash F} f\, x = y$, i.e., $(x, y) \in
\graph{\setsem{\Gamma; \Phi \vdash F} f}$.  So, we have
$\relsem{\Gamma; \Phi \vdash F}\graph{f} \subseteq
\graph{\setsem{\Gamma; \Phi \vdash F}f}$.

Expanding the second equation gives that if $x \in \setsem{\Gamma;
  \Phi \vdash F}\rho$ then
\[(\setsem{\Gamma; \Phi \vdash F}\id_{\rho}\, x, \setsem{\Gamma;
  \Phi \vdash F} f\, x) \in \relsem{\Gamma; \Phi \vdash F}\graph{f}\]
Then $\setsem{\Gamma; \Phi \vdash F}\id_{\rho}\, x =
\id_{\setsem{\Gamma; \Phi \vdash F}\rho} x = x$, so for any $x \in
\setsem{\Gamma; \Phi \vdash F}\rho$ we have $(x, \setsem{\Gamma;
  \Phi \vdash F}f\, x)$ $\in \relsem{\Gamma; \Phi \vdash F}\graph{f}$.
Moreover, $x \in \setsem{\Gamma; \Phi \vdash F}\rho$ if and only if
\[(x, \setsem{\Gamma; \Phi \vdash F} f\, x) \in \graph{\setsem{\Gamma;
    \Phi \vdash F}f}\] and, if $x \in \setsem{\Gamma; \Phi \vdash
  F}\rho$ then $(x, \setsem{\Gamma; \Phi \vdash F} f\, x) \in
\relsem{\Gamma; \Phi \vdash F} \graph{f}$. Thus, if \[(x, \setsem{\Gamma;
  \Phi \vdash F} f\, x) \in \graph{\setsem{\Gamma; \Phi \vdash F}f}\]
then $(x, \setsem{\Gamma; \Phi \vdash F} f\, x) \in \relsem{\Gamma;
  \Phi \vdash F} \graph{f}$, i.e., $\graph{\setsem{\Gamma; \Phi \vdash
    F}f} \subseteq \relsem{\Gamma; \Phi \vdash F} \graph{f}$.  \qed

\section{Interpreting Terms}\label{sec:term-interp}

If $\Delta = x_1 : F_1,\dots,x_n : F_n$ is a term context for $\Gamma$
and $\Phi$, then the interpretations $\setsem{\Gamma;\Phi \vdash
  \Delta}$ and $\relsem{\Gamma;\Phi \vdash \Delta}$ are defined by
\[\begin{array}{lll}
\setsem{\Gamma;\Phi \vdash \Delta} & = & \setsem{\Gamma;\Phi \vdash
  F_1} \times \cdots \times \setsem{\Gamma;\Phi \vdash F_n}\\
\relsem{\Gamma;\Phi \vdash \Delta} & = & \relsem{\Gamma;\Phi \vdash
  F_1} \times \cdots \times \relsem{\Gamma;\Phi \vdash F_n}\\
\end{array}\]
Every well-formed term $\Gamma;\Phi~|~\Delta \vdash t : F$ then has,
for every $\rho \in \setenv$, a set interpretation
$\setsem{\Gamma;\Phi~|~\Delta \vdash t : F}\rho$ as the component at
$\rho$ of a natural transformation from $\setsem{\Gamma; \Phi \vdash
  \Delta}$ to $\setsem{\Gamma; \Phi \vdash F}$, and, for every $\rho
\in \relenv$, a relational interpretation
$\relsem{\Gamma;\Phi~|~\Delta \vdash t : F}\rho$ as the component at
$\rho$ of a natural transformation from $\relsem{\Gamma; \Phi \vdash
  \Delta}$ to $\relsem{\Gamma; \Phi \vdash F}$.

\begin{defi}\label{def:set-interp}
If $\rho$ is a set (resp., relation) environment and
$\Gamma;\Phi~|~\Delta \vdash t : F$ then
$\setsem{\Gamma;\Phi~|~\Delta \vdash t : F}\rho$ (resp.,
$\relsem{\Gamma;\Phi~|~\Delta \vdash t : F}\rho$) is defined as in
Figure~\ref{fig:term-sem}, where $\mathsf D$ is either $\set$ or
$\rel$ as appropriate.
\end{defi}

If $t$ is closed, i.e., if $\emptyset; \emptyset~|~\emptyset \vdash t
: F$, then we write $\dsem{\vdash t : F}$ instead of $\dsem{\emptyset;
  \emptyset~|~\emptyset \vdash t : F}$.  The interpretations in
Definition~\ref{def:set-interp} respect weakening, i.e., a term and
its weakenings all have the same set and relational
interpretations. Specifically, for any $\rho \in \setenv$,
\[\setsem{\Gamma;\Phi \,|\, \Delta, x : F \vdash t : G}\rho =
(\setsem{\Gamma;\Phi \,|\, \Delta \vdash t : G}\rho) \circ
\pi_{\Delta}\] where $\pi_{\Delta}$ is the projection
$\setsem{\Gamma;\Phi \vdash \Delta, x : F} \to
\setsem{\Gamma;\Phi \vdash \Delta}$. A similar result holds for
relational interpretations.

\begin{figure*}
\begin{adjustbox}{varwidth=7.2in, max width=5.8in, fbox, center}
\[\begin{array}{lll}
\dsem{\Gamma;\Phi \,|\, \Delta,x :F \vdash x : F} \rho& = &
\pi_{|\Delta|+1}\\
\dsem{\Gamma;\emptyset \,|\, \Delta \vdash L_{\overline \alpha} x.t :
  \Nat^{\overline
    \alpha} \,F \,G}\rho & = &  \curry (\dsem{\Gamma;\overline \alpha
  \,|\, \Delta, x : F \vdash t: G}\rho[\overline{\alpha := \_}])\\
\dsem{\Gamma;\Phi \,|\, \Delta \vdash t_{\overline K} s:
  G [\overline{\alpha := K}]}\rho & = & \eval \circ \langle
  \lambda d.\,(\dsem{\Gamma;\emptyset \,|\, \Delta \vdash t :
  \Nat^{\overline{\alpha}} \,F \,G}\rho\; d)_{\overline{\dsem{\Gamma;\Phi
      \vdash K}\rho}},\\
 & & \hspace*{0.5in} \dsem{\Gamma;\Phi \,|\,
    \Delta \vdash s: F [\overline{\alpha := K}]}\rho \rangle\\
& & \\
\dsem{\Gamma;\Phi \,|\, \Delta \vdash \bot_F t : F} \rho& = &
!^0_{\dsem{\Gamma;\Phi \vdash F}\rho} \circ
  \dsem{\Gamma;\Phi~|~\Delta \vdash t : \zerot}\rho, \mbox{ where } \\
 & & \hspace*{0.1in} !^0_{\dsem{\Gamma;\Phi \vdash F}\rho}
\mbox{ is the unique morphism from } 0\\
 & & \hspace*{0.1in} \mbox{ to } \dsem{\Gamma;\Phi \vdash F}\rho\\
\dsem{\Gamma;\Phi \,|\, \Delta \vdash \top : \onet}\rho & = &
!^{\dsem{\Gamma;\Phi\vdash \Delta}\rho}_1, \mbox{ where }
!^{\dsem{\Gamma;\Phi\vdash \Delta}\rho}_1\\
& & \hspace*{0.1in} \mbox{ is the unique morphism from }
\dsem{\Gamma;\Phi\vdash \Delta}\rho \mbox{ to } 1\\
\dsem{\Gamma;\Phi \,|\, \Delta \vdash (s,t) : F \times G} \rho& = &
\langle \dsem{\Gamma;\Phi \,|\, \Delta \vdash s: F} \rho,\,
\dsem{\Gamma;\Phi \,|\, \Delta \vdash t : G} \rho\rangle\\
\dsem{\Gamma;\Phi \,|\, \Delta \vdash \pi_1 t : F} \rho& = &
\pi_1 \circ \dsem{\Gamma;\Phi \,|\, \Delta \vdash t : F \times G}\rho\\
\dsem{\Gamma;\Phi \,|\, \Delta \vdash \pi_2 t : G}\rho & = &
\pi_2 \circ \dsem{\Gamma;\Phi \,|\, \Delta \vdash t : F \times
  G} \rho\\
\dsem{\Gamma;\Phi~|~\Delta \vdash \cse{t}{x \mapsto l}{y \mapsto r} :
  K}\rho & = & \eval \circ \langle \curry \,[\dsem{\Gamma;\Phi
    \,|\, \Delta, x : F \vdash l : K}\rho,\\
   & & \hspace*{0.79in} \dsem{\Gamma;\Phi \,|\, \Delta, y
    : G \vdash r : K}\rho],\\
   & &  \hspace*{0.5in} \dsem{\Gamma;\Phi \,|\, \Delta \vdash t :
  F + G} \rho\rangle\\
\dsem{\Gamma;\Phi \,|\, \Delta \vdash \inl \,s: F + G} \rho& = &
\inl \circ \dsem{\Gamma;\Phi \,|\, \Delta \vdash s: F}\rho\\
\dsem{\Gamma;\Phi \,|\, \Delta \vdash \inr \,t: F + G}\rho & = &
\inr \circ \dsem{\Gamma;\Phi \,|\, \Delta \vdash t : G}\rho\\
\llbracket \Gamma;\emptyset \,|\, \emptyset \vdash \map^{\ol{F},\ol{G}}_H
  : \Nat^\emptyset (\Nat^{\ol\beta,\ol\gamma} F\,G)
& = & \lambda d\, \ol\eta\,\ol{C}.\,
\dsem{\Gamma; \ol{\phi},\ol{\gamma}\vdash H}\id_{\rho[\ol{\gamma:=
      C}]}[\ol{\phi := \lambda \ol B. \eta_{\ol B\, \ol C}}]\\
\hspace*{0.5in}
  (\Nat^{\ol{\gamma}}\,H[\ol{\phi :=_{\ol{\beta}} F}]\,H[\ol{\phi
      :=_{\ol{\beta}} G}]) \rrbracket^{\mathsf D} \rho & &\\
\llbracket \Gamma;\emptyset \,|\, \emptyset \vdash \tin_H :
Nat^{\ol{\beta}} \, H[\phi := (\mu \phi.\lambda {\overline
    \alpha}.H){\overline \beta}][\ol{\alpha := \beta}] & = &
\lambda
d.\,\mathit{in}_{{T}^X_{H,\rho}} \;\;\mbox{ where } X \mbox{ is } \set \mbox{ when } \\
\hspace*{0.79in}(\mu \phi.\lambda {\overline \alpha}.H){\overline
  \beta} \rrbracket^{\mathsf{D}} \rho & & \hspace*{0.2in}
\mathsf{D} = \set \mbox{ and not present when }
\mathsf{D} = \rel\\
\llbracket \Gamma;\emptyset \,|\, \emptyset \vdash
  \fold^F_H : \Nat^\emptyset\;(\Nat^{\ol{\beta}}\,H[\phi
    :=_{\ol{\beta}} F][\ol{\alpha := \beta}]\,F) & = &
\lambda d. \,\mathit{fold}_{T^X_{H,\rho}} \\
\hspace*{0.79in}(\Nat^{{\ol{\beta}} }\,(\mu
  \phi.\lambda \overline \alpha.H)\overline \beta\;F)
\rrbracket^{\mathsf{D}} \rho & & \hspace*{0.2in} \mbox{ where } X \mbox{ is as above
}\vspace*{-0.1in}
\end{array}\]
\caption{Term semantics}\label{fig:term-sem}
\end{adjustbox}\vspace*{-0.05in}
\end{figure*}

The return type for the semantic fold is
  $\dsem{\Gamma;\ol\beta \vdash F}\rho[\ol{\beta := B}]$.
This interpretation gives
\[\dsem{\Gamma;\emptyset \,|\, \Delta
  \vdash L_\empty x. t : \Nat^\emptyset F\, G}\rho = \curry
  (\dsem{\Gamma;\emptyset \,|\, \Delta, x : F \vdash t :
  G}\rho)\]
and
\[\dsem{\Gamma;\emptyset \,|\, \Delta \vdash st:
    G} \rho = \eval \circ \langle \dsem{\Gamma;\emptyset \,|\, \Delta
  \vdash s: \Nat^\emptyset F\, G}\rho, \dsem{\Gamma;\emptyset \,|\, \Delta \vdash
  t: F}\rho \rangle\]

\vspace*{0.05in}

\noindent
so it specializes to the standard interpretations of those System F
terms that are representable in our calculus.  Term interpretation
also respects substitution for both functorial and non-functorial type
variables, as well as term substitution. That is, if
$\Gamma,\alpha;\Phi \,|\, \Delta \vdash t : F$ and $\Gamma;\Phi,\alpha
\,|\, \Delta \vdash t' : F$ and $\Gamma;\Phi \vdash G$ then
\[\dsem{\Gamma;\Phi \,|\, \Delta[\alpha := G] \vdash t[\alpha :=
    G] : F[\alpha := G]}\rho = \dsem{\Gamma,\alpha;\Phi \,|\, \Delta
  \vdash t : F }\rho [ \alpha := \dsem{\Gamma;\Phi\vdash G}\rho ]\]
and
\[\dsem{\Gamma;\Phi \,|\, \Delta[\alpha := G] \vdash t'[\alpha :=
    G] : F[\alpha := G]}\rho = \dsem{\Gamma;\Phi,\alpha \,|\, \Delta
  \vdash t' : F }\rho [ \alpha := \dsem{\Gamma;\Phi\vdash G}\rho ]\]

\vspace*{0.05in}

\noindent
and if $\Gamma;\Phi \,|\, \Delta, x: G \vdash t : F$ and $\Gamma;\Phi
\,|\, \Delta \vdash s : G$ then
\[\lambda A.\, \dsem{\Gamma;\Phi \,|\, \Delta \vdash t[x := s] :
  F }\rho \, A = \lambda A.\,\dsem{\Gamma;\Phi \,|\, \Delta, x: G
  \vdash t : F}\rho \,(A, \dsem{\Gamma;\Phi \,|\, \Delta\vdash s:
  G}\rho\, A)\]

\vspace*{0.05in}

\noindent
Direct calculation further reveals that term interpretations also
satisfy
\[\dsem{\Gamma; \Phi~|~\Delta \vdash
  (L_{\ol{\alpha}}x.t)_{\ol{K}}s} = \dsem{\Gamma; \Phi~|~\Delta \vdash
  t [\ol{\alpha := K}][x := s]}\]
and
\[\dsem{\Gamma; \emptyset~|~\Delta \vdash L_{\ol{\alpha}} x.t_{\ol{\alpha}} x :
   \Nat^{\ol{\alpha}}\, F\, G} = \dsem{\Gamma; \emptyset~|~\Delta
  \vdash t : \Nat^{\ol{\alpha}}\, F\, G}\]


\section{Naturality and the Abstraction Theorem}\label{sec:abs-and-nat}

\subsection{Naturality and Its Consequences}\label{sec:Nat-type-terms}

We first show that terms of $\Nat$-type behave as natural
transformations with respect to their source and target functorial
types, and derive some conseqences of this observation. If $\Gamma;
\ol{\alpha} \vdash F$ then define the \emph{identity} $\id_F$ on $F$ by
$\id_F = \Gamma;\emptyset~|~\emptyset \vdash L_{\ol{\alpha}}x.x :
\Nat^{\ol{\alpha}} F\,F$. Also, recall from Section~\ref{sec:terms} that if
$\Gamma; \emptyset \,|\, \Delta \vdash t: \Nat^{\overline{\alpha}}
F\,G$ and $\Gamma; \emptyset \,|\, \Delta \vdash s:
\Nat^{\overline{\alpha}} G\,H$ are terms then the composition $s \circ
t$ of $t$ and $s$ is defined by $s \circ t = \Gamma; \emptyset\,|\,
\Delta \vdash L_{\overline{\alpha}}
x. s_{\overline{\alpha}}(t_{\overline{\alpha}}x):
\Nat^{\overline{\alpha}} F\,H$. Then \[\setsem{\Gamma; \emptyset \,|\,
  \emptyset \vdash \id_{F} : \Nat^{\ol{\alpha}} F\,F} \rho\, \ast =
\id_{\lambda \ol{A}. \setsem{\Gamma; \ol{\alpha} \vdash F} \rho
  [\ol{\alpha := A}]}\] for any set environment $\rho$, and
\[\setsem{\Gamma; \emptyset \,|\, \Delta \vdash s \circ t:
  \Nat^{\overline{\alpha}} F\,H} = \setsem{\Gamma; \emptyset \,|\,
  \Delta \vdash s: \Nat^{\overline{\alpha}} G\,H} \circ
\setsem{\Gamma; \emptyset \,|\, \Delta \vdash t:
  \Nat^{\overline{\alpha}} F\,G}\]
Naturality of term interpretations is then easily verified:
\begin{thm}\label{eq:ft-from-nat}
  If\, $\Gamma; \emptyset \,|\, \Delta \vdash s : \Nat^{\overline{\alpha},
  \overline{\gamma}} F\,G$ and $\overline{\Gamma; \emptyset \,|\,
  \Delta \vdash t : \Nat^{\overline{\gamma}} K\, H}$, then
  \[\begin{array}{l}
  \hspace*{0.135in}
\setsem{\Gamma; \emptyset\,|\, \Delta
  \vdash
  ((\map_G^{\overline{K}, \overline{H}})_\emptyset \,\overline{t}) \circ
(L_{\overline{\gamma}} z. s_{\overline{K}, \overline{\gamma}}
  z)
  : \Nat^{\overline{\gamma}} F[\overline{\alpha := K}]\,
  G[\overline{\alpha := H}]}\\
= \setsem{ \Gamma; \emptyset \,|\, \Delta \vdash
(L_{\overline{\gamma}} z.
  s_{\overline{H}, \overline{\gamma}} z)\circ
  ((\map_F^{\overline{K}, \overline{H}})_\emptyset \,
  \overline{t})  : \Nat^{\overline{\gamma}} F[\overline{\alpha :=
      K}]\, G[\overline{\alpha := H}]}
\end{array}\]
\end{thm}
\noindent

Theorem~\ref{eq:ft-from-nat} gives rise to an entire family of free
theorems that are consequences of naturality, and thus do not require
the full power of parametricity. In particular,
Definition~\ref{def:set-interp} alone ensures that standard properties
of the initial algebraic constructs $\map$, $\tin$, and $\fold$
hold. We have, for example, that the interpretation of every $\map_H$
is a functor, i.e., if $\Gamma; \ol\alpha, \ol\gamma \vdash H$,
$\ol{\Gamma; \emptyset \,|\, \Delta \vdash g : \Nat^{\ol\gamma} F \,
  G}$, and $\ol{\Gamma; \emptyset \,|\, \Delta \vdash f :
  \Nat^{\ol\gamma} G \, K}$, then
  \begin{align*}
&\setsem{
\Gamma; \emptyset \,|\, \Delta \vdash
  (\map^{\ol{F}, \ol{K}}_H \,)_{\emptyset} \, \ol{(f \circ g)}
  : \Nat^{\ol\gamma} H[\ol{\alpha := F}] \, H[\ol{\alpha := K}]} \\
= \hspace{0.03in}
  &\setsem{
  \Gamma; \emptyset \,|\, \Delta \vdash
  (\map^{\ol{G}, \ol{K}}_H \,)_{\emptyset} \, \ol{f} \circ
  (\map^{\ol{F}, \ol{G}}_H \,)_{\emptyset} \, \ol{g}
  : \Nat^{\ol\gamma} H[\ol{\alpha := F}] \, H[\ol{\alpha := K}]}
\end{align*}
In fact, we have that if $\Gamma; \ol{\alpha} \vdash H$, $\ol{\Gamma;
  \emptyset \vdash F}$, and $\ol{\Gamma; \emptyset \vdash G}$, then,
for all $\ol{f \in \setsem{\Gamma;\emptyset \vdash \Nat^\emptyset
    F\,G}\rho}$,
\[\begin{array}{rl}
&\setsem{ \Gamma;\emptyset\,|\,\ol{x : \Nat^\emptyset F \, G}
  \vdash (\map^{\ol{F},\ol{G}}_H)_\emptyset\,\ol x :
  \Nat^{\emptyset}\,H[\ol{\alpha := F}]\,H[\ol{\alpha := G}]}
\rho\,\ol f\\ =& \setsem{\Gamma; \ol{\alpha}\vdash H}
\id_{\rho}[\ol{\alpha := f}] \\ =&
\mathit{map}_{\lambda \ol{A}.\,\setsem{\Gamma; \ol{\alpha}\vdash H} \rho [\ol{\alpha := A}]} \ol{f} \\
\end{array}\]
Here, we obtain the first equality from the appropriate instance of
Definition~\ref{def:set-interp}, and the second one by noting that
$\lambda \ol{A}.\,\setsem{\Gamma; \ol{\alpha}\vdash H} \rho
[\ol{\alpha := A}]$ is a functor in $\ol{A}$ and using
$\mathit{map}_G\,\ol f$ to denote the action of the semantic functor
$G$ on morphisms $\ol f$. We also have, e.g., that:

\vspace*{0.1in}

\noindent
$\bullet$\; $\map$ is a higher-order functor, i.e., if $\Gamma;
\ol{\psi}, \ol{\gamma} \vdash H$,\; $\ol{\Gamma; \ol{\alpha},
  \ol{\gamma}, \ol{\phi} \vdash K}$,\; $\ol{\Gamma; \ol{\beta},
  \ol{\gamma} \vdash F}$,\; $\ol{\Gamma; \ol{\beta}, \ol{\gamma}
  \vdash G}$, and $\xi = \Nat^{\emptyset} (\ol{\Nat^{\ol{\alpha},
    \ol{\beta}, \ol{\gamma}} F\, G})\, (\Nat^{\ol{\gamma}} H[\ol{\psi
    := K}][\ol{\phi := F}]\, H[\ol{\psi := K}][\ol{\phi := G}])$, then

\vspace*{-0.08in}

\[\setsem{\Gamma; \emptyset \,|\, \emptyset \vdash
\map_{H[\ol{\psi := K}]}^{\ol{F}, \ol{G}} : \xi} = \setsem{\Gamma;
  \emptyset \,|\, \emptyset \vdash L_\emptyset \ol
  z. (\map_H^{\ol{K[\ol{\phi := F}]}, \ol{K[\ol{\phi :=
          G}]}})_\emptyset \,(\ol{(\map_K^{\ol{F}, \ol{G}})_\emptyset
    \, \ol z}) : \xi}\]

\vspace*{0.1in}

\noindent
$\bullet$\; $\fold$ behaves as expected, i.e., if $\xi =
\Nat^{\overline{\beta}}\, H[\phi :=
  (\mu \phi. \lambda
  \overline{\alpha}. H)\overline{\beta}][\overline{\alpha := \beta}]\;
F$, then

\vspace*{-0.05in}

\[\hspace*{-1in}\setsem{\Gamma; \emptyset \,|\, x:
  \Nat^{\overline{\beta}} H[\phi :=
    F][\overline{\alpha := \beta}]\, F \vdash ((\fold_{H,
    F})_{\emptyset} x) \circ \tin_{H} : \xi}\]

\vspace*{-0.15in}

\[= \setsem{ \Gamma; \emptyset \,|\, x: \Nat^{\overline{\beta}} H[\phi
    := F][\overline{\alpha := \beta}]\, F
  \vdash x \circ \big( (\map_{H [\ol{\alpha := \beta}]}^{(\mu
    \phi. \lambda \overline{\alpha} H) \overline{\beta},
    F})_\emptyset ((\fold_{H, F})_{\emptyset} x) \big) : \xi}\]

\vspace*{0.1in}

\noindent
$\bullet$\; $\tin$ has a right inverse, i.e., if $\xi =
\Nat^{\overline{\beta}}\, (\mu \phi. \lambda
\overline{\alpha}. H)\overline{\beta}\; (\mu \phi. \lambda
\overline{\alpha}. H)\overline{\beta}$, then

\vspace*{-0.05in}

\[\hspace*{0.43in}\setsem{\Gamma; \emptyset \,|\, \emptyset
  \vdash \tin_H \circ (\fold_{H, H[\phi := (\mu \phi. \lambda
      \overline{\alpha}. H)\overline{\beta}]})_{\emptyset}
  ((\map_{H}^{H[\phi := (\mu \phi. \lambda
      \overline{\alpha}. H)\overline{\beta}][\overline{\alpha :=
        \beta}], (\mu \phi. \lambda
    \overline{\alpha}. H)\overline{\beta}})_\emptyset \tin_H) :
  \xi}\]

\vspace*{-0.15in}

\[\hspace*{-3.46in} = \setsem{\Gamma; \emptyset \,|\, \emptyset \vdash \mathit{id}_{(\mu  \phi. \lambda \overline{\alpha}. H)\overline{\beta}} : \xi}\]

\vspace*{0.15in}

\noindent
$\bullet$\; $\tin$ has a left inverse, i.e., if $\xi =
\Nat^{\overline{\beta}} H[\phi := (\mu
  \phi. \lambda \overline{\alpha}. H)\overline{\beta}]\, H[\phi :=
  (\mu \phi. \lambda \overline{\alpha}. H)\overline{\beta}]$, then

\vspace*{-0.05in}

\[\hspace*{0.41in}\setsem{\Gamma; \emptyset \,|\, \emptyset \vdash
(\fold_{H, H[\phi := (\mu \phi. \lambda
      \overline{\alpha}. H)\overline{\beta}]})_{\emptyset}
  ((\map_{H}^{H[\phi := (\mu \phi. \lambda
      \overline{\alpha}. H)\overline{\beta}][\overline{\alpha :=
        \beta}], (\mu \phi. \lambda
    \overline{\alpha}. H)\overline{\beta}})_\emptyset \tin_H) \circ
  \tin_H : \xi}\]

\vspace*{-0.15in}

\[\hspace*{-3.15in} =\setsem{ \Gamma; \emptyset \,|\, \emptyset \vdash
  \mathit{id}_{H[\phi := (\mu \phi. \lambda
      \overline{\alpha}. H)\overline{\beta}]} : \xi}\]

\vspace*{0.1in}

\noindent
Analogous results hold for relational interpretations of terms and
relation environments.

\subsection{The Abstraction Theorem}\label{sec:thms}

To get consequences of parametricity that are not merely consequences
of naturality, we prove an Abstraction Theorem
(Theorem~\ref{thm:abstraction}) for our calculus. As is usual for such
theorems, we first prove a more general result
(Theorem~\ref{thm:at-gen}) for open terms, and recover our Abstraction
Theorem as a special case of it for closed terms of closed type.

\begin{thm}\label{thm:at-gen}
Every well-formed term $\Gamma;\Phi~|~\Delta \vdash t : F$ induces
a morphism from $\sem{\Gamma; \Phi \vdash \Delta}$ to
$\sem{\Gamma; \Phi \vdash F}$, i.e., a triple of natural
transformations
\[(\setsem{\Gamma;\Phi~|~\Delta \vdash t : F},\,
\setsem{\Gamma;\Phi~|~\Delta \vdash t : F},\,
\relsem{\Gamma;\Phi~|~\Delta \vdash t : F})\]
where, for $\mathsf{D} \in \{\set,\rel\}$ and for $\rho \in \setenv$
or $\rho \in \relenv$ as appropriate,
\[\begin{array}{lll}
\dsem{\Gamma;\Phi~|~\Delta \vdash t : F} & : & \dsem{\Gamma;
  \Phi \vdash \Delta} \to \dsem{\Gamma; \Phi \vdash F}
\end{array}\]
has as its component at $\rho : \setenv$ a morphism
\[\begin{array}{lll}
\setsem{\Gamma;\Phi~|~\Delta \vdash t : F}\rho & : & \setsem{\Gamma;
  \Phi \vdash \Delta}\rho \to \setsem{\Gamma; \Phi \vdash F}\rho
\end{array}\]
Moreover, for all $\rho : \relenv$,
\begin{equation}\label{eq:projs}
\relsem{\Gamma;\Phi~|~\Delta \vdash t : F}\rho =
(\setsem{\Gamma;\Phi~|~\Delta \vdash t : F}(\pi_1 \rho),\,
\setsem{\Gamma;\Phi~|~\Delta \vdash t : F}(\pi_2 \rho))
\end{equation}
\end{thm}
\proof The proof is by induction on $t$. It requires showing that set
and relational interpretations of term judgments are natural
transformations, and that all set interpretations of terms of $\Nat$
types satisfy the appropriate equality preservation conditions from
Definition~\ref{def:set-sem}.  For the interesting cases of
abstraction, application, $\map$, $\tin$, and $\mathsf{fold}$ terms,
propagating the naturality conditions is quite involved; the latter
two especially require some rather delicate diagram chasing. That it
is possible provides strong evidence that our development is sensible,
natural, and at an appropriate level of abstraction.

The only interesting cases are the cases for abstraction, application,
$\map$, $\tin$, and $\fold$. We omit the others.

\begin{itemize}
\item $\underline{\Gamma; \emptyset \,|\, \Delta \vdash
  L_{\overline{\alpha}} x.t : \Nat^{\overline{\alpha}} \,F \,G}$ \; To
  see that $\setsem{\Gamma; \emptyset \,|\, \Delta \vdash
    L_{\overline{\alpha}} x.t : \Nat^{\overline{\alpha}} \,F \,G}$ is
  a natural transformation from $\setsem{\Gamma; \emptyset \vdash
    \Delta}$ to $\setsem{\Gamma; \emptyset \vdash
    \Nat^{\overline{\alpha}} \,F \,G}$ we need to show that, for every
  $\rho : \setenv$, $\setsem{\Gamma; \emptyset \,|\, \Delta \vdash
    L_{\overline{\alpha}} x.t : \Nat^{\overline{\alpha}} \,F \,G}\rho$
  is a morphism in $\set$ from $\setsem{\Gamma; \emptyset \vdash
    \Delta}\rho$ to $\setsem{\Gamma; \emptyset \vdash
    \Nat^{\overline{\alpha}} \,F \,G}\rho$, and that such a family of
  morphisms is natural.  First, we need to show that, for all $\ol{A :
    \set}$ and all $d : \setsem{\Gamma; \emptyset \vdash \Delta}\rho =
  \setsem{\Gamma; \ol{\alpha} \vdash \Delta}\rho[\overline{\alpha :=
      A}]$, we have that
\[(\setsem{\Gamma; \emptyset \,|\, \Delta \vdash L_{\overline{\alpha}}
      x.t : \Nat^{\overline{\alpha}} \,F \,G}\rho\,d)_{\ol{A}} :
\setsem{\Gamma; \ol{\alpha} \vdash F}\rho[\overline{\alpha := A}]
    \to \setsem{\Gamma; \ol{\alpha} \vdash G}\rho[\overline{\alpha :=
        A}]\]
but this follows easily from the induction hypothesis.
That these maps comprise a natural transformation $\eta :
\setsem{\Gamma; \ol{\alpha} \vdash F}\rho[\overline{\alpha := \_}] \to
\setsem{\Gamma; \ol{\alpha} \vdash G}\rho[\overline{\alpha := \_}]$ is
clear since \[\eta_{\ol{A}} \, = \, \curry\,
(\setsem{\Gamma; \overline{\alpha} \,|\, \Delta, x : F \vdash t:
  G}\rho[\overline{\alpha := A}])\,d\] is the component at $\ol{A}$ of
the partial specialization to $d$ of the natural transformation
\[\setsem{\Gamma; \overline{\alpha} \,|\, \Delta, x : F \vdash t:
  G}\rho[\overline{\alpha := \_}]\]  To see that the components of
$\eta$ also satisfy the additional condition needed for $\eta$ to
be in $\setsem{\Gamma; \emptyset \vdash \Nat^{\overline{\alpha}} \,F
  \,G}\rho$, let $\overline{R : \rel(A, B)}$ and suppose
\[\begin{array}{lll}
(u, v) &  \in & \relsem{\Gamma;\overline{\alpha} \vdash F}
\Eq_{\rho}[\overline{\alpha := R}]\\
&  = & (\setsem{\Gamma;\overline{\alpha} \vdash F}
\rho[\overline{\alpha := A}], \\
  & &\,\,\setsem{\Gamma;\overline{\alpha} \vdash
  F} \rho[\overline{\alpha := B}],
(\relsem{\Gamma;\overline{\alpha} \vdash F}
\Eq_{\rho}[\overline{\alpha := R}])^*)
\end{array}\]
Then the induction hypothesis and
$(d,d) \in \relsem{\Gamma; \emptyset \vdash \Delta} \Eq_\rho =
\relsem{\Gamma; \emptyset \vdash \Delta} \Eq_\rho[\ol{\alpha := R}]$
ensure that
\[\begin{array}{ll}
& (\eta_{\ol{A}}u,\eta_{\ol{B}}v)\\
= & (\curry\, (\setsem{\Gamma; \overline{\alpha} \,|\, \Delta, x : F
  \vdash t: G}\rho[\overline{\alpha := A}])\,d\,u, \\
  &\hspace{0.5in} \curry\,
(\setsem{\Gamma; \overline{\alpha} \,|\, \Delta, x : F \vdash t:
  G}\rho[\overline{\alpha := B}])\,d\,v)\\
= & \curry\, (\relsem{\Gamma; \overline{\alpha} \,|\, \Delta, x : F
  \vdash t: G}\Eq_\rho[\overline{\alpha := R}])\,(d,d)\,(u,v)\\
: & \relsem{\Gamma; \overline{\alpha} \vdash G}
\Eq_{\rho}[\overline{\alpha := R}]
\end{array}\]
Moreover,
$\setsem{\Gamma; \emptyset \,|\, \Delta \vdash L_{\overline{\alpha}} x.t
: \Nat^{\overline{\alpha}} \,F \,G} \rho$
is trivially natural in $\rho$,
as the functorial action of
$\setsem{\Gamma; \emptyset \vdash \Delta}$
and $\setsem{\Gamma; \emptyset \vdash \Nat^{\overline{\alpha}} \,F \,G}$
on morphisms is the identity.

\item
$\underline{\Gamma;\Phi \,|\, \Delta \vdash t_{\overline K} s: G
  [\overline{\alpha := K}]}$\;
  To see that $\setsem{\Gamma;\Phi \,|\,
  \Delta \vdash t_{\overline K} s: G [\overline{\alpha := K}]}$ is a
  natural transformation from $\setsem{\Gamma; \Phi \vdash \Delta}$ to
  $\setsem{\Gamma;\Phi \vdash G [\overline{\alpha := K}]}$ we
  show that, for every $\rho : \setenv$, \[\setsem{\Gamma;\Phi \,|\,
    \Delta \vdash t_{\overline K} s: G [\overline{\alpha := K}]}\rho\]
  is a morphism from $\setsem{\Gamma; \Phi \vdash \Delta}\rho$ to
  $\setsem{\Gamma;\Phi \vdash G [\overline{\alpha := K}]}\rho$, and
  that this family of morphisms is natural in $\rho$. Let $d :
  \setsem{\Gamma; \Phi \vdash \Delta}\rho$. Then
  \[\begin{array}{ll}
  & \setsem{\Gamma;\Phi \,|\, \Delta \vdash t_{\overline K} s: G
  [\overline{\alpha := K}]}\,\rho\,d\\
= & (\eval \circ \langle (\setsem{\Gamma; \emptyset \,|\, \Delta \vdash
  t : \Nat^{\overline{\alpha}} \,F \,G}\rho\;
\_)_{\overline{\setsem{\Gamma;\Phi \vdash K}\rho}},\,
\setsem{\Gamma;\Phi \,|\, \Delta \vdash s: F [\overline{\alpha :=
      K}]}\rho \rangle)\,d\\
= & \eval ((\setsem{\Gamma; \emptyset \,|\, \Delta \vdash t :
  \Nat^{\overline{\alpha}} \,F \,G}\rho\;
\_)_{\overline{\setsem{\Gamma;\Phi \vdash K}\rho}} \,d,\,
\setsem{\Gamma;\Phi \,|\, \Delta \vdash s: F [\overline{\alpha :=
      K}]}\rho\, d)\\
= & \eval ((\setsem{\Gamma; \emptyset \,|\, \Delta \vdash t :
  \Nat^{\overline{\alpha}} \,F \,G}\rho\;
d)_{\overline{\setsem{\Gamma;\Phi \vdash K}\rho}},\,
\setsem{\Gamma;\Phi \,|\, \Delta \vdash s: F [\overline{\alpha :=
      K}]}\rho\, d)\\
\end{array}\]
The induction hypothesis ensures that \[(\setsem{\Gamma; \emptyset \,|\,
  \Delta \vdash t : \Nat^{\overline{\alpha}} \,F \,G}\rho\,
d)_{\overline{\setsem{\Gamma;\Phi \vdash K}\rho}}\] has type
$\setsem{\Gamma; \ol{\alpha} \vdash F}\rho[\ol{\alpha :=
    \setsem{\Gamma;\Phi \vdash K}\rho}] \to \setsem{\Gamma;
  \ol{\alpha} \vdash G}\rho[\ol{\alpha := \setsem{\Gamma;\Phi \vdash
      K}\rho}]$.  Since, in addition,
\[\begin{array}{ll}
  &\setsem{\Gamma;\Phi \,|\,
  \Delta \vdash s: F [\overline{\alpha := K}]}\rho\, d :
\setsem{\Gamma; \Phi \vdash F[\ol{\alpha := K}]}\rho \\
  =&\setsem{\Gamma; \Phi, \ol{\alpha} \vdash F}\rho[\ol{\alpha :=
    \setsem{\Gamma;\Phi \vdash K}\rho}] \\
    =& \setsem{\Gamma;
  \ol{\alpha} \vdash F}\rho[\ol{\alpha := \setsem{\Gamma;\Phi \vdash
      K}\rho}]
\end{array}\]
      we have that
\[\begin{array}{ll}
  &\setsem{\Gamma;\Phi \,|\, \Delta
  \vdash t_{\overline K} s: G [\overline{\alpha :=
      K}]}\,\rho\,d : \setsem{\Gamma; \Phi ,\ol{\alpha} \vdash
  G}\rho[\ol{\alpha := \setsem{\Gamma;\Phi \vdash K}\rho}]  \\
  =&\setsem{\Gamma; \Phi \vdash G[\ol{\alpha := K}]}\rho
\end{array}\]
as desired.

\vspace*{0.1in}

To see that the family of maps comprising $\setsem{\Gamma;\Phi \,|\,
  \Delta \vdash t_{\overline K} s: G [\overline{\alpha := K}]}$
is natural in $\rho$
we need to show that, if $f : \rho \to \rho'$ in $\setenv$, then the
following diagram commutes, where $g = \setsem{\Gamma; \emptyset \,|\,
  \Delta \vdash t : \Nat^{\overline{\alpha}} \,F \,G}$ and $h =
\setsem{\Gamma;\Phi \,|\, \Delta \vdash s: F [\overline{\alpha :=
      K}]}$:
{\footnotesize
\[\hspace*{-0.6in}
\begin{tikzcd}
\setsem{\Gamma;\Phi \vdash \Delta}\rho \ar[r, "{\setsem{\Gamma;\Phi
  \vdash \Delta}f}"] \ar[d, "{\langle g \rho, h \rho\rangle}"']
& \setsem{\Gamma;\Phi \vdash
  \Delta}\rho' \ar[d, "{\langle g \rho', h \rho' \rangle}"]\\
\setsem{\Gamma;\emptyset \vdash \Nat^{\overline{\alpha}} \,F \,G}\rho
\times \setsem{\Gamma;\Phi \vdash F [\overline{\alpha := K}]}\rho
\ar[d, "{\eval \circ ((-)_{\overline{\sem{\Gamma;\Phi \vdash K}\rho}} \times
    \id)}"']
\ar[r, bend left = 5, "{\setsem{\Gamma;\emptyset\vdash
      \Nat^{\overline{\alpha}} \,F \,G}f\, \times\, \setsem{\Gamma;\Phi
      \vdash F [\overline{\alpha := K}]}f}"] &
\setsem{\Gamma;\emptyset \vdash \Nat^{\overline{\alpha}} \,F \,G}\rho'
\times \setsem{\Gamma;\Phi \vdash F [\overline{\alpha := K}]}\rho'
\ar[d, "{\eval \circ ((-)_{\overline{\sem{\Gamma;\Phi \vdash
          K}\rho'}} \times \id)}"] \\
\setsem{\Gamma;\Phi \vdash G [\overline{\alpha := K}]}\rho
\ar[r, "{\setsem{\Gamma;\Phi \vdash G [\overline{\alpha := K}]}f}"']
&
\setsem{\Gamma;\Phi \vdash G [\overline{\alpha := K}]}\rho'
\end{tikzcd}\]}

\noindent
The top diagram commutes because $g$ and $h$ are natural in $\rho$ by
the induction hypothesis.
To see that the bottom diagram commutes,
we need to show that
\[\begin{array}{ll}
& \setsem{\Gamma;\Phi \vdash G [\overline{\alpha := K}]}f
(\eta_{\overline{\sem{\Gamma;\Phi \vdash K}\rho}} x)\\
= &
(\setsem{\Gamma; \emptyset \vdash \Nat^{\overline{\alpha}} \,F \,G} f\, \eta
)_{\overline{\sem{\Gamma;\Phi \vdash K}\rho'}}
(\setsem{\Gamma;\Phi \vdash F [\overline{\alpha := K}]}f x)
\end{array}\]
holds for all $\eta \in \setsem{\Gamma; \emptyset \vdash
  \Nat^{\overline{\alpha}} \,F \,G}\rho$ and $x \in
\setsem{\Gamma;\Phi \vdash F [\overline{\alpha := K}]}\rho$,
i.e.,
that
\[\begin{array}{ll}
 & \setsem{\Gamma; \ol{\alpha} \vdash G} f[\overline{\alpha :=
    \setsem{\Gamma;\Phi \vdash K} f }] \circ
\eta_{\overline{\setsem{\Gamma;\Phi \vdash K}\rho}} \\
=\;&
\eta_{\overline{\setsem{\Gamma;\Phi \vdash K}\rho'}}
\circ
\setsem{\Gamma; \ol{\alpha} \vdash F} f [\ol{\alpha := \setsem{\Gamma;\Phi \vdash K}f}]
\end{array}\]
for all $\eta \in \setsem{\Gamma;\emptyset \vdash
  \Nat^{\overline{\alpha}} \,F \,G}\rho$.
But this follows from the naturality of $\eta$, which indeed ensures
the commutativity of
{\footnotesize
\[\hspace*{-0.3in}\begin{tikzcd}[column sep = large]
\setsem{\Gamma; \ol{\alpha} \vdash F}\rho[\ol{\alpha :=
    \setsem{\Gamma;\Phi \vdash K}\rho}] \ar[r,
  "{\;\;\;\eta_{\ol{\setsem{\Gamma;\Phi \vdash K}\rho}}\;\;\; }"]
\ar[d, "{\setsem{\Gamma; \ol{\alpha} \vdash F} f [\ol{\alpha :=
        \setsem{\Gamma;\Phi \vdash K}f}]}"']
& \setsem{\Gamma;
  \ol{\alpha} \vdash G}\rho[\ol{\alpha := \setsem{\Gamma;\Phi \vdash
      K}\rho}]
\ar[d, "{\setsem{\Gamma; \ol{\alpha} \vdash G} f [\ol{\alpha :=
        \setsem{\Gamma;\Phi \vdash K}f}]}"]\\
\setsem{\Gamma; \ol{\alpha} \vdash F}\rho'[\ol{\alpha :=
    \setsem{\Gamma;\Phi \vdash K}\rho'}] \ar[r,
  "{\eta_{\ol{\setsem{\Gamma;\Phi \vdash K}\rho'}} }"]
& \setsem{\Gamma; \ol{\alpha} \vdash G}\rho'[\ol{\alpha :=
    \setsem{\Gamma;\Phi \vdash K}\rho'}]
\end{tikzcd}\]}

\item
  $\underline{\Gamma;\emptyset~|~\emptyset \vdash
  \map^{\ol{F},\ol{G}}_H :
  \Nat^\emptyset\;(\ol{\Nat^{\ol{\beta},\ol{\gamma}}\,F\,G})\;
  (\Nat^{\ol{\gamma}}\,H[\ol{\phi :=_{\ol{\beta}} F}]\,H[\ol{\phi
      :=_{\ol{\beta}} G}])}$\;
To see that
\[
\setsem{\Gamma; \emptyset~|~\emptyset \vdash \map^{\ol{F},\ol{G}}_H
    : \Nat^\emptyset\;(\ol{\Nat^{\ol{\beta},\ol{\gamma}}\,F\,G})\;
    (\Nat^{\ol{\gamma}}\,H[\ol{\phi :=_{\ol{\beta}} F}]\,H[\ol{\phi
        :=_{\ol{\beta}} G}])}\,\rho\,d\, \ol{\eta}
\]
is in $\setsem{\Gamma; \emptyset \vdash
    \Nat^{\ol{\gamma}}\,H[\ol{\phi :=_{\ol{\beta}} F}]\,H[\ol{\phi
        :=_{\ol{\beta}} G}]} \rho$
for all $\rho : \setenv$, $d : \setsem{\Gamma;\emptyset \vdash \emptyset} \rho$,
and \[\ol{\eta : \setsem{\Gamma; \emptyset
  \vdash\Nat^{\ol{\beta},\ol{\gamma}}\,F\,G} \rho}\]
  we first note that
$\setsem{\Gamma ;\ol{\phi}, \ol{\gamma} \vdash H}$ is a functor from
  $\setenv$ to $\set$ and, for any $\ol C$, \[\id_{\rho[\ol{\gamma :=
        C}]}[\ol{\phi := \lambda \ol{B}. \eta_{\ol{B}\,\ol{C}}}]\] is a
  morphism in $\setenv$ from \[\rho[\ol{\gamma := C}][\ol{\phi :=
      \lambda \ol{B}.\setsem{\Gamma; \ol{\gamma},\ol{\beta} \vdash
        F}\rho[\ol{\gamma := C}][\ol{\beta := B}]}]\] to
\[\rho[\ol{\gamma := C}][\ol{\phi := \lambda \ol{B}.\setsem{\Gamma;
\ol{\gamma},\ol{\beta} \vdash G}\rho[\ol{\gamma := C}][\ol{\beta := B}]}]\]
so that
\begin{align*}
&(\setsem{\Gamma; \emptyset~|~\emptyset \vdash
  \map^{\ol{F},\ol{G}}_H :
\Nat^\emptyset\;(\ol{\Nat^{\ol{\beta},\ol{\gamma}}\,F\,G})\;
(\Nat^{\ol{\gamma}}\,H[\ol{\phi :=_{\ol{\beta}} F}]\,H[\ol{\phi
    :=_{\ol{\beta}} G}])}\,\rho\,d\, \ol{\eta})_{\ol{C}} \\
  &=  \setsem{\Gamma; \ol{\phi},\ol{\gamma} \vdash H}\id_{\rho[\ol{\gamma
      := C}]}[\ol{\phi := \lambda \ol{B}. \eta_{\ol{B}\,\ol{C}}}]
\end{align*}
is indeed a morphism of type
\[\setsem{\Gamma ;\ol{\gamma} \vdash H[\ol{\phi := F}]}\rho[\ol{\gamma
      := C}]
\to
\setsem{\Gamma ;\ol{\gamma} \vdash H[\ol{\phi := G}]}\rho[\ol{\gamma
      := C}]\]
This family of morphisms is natural in $\ol C$: if $\ol{f : C \to C'}$
then, writing $\xi$ for
\[
\setsem{\Gamma; \emptyset~|~\emptyset \vdash \map^{\ol{F},\ol{G}}_H
    : \Nat^\emptyset\;(\ol{\Nat^{\ol{\beta},\ol{\gamma}}\,F\,G})\;
    (\Nat^{\ol{\gamma}}\,H[\ol{\phi :=_{\ol{\beta}} F}]\,H[\ol{\phi
        :=_{\ol{\beta}} G}])}\,\rho\,d\, \ol{\eta}
\]
the naturality of $\eta$,
together with the fact that composition of environments is
computed componentwise, ensure that the following naturality diagram
for $\xi$ commutes:
{\footnotesize
\[\begin{tikzcd}
\setsem{\Gamma ;\ol{\gamma} \vdash H[\ol{\phi := F}]}\rho[\ol{\gamma
      := C}] \ar[r, "{\xi_{\ol{C}}}"]
\ar[d, "{\setsem{\Gamma ;\ol{\gamma} \vdash H[\ol{\phi :=
          F}]}\id_{\rho}[\ol{\gamma := f}]}"']
& \setsem{\Gamma ; \ol{\gamma} \vdash H[\ol{\phi := G}]}\rho[\ol{\gamma
      := C}]
\ar[d, "{\setsem{\Gamma ; \ol{\gamma} \vdash H[\ol{\phi :=
          G}]}\id_{\rho}[\ol{\gamma := f}]}"]\\
\setsem{\Gamma ;\ol{\gamma} \vdash H[\ol{\phi := F}]}\rho[\ol{\gamma
      := C'}] \ar[r, "{\xi_{\ol{C'}}}"]
& \setsem{\Gamma ; \ol{\gamma} \vdash H[\ol{\phi := G}]}\rho[\ol{\gamma
      := C'}]
\end{tikzcd}\]}
That, for all $\rho : \setenv$ and $d : \setsem{\Gamma; \emptyset \vdash
  \emptyset}\rho$, $\xi$ satisfies the additional condition needed for
it to be in $\setsem{\Gamma; \emptyset \vdash
  \Nat^{\ol{\gamma}}\,H[\ol{\phi :=_{\ol{\beta}} F}]\,H[\ol{\phi
      :=_{\ol{\beta}} G}]}\rho$ follows from the fact
that $\eta$ satisfies the extra
condition needed for it to be in its corresponding
$\setsem{\Gamma; \emptyset \vdash \Nat^{\ol{\beta},\ol{\gamma}}\,F\,G} \rho$.
Finally, since $\Phi = \emptyset$, the naturality of
\[
\setsem{\Gamma; \emptyset~|~\emptyset \vdash \map^{\ol{F},\ol{G}}_H
    : \Nat^\emptyset\;(\ol{\Nat^{\ol{\beta},\ol{\gamma}}\,F\,G})\;
    (\Nat^{\ol{\gamma}}\,H[\ol{\phi :=_{\ol{\beta}} F}]\,H[\ol{\phi
        :=_{\ol{\beta}} G}])}\rho
\]
in $\rho$ is trivial.

\item
$\underline{\Gamma;\emptyset \,|\, \emptyset \vdash \tin_H :
  \Nat^{\ol{\beta}} \, H[\phi := (\mu \phi.\lambda
    {\overline \alpha}.H){\overline \beta}][\ol{\alpha := \beta}]
  \;(\mu \phi.\lambda {\overline \alpha}.H){\overline \beta}}$\; To
  see that if $d : \setsem{\Gamma;\emptyset \vdash \emptyset} \rho$
  then \[\setsem{\Gamma;\emptyset \,|\, \emptyset \vdash \tin_H :
    \Nat^{\ol{\beta}} \, H[\phi := (\mu \phi.\lambda
      {\overline \alpha}.H){\overline \beta}][\ol{\alpha := \beta}]
    \;(\mu \phi.\lambda {\overline \alpha}.H){\overline \beta}}\,
  \rho\,d\]
  is in $\setsem{\Gamma;\emptyset \vdash
    \Nat^{\ol{\beta}} \, H[\phi := (\mu \phi.\lambda
      {\overline \alpha}.H){\overline \beta}][\ol{\alpha := \beta}]
    \;(\mu \phi.\lambda {\overline \alpha}.H){\overline \beta}}\,
  \rho$, we first note that, for all $\ol{B}$,
  \begin{align*}
    &(\setsem{\Gamma;\emptyset \,|\, \emptyset \vdash \tin_H :
    \Nat^{\ol{\beta}} \, H[\phi := (\mu \phi.\lambda
      {\overline \alpha}.H){\overline \beta}][\ol{\alpha := \beta}]
    \;(\mu \phi.\lambda {\overline \alpha}.H){\overline \beta}}\,
  \rho\,d)_{\ol{B}}\, \\
    =\,
    &(\mathit{in}_{T^\set_{H,\rho}})_{\ol{B}}
  \end{align*}
  maps
  \begin{align*}
    &\setsem{\Gamma;\ol{\beta} \vdash H[\phi := (\mu
      \phi.\lambda {\overline \alpha}.H){\overline \beta}][\ol{\alpha
        := \beta}]}\rho[\ol{\beta := B}]  \\ = \,
    &T^\set_{H,\rho}\, (\mu T^\set_{H,\rho}) \, \ol{B}
  \end{align*}
        to $\setsem{\Gamma;\ol{\beta}
    \vdash (\mu \phi.\lambda {\overline \alpha}.H){\overline \beta}}
  \rho[\ol{\beta := B}] = (\mu
  T^\set_{H,\rho}) \, \ol{B}$.
  Secondly, we observe
  that
  \begin{align*}
    &\setsem{\Gamma;\emptyset \,|\, \emptyset \vdash \tin_H :
    \Nat^{\ol{\beta}} \, H[\phi := (\mu \phi.\lambda
      {\overline \alpha}.H){\overline \beta}][\ol{\alpha := \beta}]
    \;(\mu \phi.\lambda {\overline \alpha}.H){\overline
      \beta}}\,\rho\,d \\
      = \, &\mathit{in}_{T^\set_{H,\rho}}
  \end{align*}
  is natural in $\ol{B}$, since naturality of
  $\mathit{in}_{T^\set_{H,\rho}}$
  ensures that the following diagram commutes for all $\ol{f : B \to
    B'}$:

  {\tiny
  \[\hspace*{-0.15in}\begin{tikzcd}[column sep=2.5in, row sep=0.75in]
T^\set_{H,\rho}\, (\mu
T^\set_{H,\rho})\, \ol{B} \ar[d, "{T^\set_{H,\rho}\, (\mu T^\set_{H,\rho})\, \ol{f}}"]
\ar[r,"{(\mathit{in}_{T^\set_{H,\rho})_{\ol{B}}}}" ] &
  (\mu T^\set_{H,\rho})\, \ol{B}
\ar[d,"{\mu {T^\set_{H,\rho}}\, {\ol{f}}}" ]
\\
T^\set_{H,\rho}\, (\mu T^\set_{H,\rho})\,
\ol{B'} \ar[r, "{(\mathit{in}_{T^\set_{H,\rho}})_{\ol{B'}}}"] & (\mu
T^\set_{H,\rho})\,
\ol{B'}
\end{tikzcd}\]
}
  That, for all $\rho : \setenv$ and $d :
\setsem{\Gamma;\emptyset \vdash \emptyset}\rho$,
\[\setsem{\Gamma;\emptyset \,|\, \emptyset \vdash
  \tin_H : \Nat^{\ol{\beta}} \, H[\phi := (\mu \phi.\lambda
    {\overline \alpha}.H){\overline \beta}][\ol{\alpha := \beta}]
  \;(\mu \phi.\lambda {\overline \alpha}.H){\overline
    \beta}}\,\rho\,d\] satisfies the additional property needed for
it to be in
\[\setsem{\Gamma;\emptyset \vdash
  \Nat^{\ol{\beta}} \, H[\phi := (\mu \phi.\lambda
    {\overline \alpha}.H){\overline \beta}][\ol{\alpha := \beta}]
  \;(\mu \phi.\lambda {\overline \alpha}.H){\overline \beta}}\,\rho\]
follows from the fact that, for every $\ol{R : \rel(B,B')}$,
\[\hspace*{-0.2in}
  \begin{array}{ll}
 & (\,(\setsem{\Gamma;\emptyset \,|\, \emptyset \vdash \tin_H :
      \Nat^{\ol{\beta}} \, H[\phi := (\mu \phi.\lambda
        {\overline \alpha}.H){\overline \beta}][\ol{\alpha := \beta}]
      \;(\mu \phi.\lambda {\overline \alpha}.H){\overline
        \beta}}\,\rho\,d)_{\ol{B}},\,\\ & \hspace*{0.5in}(\setsem{\Gamma;\emptyset
      \,|\, \emptyset \vdash \tin_H : \Nat^{\ol{\beta}} \,
      H[\phi := (\mu \phi.\lambda {\overline \alpha}.H){\overline
          \beta}][\ol{\alpha := \beta}] \;(\mu \phi.\lambda {\overline
        \alpha}.H){\overline
        \beta}}\,\rho\,d)_{\ol{B'}}\,)\\ &= (\,
    (\mathit{in}_{T^\set_{H,\rho}})_{\ol{B}},
    (\mathit{in}_{T^\set_{H,\rho}})_{\ol{B'}}\,)
\end{array}\]
has type
\[\begin{array}{ll}
& (\, T^\set_{H,\rho}\, (\mu T^\set_{H,\rho}) \, \ol{B} \to (\mu
T^\set_{H,\rho}) \, \ol{B}, \, T^\set_{H,\rho}\, (\mu T^\set_{H,\rho})
\, \ol{B'} \to (\mu T^\set_{H,\rho}) \,\ol{B'} \, )\\[1ex]
= &
 \relsem{\Gamma;\ol{\beta} \vdash H[\phi := (\mu
    \phi.\lambda {\overline \alpha}.H){\overline \beta}][\ol{\alpha :=
      \beta}]}\Eq_\rho[\ol{\beta := R}]\to\\
 & \hspace*{0.5in} \relsem{\Gamma;\ol{\beta} \vdash (\mu
  \phi.\lambda \ol{\alpha}.H)\ol{\beta}} \Eq_\rho[\ol{\beta:=
    R}]
\end{array}\]
Finally, since $\Phi = \emptyset$, naturality of
\[
\setsem{\Gamma;\emptyset \,|\, \emptyset \vdash \tin_H :
  \Nat^{\ol{\beta}} \, H[\phi := (\mu \phi.\lambda
    {\overline \alpha}.H){\overline \beta}][\ol{\alpha := \beta}]
  \;(\mu \phi.\lambda {\overline \alpha}.H){\overline \beta}}
\]
in $\rho$ is trivial.

\item
$\underline{\Gamma; \emptyset~|~\emptyset \vdash \fold^F_H :
  \Nat^\emptyset\;(\Nat^{\ol{\beta}}\,H[\phi
    :=_{\ol{\beta}} F][\ol{\alpha := \beta}]\,F)\;
  (\Nat^{{\ol{\beta}} }\,(\mu \phi.\lambda \overline
  \alpha.H)\overline \beta \;F)}$ \; Since $\Phi$ is empty, to see
  that
  \[\setsem{ \Gamma; \emptyset~|~\emptyset \vdash \fold^F_H :
    \Nat^\emptyset\;(\Nat^{\ol{\beta}}\,H[\phi :=_{\ol{\beta}}
      F][\ol{\alpha := \beta}]\,F)\; (\Nat^{{\ol{\beta}} }\,(\mu
    \phi.\lambda \overline \alpha.H)\overline \beta \;F)}\] is a
  natural transformation from $\setsem{\Gamma;\emptyset \vdash
    \emptyset}$ to \[\setsem{\Gamma; \emptyset \vdash
    \Nat^\emptyset\;(\Nat^{\ol{\beta}}\,H[\phi :=_{\ol{\beta}}
      F][\ol{\alpha := \beta}]\,F)\; (\Nat^{{\ol{\beta}} }\,(\mu
    \phi.\lambda \overline \alpha.H)\overline \beta\,F)}\] we need
  only show that, for all $\rho : \setenv$, all $\eta :
  \setsem{\Gamma; \emptyset \vdash \Nat^{\ol{\beta}}\,H[\phi
      :=_{\ol{\beta}} F][\ol{\alpha := \beta}]\,F} \rho$, and the
  unique $d : \setsem{\Gamma;\emptyset \vdash \emptyset} \rho$,
\[ \setsem{\Gamma; \emptyset~|~\emptyset \vdash \fold^F_H :
  \Nat^\emptyset\;(\Nat^{\ol{\beta}}\,H[\phi
    :=_{\ol{\beta}} F][\ol{\alpha := \beta}]\,F)\;
  (\Nat^{{\ol{\beta}} }\,(\mu \phi.\lambda \overline
  \alpha.H)\overline \beta\,F)}\,\rho\,d\,\eta\] has type
$\setsem{\Gamma; \emptyset \vdash \Nat^{{\ol{\beta}}
  }\,(\mu \phi.\lambda \overline \alpha.H)\overline \beta\,F}\,\rho$,
i.e., for any $\ol{B}$,
\[(\setsem{\Gamma; \emptyset~|~\emptyset \vdash \fold^F_H :
  \Nat^\emptyset\;(\Nat^{\ol{\beta}}\,H[\phi :=_{\ol{\beta}}
    F][\ol{\alpha := \beta}]\,F)\; (\Nat^{{\ol{\beta}} }\,(\mu
  \phi.\lambda \overline \alpha.H)\overline
  \beta\,F)}\,\rho\,d\,\eta)_{\ol{B}}\] is a morphism from
$\setsem{\Gamma; \ol{\beta}\vdash (\mu \phi.\lambda \overline
  \alpha.H)\overline \beta}\rho[\ol{\beta := B}] \,=\,(\mu
T^\set_{H,\rho})_{\ol{B}}$ to $\setsem{\Gamma; \ol{\beta} \vdash
  F}\rho[\ol{\beta := B}]$.  To see this, note that $\eta$ is a
natural transformation from
\[\begin{array}{ll}
 & \lambda \ol{B}.\,\setsem{\Gamma; \ol{\beta}
  \vdash H[\phi := F][\ol{\alpha := \beta}]}\rho[\ol{\beta :=
    B}]\\
= & \lambda \ol{B}.\,T^\set_{H,\rho}\,(\lambda \ol{A}. \,
\setsem{\Gamma;\ol{\beta}
  \vdash F}\rho[\ol{\beta := A}]) \, \ol{B}
\end{array}\]
to
\[\begin{array}{ll}
 & \lambda \ol{B}.\,(\lambda
\ol{A}.\,\setsem{\Gamma;\ol{\beta}\vdash F}\rho[\ol{\beta
    := A}]) \ol{B}\\
= & \lambda
\ol{B}.\,\setsem{\Gamma;\ol{\beta} \vdash
  F}\rho[\ol{\beta := B}]
\end{array}\]
and thus for each $\ol{B}$,
\[(\setsem{\Gamma; \emptyset~|~\emptyset \vdash \fold^F_H :
  \Nat^\emptyset\;(\Nat^{\ol{\beta}}\,H[\phi :=_{\ol{\beta}}
    F][\ol{\alpha := \beta}]\,F)\; (\Nat^{{\ol{\beta}} }\,(\mu
  \phi.\lambda \overline \alpha.H)\overline
  \beta\,F)}\,\rho\,d\,\eta)_{\ol{B}}\] is a morphism from
$\setsem{\Gamma; \ol{\beta} \vdash (\mu \phi.\lambda \overline
  \alpha.H)\overline \beta}\rho[\ol{\beta := B}]\,=\, (\mu
T^\set_{H,\rho})_{\ol{B}}$ to $\setsem{\Gamma;\ol{\beta}\vdash
  F}\rho[\ol{\beta := B}]$.  To see that this family of morphisms is
natural in $\ol{B}$, we observe that the following diagram commutes
for all $\ol{f : B \to B'}$':
{\tiny
\[\hspace*{-0.3in}\begin{tikzcd}[column sep=2.5in, row sep=0.75in]
(\mu T^\set_{H,\rho})\, \ol{B}
\ar[r,"{(\mathit{fold}_{T^\set_{H,\rho}}\,\eta)_{\ol{B}}\,}" ]
\ar[d, "{(\mu T^\set_{H,\rho}) \, \ol{f}}"'] &
 \setsem{\Gamma; \ol{\beta} \vdash F}\rho[\ol{\beta := B}]\ar[d,
   "{\setsem{\Gamma;
      \ol{\beta} \vdash F}\id_\rho[\ol{\beta := f}]}"]\\
 (\mu T^\set_{H,\rho})\, \ol{B'}
\ar[r,"{(\mathit{fold}_{T^\set_{H,\rho}}\eta)_{\ol{B'}}\,}"
] &
 \setsem{\Gamma; \ol{\beta} \vdash F}\rho[\ol{\beta := B'}]
\end{tikzcd}\]}
by naturality of $\mathit{fold}_{T^\set_{H,\rho}}\eta$.  To see that,
for all $\rho : \setenv$, \[\eta : \setsem{\Gamma; \emptyset \vdash
  \Nat^{\ol{\beta}}\,H[\phi :=_{\ol{\beta}} F][\ol{\alpha :=
      \beta}]\,F} \rho\] and $d \in \setsem{\Gamma; \emptyset \vdash
  \emptyset}\rho$,
\[\setsem{\Gamma; \emptyset~|~\emptyset
  \vdash \fold^F_H : \Nat^\emptyset\;(\Nat^{\ol{\beta}}\,H[\phi :=_{\ol{\beta}} F][\ol{\alpha :=
      \beta}]\,F)\; (\Nat^{{\ol{\beta}} }\,(\mu
  \phi.\lambda \overline \alpha.H)\overline
  \beta\,F)}\,\rho\,d\,\eta\] satisfies the additional condition
needed to be in $\setsem{\Gamma;\emptyset \vdash
  \Nat^{{\ol{\beta}} }\,(\mu \phi.\lambda \overline
  \alpha.H)\overline \beta\;F}\,\rho$, let $\ol{R : \rel(B,B')}$.
Since $\eta$ satisfies the
additional condition needed to be in \[\setsem{\Gamma; \emptyset
  \vdash \Nat^{\ol{\beta}}\,(H[\phi := F][\ol{\alpha :=
      \beta}])\,F} \rho\]
we know that
\[\begin{array}{ll}
 & (\,(\mathit{fold}_{T^\set_{H,\rho}}\eta)_{\ol{B}},\,
(\mathit{fold}_{T^\set_{H,\rho}}\eta)_{\ol{B'}}\,)
\end{array}\]
has type
\[\begin{array}{ll}
  & (\mu T_{H,\Eq_\rho}) \,\ol{R} \to
\relsem{\Gamma;\ol{\beta} \vdash F}\Eq_\rho[\ol{\beta:= R}]\\
= & (\mu T_{H,\Eq_\rho})\,\ol{\relsem{\Gamma;\ol{\beta}
  \vdash \beta}\Eq_\rho[\ol{\beta := R}]} \to
\relsem{\Gamma;\ol{\beta} \vdash F}\Eq_\rho[\ol{\beta:= R}]\\
= & \relsem{\Gamma; \ol{\beta} \vdash (\mu \phi. \lambda
  \ol{\alpha}. H)\ol{\beta}}\Eq_\rho[\ol{\beta := R}] \to
\relsem{\Gamma;\ol{\beta} \vdash F}\Eq_\rho[\ol{\beta:= R}]
\end{array}\]
Finally, since $\Phi = \emptyset$, naturality of
\[\setsem{\Gamma; \emptyset~|~\emptyset \vdash \fold^F_H :
  \Nat^\emptyset\;(\Nat^{\ol{\beta}}\,H[\phi
    :=_{\ol{\beta}} F][\ol{\alpha := \beta}]\,F)\;
  (\Nat^{{\ol{\beta}} }\,(\mu \phi.\lambda
  \overline \alpha.H)\overline \beta \;F)}\] in $\rho$ is trivial.\qed
\end{itemize}

\noindent
The following theorem is an immediate consequence of
Theorem~\ref{thm:at-gen}:
\begin{thm}\label{thm:at-gen-rel}
If \,$\Gamma; \Phi \,|\, \Delta \vdash t : F$ and \,$\rho \in \relenv$,
  and if \,$(a, b) \in \relsem{\Gamma; \Phi \vdash \Delta} \rho$,
  then \\
  $(\setsem{\Gamma; \Phi \,|\, \Delta \vdash t : F} (\pi_1 \rho) \,a \, ,
      \setsem{\Gamma; \Phi \,|\, \Delta \vdash t : F } (\pi_2 \rho) \,b) \in
    \relsem{\Gamma; \Phi \vdash F} \rho$.
\end{thm}
\noindent
Finally, the Abstraction Theorem is the instantiation of
Theorem~\ref{thm:at-gen-rel} to closed terms of closed type:
\begin{thm}[Abstraction Theorem]\label{thm:abstraction}
If $\vdash t : F$, then $(\setsem{\vdash t : F},\setsem{\vdash t
  : F}) \in \relsem{\vdash F}$.
\end{thm}

\section{Free Theorems for Nested Types}\label{sec:ftnt}

In this section we show how Theorem~\ref{thm:at-gen} and its
consequences can be used to prove free theorems. Those in
Sections~\ref{sec:ft-adt},~\ref{sec:short-cut},
and~\ref{sec:short-cut-nested} go beyond mere naturality.  We also
show that we can extend short cut fusion for lists~\cite{glp93} to
nested types, thereby formally proving correctness of the
categorically inspired theorem from~\cite{jg10}.

\subsection{Free Theorem for Type of Polymorphic
  Bottom}\label{sec:bottom}

Suppose $ \vdash g : \Nat^\alpha \,\onet\,\alpha$ and $g^\set =
\setsem{\vdash g : \Nat^\alpha \,\onet\,\alpha}$. Then $g^\set$ is a
natural transformation from the constantly $1$-valued functor to the
identity functor in $\set$.  In particular, for every $S : \set$,
$g^\set_S : 1 \to S$. Note, however, that if $S = \emptyset$, then
there can be no such morphism, so no such natural transformation, and
thus no term $\vdash g : \Nat^\alpha \onet \,\alpha$, can exist.  That
is, our calculus admits no (non-terminating) terms with the closed
type $\Nat^\alpha \onet \,\alpha$ of the polymorphic bottom.

\subsection{Free Theorem for Type of Polymorphic
  Identity}\label{sec:identity}

Suppose $ \vdash g : \Nat^\alpha \,\alpha\,\alpha$ and $g^\set =
\setsem{\vdash g : \Nat^\alpha \,\alpha\,\alpha}$.  Then $g^\set$ is a
natural transformation from the identity functor on $\set$ to
itself. If $S$ is any set, if $a$ is any element of $S$, and if $K_a
:S \to S$ is the constantly $a$-valued morphism on $S$, then
naturality of $g^\set$ gives that $g^\set_S \circ K_a = K_a \circ
g^\set_S$, i.e., $g^\set_S \, a = a$, i.e., $g^\set_S = \id_S$.  That
is, $g^\set$ is the identity natural transformation for the identity
functor on $\set$. So every closed term $g$ of closed type
$\Nat^\alpha\alpha\,\alpha$ always denotes the identity natural
transformation for the identity functor on $\set$, i.e., every closed
term $g$ of type $\Nat^\alpha\alpha\,\alpha$ denotes the polymorphic
identity function.

\subsection{Standard Free Theorems for ADTs and Nested
  Types}\label{sec:ft-adt}

We can derive in our calculus even those free theorems for polymorphic
functions over ADTs that are not consequences of naturality.  We can,
e.g., prove the free theorem for $\mathit{filter}$'s type as follows:

\begin{thm}
If $g : A \to B$, $\rho \in \relenv$, $\rho \alpha = (A, B, \graph{g})$,
$(a, b) \in \relsem{\alpha ;\emptyset \vdash \Delta} \rho$, $(s \circ
g, s) \in \relsem{\alpha; \emptyset \vdash \Nat^\emptyset \alpha \,
  \mathit{Bool}} \rho$, and
\[\mathit{filter} = \setsem{\alpha; \emptyset \,|\, \Delta \vdash t :
  \filtype}\]
\noindent
for some $t$, then
\[  \mathit{map}_{\mathit{List}} \,g \circ \mathit{filter} \, (\pi_1
\rho) \, a \, (s \circ g) = \mathit{filter}\, (\pi_2\rho) \, b \, s
\circ \mathit{map}_{\mathit{List}} \,g\]
\end{thm}
\proof
By Theorem~\ref{thm:abstraction}, \[(\mathit{filter}\, (\pi_1 \rho)\,
a, \mathit{filter}\, (\pi_2 \rho)\, b) \in \relsem{\alpha; \emptyset
  \vdash \filtype} \rho\] so if $(s', s) \in \relsem{\alpha; \emptyset
  \vdash \Nat^\emptyset \alpha \, \mathit{Bool}} \rho = \rho\alpha \to
\Eq_{\mathit{Bool}}$ and $(xs', xs) \in \relsem{\alpha; \emptyset
  \vdash List \, \alpha} \rho$ then
\begin{equation}\label{eq:filter-thm-list}
  (\mathit{filter}\, (\pi_1\rho) \,a \,s' \,xs', \mathit{filter} \,
  (\pi_2\rho) \,b \,s \,xs) \in \relsem{\alpha; \emptyset \vdash
    \mathit{List} \, \alpha} \rho
\end{equation}
If $\rho\alpha = (A, B, \graph{g})$, then $\relsem{\alpha; \emptyset
  \vdash \mathit{List} \, \alpha} \rho =
\graph{\mathit{map}_{\mathit{List}} \, g}$ by Lemma~\ref{lem:graph}
and demotion.  Moreover, $xs =
\mathit{map}_{\mathit{List}} \,g \,xs'$ and $(s', s) \in \graph{g} \to
\Eq_{\mathit{Bool}}$, so $s' = s \circ g$. The result
follows from
Equation~\ref{eq:filter-thm-list}.\qed

A similar proof establishes the analogous result for, say, generalized
rose trees.
\begin{thm}
  If $g : A \to B$,
$F, G : \set \to \set$,
  $\eta : F \to G$ in $\set$, $\rho \in \relenv$, $\rho \alpha =
 (A, B, \graph{g})$, $\rho \psi = (F, G, \graph{\eta})$, $(a, b) \in
 \relsem{\alpha, \psi ;\emptyset \vdash \Delta} \rho$, $(s \circ
 g, s) \in \relsem{\alpha; \emptyset \vdash \Nat^\emptyset \alpha \,
   \mathit{Bool}} \rho$, and
 \[ \mathit{filter} = \setsem{\alpha, \psi; \emptyset \,|\, \Delta
      \vdash t : \filtypeGRose}  \]
for some $t$, then
\[ \semmap_{\mathit{GRose}}\, \eta\, (g + 1) \circ \mathit{filter} \,
(\pi_1 \rho) \, a \, (s \circ g) = \mathit{filter} \, (\pi_2\rho) \, b
\, s \circ
\semmap_{\mathit{GRose}}\, \eta\, g\]
\end{thm}

\noindent
This is not surprising since rose trees are essentially
ADT-like. However, as noted in Section~\ref{sec:terms}, our calculus
cannot express the type of a polymorphic filter function for a proper
nested type.

\subsection{Short Cut Fusion for Lists}\label{sec:short-cut}

We can recover standard short cut fusion for lists~\cite{glp93} in our
calculus:
\begin{thm}
If \,$\vdash F$, $\vdash H$, and
$\hat{g} \, = \, \setsem{\beta; \emptyset \,|\, \emptyset \vdash g :
  \Nat^{\emptyset} (\Nat^{\emptyset} (\onet + F \times \beta)\,
  \beta)\, \beta}$ for some $g$, and if $c \in \setsem{\vdash F}
\,\times\, \setsem{\vdash H} \to \setsem{\vdash H}$ and $n \in
\setsem{\vdash H}$, then
\[\mathit{fold}_{1 + \setsem{\vdash F} \times \_}\, n\, c\; (\hat{g}\; (\mathit{List}\,
\setsem{\vdash F})\,\mathit{nil} \,\mathit{cons}) = \hat{g} \,\setsem{\vdash
  H}\, n\, c \]
\end{thm}
\begin{proof}
Theorem~\ref{thm:abstraction} gives
that, for any $\rho \in \relenv$,
\[\begin{array}{lll}
(\hat{g} \,(\pi_1 \rho), \,\hat{g}\, (\pi_2 \rho)) & \in
& \relsem{\beta; \emptyset \vdash \Nat^{\emptyset} (\Nat^{\emptyset}
  (\onet + F \times \beta)\, \beta)\, \beta} \rho \\
& \cong & (((\relsem{\vdash F} \rho \times
\rho\beta) \to \rho\beta) \times \rho\beta) \to \rho\beta
\end{array}\]
so if $(c', c) \in \relsem{\vdash F} \rho \times
\rho\beta \to \rho\beta$ and $(n', n) \in \rho\beta$ then
$(\hat{g} \,(\pi_1 \rho)\, n'\, c', \hat{g}\, (\pi_2 \rho)\, n\, c)
\in \rho \beta$.
In addition,
\begin{align*}
  &\setsem{\vdash \fold_{\onet + F
    \times \beta}^{H} : \Nat^{\emptyset} (\Nat^{\emptyset} (\onet
  + F \times H)\, H)\, (\Nat^{\emptyset} (\mu \alpha. \onet
  + F \times \alpha)\, H)} \\ =\;
& \mathit{fold}_{1 + \setsem{\vdash F} \times \_}
\end{align*}
so that if $c \in
\setsem{\vdash F} \,\times\, \setsem{\vdash H} \to
\setsem{\vdash H}$ and $n \in \setsem{\vdash H}$,
 then
\[(n, c) \in \setsem{\vdash \Nat^{\emptyset} (\onet + F \times H)\,
  H}\] The instantiation
\[\begin{array}{lll}
\pi_1 \rho \beta & = & \setsem{\vdash \mu \alpha. \onet + F \times
  \alpha}=\mathit{List}\,\setsem{\vdash F}\\
\pi_2 \rho \beta & = & \setsem{\vdash H}\\
\rho \beta & = & \graph{\mathit{fold}_{1 + \setsem{\vdash F} \times \_}\, n\, c} :
\rel(\pi_1 \rho \beta, \pi_2 \rho \beta)\\
c' & = & \mathit{cons}\\
n' & = & \mathit{nil}
\end{array}\]
thus gives that \[(\hat{g}\;(\mathit{List}\,\setsem{\vdash
  F})\,\mathit{nil} \,\mathit{cons}, \hat{g} \,\setsem{\vdash H} \,
n\, c) \in \graph{\mathit{fold}_{1 + \setsem{\vdash F} \times \_}\,
  n\, c}\]
i.e.,
\[\mathit{fold}_{1 + \setsem{\vdash F} \times \_}\, n\, c\; (\hat{g}\;
(\mathit{List}\,\setsem{\vdash F})\, \mathit{nil}\, \mathit{cons})
= \hat{g}\,\setsem{\vdash H}\, n\, c \qedhere\]
\end{proof}

We can extend short cut fusion results to arbitrary ADTs, as
in~\cite{joh02,pit98}.

\subsection{Short Cut Fusion for Arbitrary Nested
  Types}\label{sec:short-cut-nested}

We can extend short cut fusion for lists~\cite{glp93} to nested types,
thereby formally prove correctness of the categorically inspired
theorem from~\cite{jg10}.  We have:
\begin{thm}\label{thm:short-cut-nested}
If $\emptyset;\phi,\alpha \vdash F$, \,$\emptyset; \alpha
\vdash K$, \,
$H : [\set,\set] \to [\set,\set]$ is defined by
\[\begin{array}{lll}
H\,f\,x & = & \setsem{\emptyset; \phi, \alpha \vdash F}[\phi :=
  f][\alpha := x]\\
\end{array}\]
and
\[\hat{g} = \setsem{\phi;\emptyset\,|\,\emptyset \vdash g :
\Nat^\emptyset\,(\Nat^\alpha\,F\,(\phi\alpha))\,(\Nat^\alpha\,\onet \,
(\phi\alpha))}\] for some $g$, then, for every $B \in H
\setsem{\emptyset;\alpha \vdash K} \rightarrow \setsem{\emptyset;
  \alpha \vdash K}$,
\[\mathit{fold}_{H}\, B \, (\hat{g}\; \mu H \; \mathit{in}_{H}) = \hat{g}
\,\setsem{\emptyset;\alpha \vdash K}\, B\]
\end{thm}

\proof
Theorem~\ref{thm:abstraction} gives that, for any
$\rho \in \relenv$,
\[\begin{array}{lll}
(\hat{g} \,(\pi_1 \rho), \hat{g}\, (\pi_2 \rho)) & \in &
\relsem{\phi;\emptyset\vdash \Nat^{\emptyset} (\Nat^\alpha
  F\, (\phi\alpha))\, (\Nat^\alpha\,\onet \, (\phi\alpha))}
\rho\\
& = & \relsem{\phi;\emptyset\vdash \Nat^\alpha F\,
  (\phi\alpha)}\rho \to \relsem{\phi;\emptyset\vdash
  \Nat^\alpha\,\onet \, (\phi\alpha)}\rho\\
& = & \relsem{\phi;\emptyset\vdash \Nat^\alpha F\,
  (\phi\alpha)}\rho \to \rho \phi
\end{array}\]
\noindent
so if $(A, B) \in \relsem{\phi;\emptyset\vdash \Nat^\alpha F\,
  (\phi\alpha)}\rho$ then $(\hat{g} \,(\pi_1 \rho)\, A, \hat{g}\, (\pi_2 \rho)\,
B) \in \rho \phi$.
Also,
\[\setsem{\vdash \fold_F^K :
  \Nat^{\emptyset}\, (\Nat^\alpha F[\phi := K]\,K)\, (\Nat^\alpha
  ((\mu \phi.\lambda\alpha.F)\alpha)\,K)} = \mathit{fold}_H\]
Now let $A = \mathit{in}_H : H (\mu H) \Rightarrow
\mu H$,\, $B : H\setsem{\emptyset;\alpha\vdash K} \Rightarrow
\setsem{\emptyset;\alpha \vdash K}$,\, $\rho \phi =
\graph{\mathit{fold}_H\, B}$,\, $\pi_1 \rho \phi = \mu H$,\, $\pi_2
\rho \phi = \setsem{\emptyset;\alpha\vdash K}$,\, $\rho \phi :
\rel(\pi_1 \rho \phi, \pi_2 \rho \phi)$,\, $A : \setsem{\phi;
  \emptyset \vdash \Nat^\alpha F \, (\phi \alpha)} (\pi_1 \rho)$,\,
and $B : \setsem{\phi; \emptyset \vdash \Nat^\alpha F \, (\phi
  \alpha)} (\pi_2 \rho)$.
Demotion ensures that \[A = \mathit{in}_H : H(\mu H) \Rightarrow \mu H
= \setsem{\phi;\emptyset \vdash \Nat^\alpha F
  \,(\phi\alpha)}(\pi_1\rho)\]
and demotion and Lemma~\ref{lem:graph} together give that
\[\begin{array}{lll}
(A,B) \,=\, (\mathit{in}_H,B) & \in & \relsem{\phi;\emptyset\vdash
  \Nat^\alpha F\, (\phi\alpha)}\rho\\
& = & \lambda A. \relsem{\phi;\alpha\vdash F}[\phi :=
  \graph{\mathit{fold}_H\, B}][\alpha := A] \Rightarrow
 \graph{\mathit{fold}_H\, B} \\
& = & \relsem{\emptyset;\phi,\alpha\vdash F}
  \graph{\mathit{fold}_H\, B} \Rightarrow \graph{\mathit{fold}_H\,
    B}\\
  & = & \graph{\setsem{\emptyset;\phi,\alpha\vdash F}
    \,(\mathit{fold}_H\,B)} \Rightarrow \graph{\mathit{fold}_H\, B}\\
  & = & \graph{\mathit{map}_H \,(\mathit{fold}_H\,B)} \Rightarrow
\graph{\mathit{fold}_H\, B}\\
\end{array}\]
since if $(x,y) \in \graph{\mathit{map}_H \,(\mathit{fold}_H\,B)}$,
then \[\mathit{fold}_H\, B\, (\mathit{in}_H\,x) = B\,y = B\,
(\mathit{map}_H \,(\mathit{fold}_H\,B) \, x)\] by the definition of
$\mathit{fold}_H$ as a (indeed, the unique) morphism from
$\mathit{in}_H$ to $B$.  Thus, \[(\hat{g} \,(\pi_1 \rho)\, A,  \hat{g}\, (\pi_2
\rho)\, B) \in \graph{\mathit{fold}_H\, B},\] i.e., $\mathit{fold}_H \,
B \, (\hat{g}\, (\pi_1 \rho) \, \mathit{in}_H) = \hat{g}\,(\pi_2 \rho)\,B$.  But
since $\phi$ is the only free variable in $\hat{g}$, this simplifies to
$\mathit{fold}_H\, B \, (\hat{g}\, \mu H\, \mathit{in}_H) =
\hat{g}\,\setsem{\emptyset;\alpha\vdash K}\,B$. \qed

\medskip
As in~\cite{jg10}, replacing $\onet$ with any type $\emptyset;\alpha
\vdash C$ generalizes Theorem~\ref{thm:short-cut-nested} to deliver a
more general free theorem whose conclusion is $\mathit{fold}_{H}\, B
\; \circ \; \hat{g}\; \mu H \; \mathit{in}_{H} = \hat{g}
\,\setsem{\emptyset;\alpha \vdash K}\, B$.

\medskip
Although it is standard to prove that the parametric model constructed
verifies the existence of initial algebras, this is unnecessary here
since initial algebras are built directly into our model.

\section{Parametricity for GADTs}\label{sec:GADTs}

As discussed in Section~\ref{sec:intro}, type indices for nested types
can be any types, including, in the case of truly nested types like
that of bushes, types involving the very same nested type that is
being defined. But every data constructor for a nested type must still
have as its return type exactly the instance being defined. For
example, the data constructors for the instance $\mathtt{PTree\; A}$ of the
nested type $\mathtt{PTree}$ are $\mathtt{pleaf :: PTree\; A}$ and
$\mathtt{pnode :: PTree\; (A \times A) \to PTree\; A}$, and
the data constructors for the instance $\mathtt{Bush \;A}$ of the truly
nested type $\mathtt{Bush}$ are $\mathtt{bnil :: Bush\; A}$ and
$\mathtt{bcons :: A \to Bush\; (Bush\; A) \to Bush\; A}$.

Generalized algebraic data types (GADTs) --- also known as guarded
recursive data types~\cite{xcc03} or first-class phantom
types~\cite{ch03} --- generalize nested types by relaxing the above
restriction to allow the return types of data constructors to be
different instances of the data type than the one being defined. For
example, the GADT
\[\begin{array}{l}
\mathtt{data\; Seq \;(A : Set)\;:\;Set\;where}\\
\hspace*{0.4in}\mathtt{sconst\;:\; A \rightarrow Seq\;A}\\
\hspace*{0.4in}\mathtt{spair\;:\;Seq\;A \rightarrow Seq\;B \rightarrow
  Seq\;(A\times B)}\\
\hspace*{0.4in}\mathtt{sseq\;:\;(Nat \rightarrow Seq\;A) \rightarrow
  Seq\;(Nat \to A)}
\end{array}\]
\noindent
has data constructors $\mathtt{spair}$ and $\mathtt{sseq}$ with
return types $\mathtt{Seq\; (A \times B)}$ and $\mathtt{Seq\;(Nat \to
  A)}$. These types are not only at different instances of \verb|Seq|
from the instance $\mathtt{Seq\; A}$ being defined, but also at
different instances from one another. The resulting interdependence of
different instances of GADTs means that they can express more
constraints than ADTs and nested types. For example, the ADT
$\mathtt{List}$ expresses the invariant that all of the data in the
lists it defines is of the same type, while the nested type
$\mathtt{PTree}$ expresses this invariant \emph{as well as} the
invariant that all of the lists it defines have lengths that are
powers of 2. The GADT $\mathtt{Seq}$ enforces \emph{even more general}
well-formedness conditions for sequences of values that simply cannot
be expressed with ADTs and nested types alone.

GADTs are widely used in modern functional languages, such as Haskell,
as well as in proof assistants, such as Agda, that are based on
dependent type theories. A natural next step in the line of work
reported in this paper is therefore to extend our parametricity
results to GADTs. A promising starting point for this endeavor is the
observation from~\cite{jg08} that the data objects of GADTs can be
represented using object-level left Kan extensions over discrete
categories. The more recent results of~\cite{jp19} further show that
adding a carefully designed object-level left Kan extension construct
to a calculus supporting primitive nested types preserves the
cocontinuity needed for primitive GADTs to have well-defined, properly
functorial interpretations. Together this suggests extending the type
system in Definition~\ref{def:wftypes} with such a left Kan extension
construct, and extending the calculus in Figure~\ref{fig:terms} with
corresponding categorically inspired constructs to introduce and
eliminate terms of these richer types. This approach exactly mirrors
the (entirely standard) approach taken above for product, coproduct,
and fixpoint types. In this section we outline the obvious approach to
extending our model from Section~\ref{sec:term-interp} to a parametric
model when some classes of primitive GADTs are incorporated in this
manner. However, as we argue at the end of this section, this naive
approach fails because the IEL does not hold. Unfortunately, the IEL
not holding derails more than just parametricity: it also affects the
well-definedness of the term semantics (in particular, the semantics
of $L$-terms).

\subsection{Left Kan Extensions}\label{sec:lke}

We begin by recalling the definition of a left Kan extension and
establishing some useful notation and results for them.

\begin{defi}\label{def:lke}
If $F : \set^k \to \set$ and $\ol{K} : \set^k \to \set^h$ are functors
over $\set$, then the \emph{left Kan extension of $F$ along $\ol K$} is
a functor $\mathit{Lan}_{\ol K}\,F : \set^h \to \set$ together with a
natural transformation $\eta : F \to (\mathit{Lan}_{\ol K}\,F)
\circ \ol K$
such that, for every functor $G : \set^h \to \set$ and natural
transformation $\gamma : F \to G \circ \ol K$, there exists a unique
natural transformation $\mu : \mathit{Lan}_{\ol K}\,F \to G$ such that
$(\mu {\ol K}) \circ \eta = \gamma$. This is depicted in the
following diagram:

\[
\begin{tikzcd}[row sep = huge]
\set^k
\ar[rr, "{F}"{name=Fa, above}, ""{name=F, below}]
\ar[rd, "{\ol K}"']
&& \set \\
& \set^h
\ar[Rightarrow, bend right = 25, from=F, "{\eta}"']
\ar[ur, bend left, "{\mathit{Lan}_{\ol K}\,F}"{name=Lan, description}]
\ar[ur, bend right, "{G}"'{name=L, right}, ""'{name=Lr, right}]
\ar[Rightarrow, bend left = 45, from=Fa, to=Lr, "{\gamma}" near start]
\ar[Rightarrow, dashed, from=Lan, to=L, "{\mu}"']
\end{tikzcd}
\]

\noindent
Replacing $\set$ by $\rel$ everywhere in Definition~\ref{def:lke} we
can similarly define the left Kan extension of $F$ along $\ol K$ for
functors $F$ and $\ol K$ over $\rel$.
\end{defi}

An alternative presentation characterizes the left Kan extension
$(\mathit{Lan}_{\ol K} F, \eta)$ in terms of the bijection between
natural transformations from $F$ to $G \circ \ol{K}$ and natural
transformations from $\mathit{Lan}_{\ol K} F$ to $G$, for which
$\eta$ is the unit. If Agda were to support a primitive
$\mathsf{Lan}$ for left Kan extensions, we could use this bijection
to rewrite the type of each data constructor for the GADT
$\mathtt{Seq}$ to arrive at the following equivalent representation:
\[\begin{array}{l}
\mathtt{data\; Seq \;(A : Set)\;:\;Set\;where}\\
\hspace*{0.4in}\mathtt{sconst\;:\; A \rightarrow Seq\;A}\\
\hspace*{0.4in}\mathtt{spair\;:\;(\mathsf{Lan}_{\lambda C\,D.\,C \times D}
  (\lambda C\,D.\,Seq\;C\times Seq\;D))\;A \rightarrow Seq\;A}\\
\hspace*{0.4in}\mathtt{sseq\;:\;(\mathsf{Lan}_{\lambda C.\,Nat \to C}
  (\lambda C.\,Nat \rightarrow Seq\;C))\;A \rightarrow
  Seq\;A}
\end{array}\]
\noindent
\!\!Our calculus will represent $\mathtt{Seq}$ and other GADTs in
precisely this way.

\medskip
A third representation of left Kan extensions in locally presentable
categories is given in terms of colimits. Writing ${\mathcal C}_0$ for
the full subcategory of finitely presentable objects in the locally
presentable category $\mathcal C$, the left Kan extension can be
expressed as
\[
(\mathit{Lan}_{\ol K} F)\ol A \;=\; \colim{\ol{X : \,{\mathcal
        C}_0},\, \ol{f : K{\ol X} \to A}}{F \ol X}\]
Thus, if $F : \set^k \to \set$, $\ol{K} :
  \set^k \to \set^h$, $\ol{A} : \set^h$, and $\set_0$ is the full
  subcategory of finitely presentable objects in $\set$, --- i.e., is
  the category of finite sets --- we  have that
\begin{equation}\label{eq:colim-form}
(\mathit{Lan}_{\ol K} F) \ol{A} = \colim{\ol{S} : \set_0^k, \,\ol{f :
      K\ol{S} \to  A}}{F \ol{S}}
\end{equation}
If $\ol{S} : \set_0^k$ and $\ol{f : K\ol{S} \to A}$, let $j_{\ol{S},
  \ol{f}} : F \ol{S} \to (\mathit{Lan}_{\ol K} F) \ol{A}$ be the morphism
indexed by $\ol{S}$ and $\ol{f}$ mapping the cocone into the colimit
in Equation~\ref{eq:colim-form}.
For any $h$-tuple of functions
$\ol{g : A \to B}$, the functorial action $(\mathit{Lan}_{\ol K
  }F) \ol{g}$
is the unique function from $(\mathit{Lan}_{\ol K} F) \ol{A}$ to
$(\mathit{Lan}_{\ol K} F) \ol{B}$ such that, for all $\ol{S} : \set^k_0$ and
$\ol{f : K\ol{S} \to A}$,\,
\begin{equation}\label{eq:cocone-def}
(\mathit{Lan}_{\ol K}F) \ol{g} \circ j_{\ol{S},\, \ol{f}}
= j'_{\ol{S},\, \ol{g\, \circ f}}
: F \ol{S} \to (\mathit{Lan}_{\ol K}F) \ol{B}
\end{equation}
holds, where $j$ is the morphism mapping the cocone into
$(\mathit{Lan}_{\ol K} F) \ol{A} = \colim{\ol{S} : \set_0^k, \,\ol{f :
    K\ol{S} \to A}}{F \ol{S}}$ and $j'$ is the morphism mapping the
cocone into $(\mathit{Lan}_{\ol K} F) \ol{B} = \colim{\ol{S} : \set_0^k,
  \,\ol{f : K\ol{S} \to B}}{F \ol{S}}$.  Moreover, if $\alpha : F \to
F'$ is a natural transformation, then $\mathit{Lan}_{\ol K} \alpha :
(\mathit{Lan}_{\ol K} F) \to (\mathit{Lan}_{\ol K} F')$ is defined to be the
induced natural transformation whose component $(\mathit{Lan}_{\ol K}
\alpha) \ol{A} : (\mathit{Lan}_{\ol K} F) \ol{A} \to (\mathit{Lan}_{\ol K} F')
\ol{A}$ is the unique function such that $(\mathit{Lan}_{\ol K} \alpha)
\ol{A} \,\circ\, j_{\ol{S}, \ol{f}} = j'_{\ol{S}, \ol{f}} \,\circ\,
\alpha_{\ol{S}}$, where $j$ is the cocone into $(\mathit{Lan}_{\ol K} F)
\ol{A}$ and $j'$ is the cocone into $(\mathit{Lan}_{\ol K} F') \ol{A}$. We
can similarly represent left Kan extensions of functors over $\rel$ in
terms of colimits. In that setting we will denote the morphism mapping
the cocone into the colimit $\iota$. We will make good use of both of
these representations in Section~\ref{sec:term-sem} below.

\medskip
More about each of the above three representations of left Kan
extensions and the connections between them can be found in,
e.g.,~\cite{rie16}.

\subsection{Extending the Calculus}\label{sec:ext-calc}

We incorporate GADTs into our calculus by first adding to
Definition~\ref{def:wftypes} the following type formation rule for
$\Lan$-types:

\[
\AXC{$\Gamma; \Phi,\ol\alpha^0 \vdash F$}
\AXC{$\ol{\Gamma;\ol\alpha^0 \vdash K}$}
\AXC{$\ol{\Gamma;\Phi \vdash A}$}
\TIC{$\Gamma;\Phi \vdash (\Lan^{\ol{\alpha^0}}_{\ol K} F) \ol A$}
\DisplayProof
\]

\noindent
Here, the type constructor $\Lan$ binds the variables in $\ol\alpha$,
and these variables must always have arity $0$. In addition, the
vectors $\ol K$ and $\ol A$ must have the same length.

Intuitively, $\Lan^{\ol{\alpha^0}}_{\ol K} F$ is a syntactic representation
of the left Kan extension of the functor in the variables in
$\ol\alpha$ denoted by $F$ along the functor in the variables in
$\ol\alpha$ denoted by $\ol K$. Using $\Lan$-types, we can therefore
represent the GADT $\mathtt{Seq}$ in our calculus as
 \begin{align*}
\mathtt{Seq}\,\alpha &\df \left( \mu \phi.\lambda \beta.
\beta + \left(\Lan^{\gamma_1,\gamma_2}_{\gamma_1 \times \gamma_2}
(\phi \gamma_1 \times \phi \gamma_2) \right) \beta \right.
\left. + \left(\Lan^{\gamma}_{\mathtt{Nat} \to \gamma} (\mathtt{Nat} \to \phi
\gamma) \right) \beta \right) \alpha
\end{align*}
\noindent
where $\mathtt{Nat} = \mu \alpha. 1 + \alpha$ and $\mathtt{Nat} \to C$
abbreviates $(\Lan^\emptyset_{\mathtt{Nat}} 1)\,C$, which equals
$\colim{f : \mathtt{Nat} \to C}{1}$ by Equation~\ref{eq:colim-form}
(here, $k = 0$). This representation of $\mathit{Seq}$ corresponds
exactly to the rewriting in Section~\ref{sec:lke}.  As explained in
Section~IV.\,D of~\cite{jp19}, the more general construct
$\Lan^{\ol{\alpha^0}}_{\ol K} F$ allowing extensions along vectors of
functors as depicted above makes it possible to represent GADTs with
two or more type arguments that depend on one another. Such GADTs
cannot be represented using just unary $\Lan$-types, i.e, $\Lan$-types
of the form $(\Lan^{\ol{\alpha^0}}_K F) A$.

As we will see below, in order to define the term-formation rules for
terms of $\Lan$-types we also need to generalize $\Nat$-types to bind
type constructor variables not just of arity $0$, but of arbitrary
arity. Accordingly, we replace the rule for $\Nat$-types in
Definition~\ref{def:wftypes} with

\[
\AXC{$\Gamma;\Phi \vdash F$}
\AXC{$\Gamma;\Phi  \vdash G$}
\BIC{$\Gamma;\emptyset \vdash \Nat^\Phi F \,G$}
\DisplayProof
\]

\noindent
Of course, we could have performed this replacement from the outset of
the present paper. But since using the above rule instead of the one
currently in Definition~\ref{def:wftypes} only makes the notation of
our calculus heavier without increasing its expressivity, we have
chosen not to.

Finally, adding $\Lan$-types to the types of our calculus also
requires the extension of Definition~\ref{def:second-order-subst} with
the following new clause:

\[((\Lan^{\ol{\alpha}}_{\ol K}\,F) \ol A)[\phi :=_{\ol{\beta}} G] =
(\Lan^{\ol\alpha}_{\ol K}\,F[\phi :=_{\ol{\beta}} G]) \ol{A[\phi
    :=_{\ol{\beta}} G]}\]

We must also extend our term calculus to accommodate GADTs. To this
end, we first give introduction and elimination rules appropriate to
our generalized $\Nat$-types, replacing the tenth and eleventh rules
in Figure~\ref{fig:terms} with

\[
\AXC{$\Gamma; \Phi \vdash F$}
\AXC{$\Gamma; \Phi \vdash G$}
\AXC{$\Gamma; \Phi \,|\, \Delta, x : F \vdash t: G$}
\TIC{$\Gamma; \emptyset
  \,|\, \Delta \vdash L_\Phi x.t : \Nat^\Phi \,F \,G$}
\DisplayProof
\]
and
\[
\AXC{$\ol{\Gamma;\Phi,\ol\beta \vdash K}$}
\AXC{$\Gamma; \emptyset
  \,|\, \Delta \vdash t : \Nat^{\ol\psi} \,F \,G$}
\AXC{$\Gamma;\Phi \,|\, \Delta \vdash s: F[\overline{\psi :=_{\ol\beta} K}]$}
\TIC{$\Gamma;\Phi\,|\, \Delta \vdash t_{\ol K} s:
  G[\overline{\psi :=_{\ol\beta} K}]$}
\DisplayProof
\]
respectively. Next, we add to the term calculus in
Figure~\ref{fig:terms} introduction and elimination rules for terms of
$\Lan$-types according to Definition~\ref{def:lke}. This gives the
rules

\[
\AXC{$\Gamma; \Phi,\ol{\alpha}\vdash F$}
\AXC{$\ol{\Gamma; \ol{\alpha} \vdash K}$}
\BIC{$\Gamma;\emptyset~|~\emptyset \vdash \int_{\ol K,F} : \Nat^{\Phi,\ol\alpha}\, F\, (\Lan^{\ol\alpha}_{\ol{K}}\,  F) \ol{K}$}
\DisplayProof
\]
and
\[
\AXC{$\Gamma;\emptyset~|~\Delta \vdash t : \Nat^{\Phi,\ol\alpha}\,
  F\;G[\ol{\beta := K}]$}
\UIC{$\Gamma;\emptyset~|~\Delta \vdash \partial^{G, \ol K}_F t :
\Nat^{\Phi, \ol\beta}\, (\Lan^{\ol\alpha}_{\ol K}\,F)\ol\beta\, G$}
\DisplayProof
\]

\noindent
respectively. Note that both of these rules make essential use of our
generalized $\Nat$-types. If $\Phi$ were required to be empty then we
would not be able to express, e.g., the types
\[\Nat^{\phi, \gamma,  \delta} (\phi \gamma \times
\phi \delta) \; (\phi (\gamma \times \phi))\] and
\[\Nat^{\phi, \beta}\,(\Lan^{\gamma,\delta}_{\gamma \times
  \delta} \,(\phi \gamma \times \phi \delta))\beta \; (\phi \beta)\]
associated with $\mathtt{Seq}$'s data constructor $\mathtt{spair}$.

\subsection{Extending the Type Semantics}\label{sec:type-sem}

To construct a parametric model of the kind we seek, we must give
suitable interpretations of our generalized $\Nat$-types and
$\Lan$-types, as well as of terms of both such types, in both $\set$
and $\rel$. (Note that it is having both $\mathsf{Set}$ and
$\mathsf{Rel}$ interpretations, and having these be interconnected in
the specific way parametricity demands, that is at issue here; if we
are only interested in having \emph{some} model of the type calculus,
then the one in~\cite{jp19} will do.) Recalling that the type
constructor variables bound by $\Nat$ in the formation rule for
generalized $\Nat$-types are interpreted in $\set$ as functors and in
$\rel$ as relation transformers, we let $k_i$ be the arity of the
$i^{th}$ element of the sequence $\Phi$ and extend the set and
relational interpretations for $\Nat$ types from
Section~\ref{sec:type-interp} as follows:
\begin{align*}
\setsem{\Gamma; \emptyset \vdash \Nat^\Phi\,F\,G}\rho
&= \{\eta : \lambda \ol{K}. \,\setsem{\Gamma; \Phi \vdash F}\rho[\Phi := \ol{K}]
  \Rightarrow \lambda \ol{K}.\,\setsem{\Gamma; \Phi \vdash G}\rho[\Phi := \ol{K}]~|~\\
&\hspace{0.3in} \forall \ol{K = (K^1, K^2, K^*) : RT_k}.\\
&\hspace{0.4in} (\eta_{\overline{K^1}}, \eta_{\overline{K^2}})
: \relsem{\Gamma; \Phi \vdash F}\Eq_{\rho}[\Phi := \ol{K}]
\rightarrow \relsem{\Gamma; \Phi \vdash G}\Eq_{\rho}[\Phi := \ol{K}] \} \\
\end{align*}
\begin{align*}
\relsem{\Gamma; \emptyset \vdash \Nat^\Phi \,F\,G}\rho
&= \{\eta : \lambda \ol{K}.\,\relsem{\Gamma; \Phi \vdash F}\rho[\Phi := \ol{K}]
\Rightarrow \lambda \ol{K}. \,\relsem{\Gamma; \Phi \vdash G}\rho[\Phi := \ol{K}]\}\\
&=
  \{(\eta_1,\eta_2) \in \setsem{\Gamma; \emptyset
    \vdash \Nat^\Phi
    \,F\,G} (\pi_1 \rho) \times \setsem{
    \Gamma;\emptyset
    \vdash \Nat^\Phi \,F\,G} (\pi_2
  \rho)~|~\\
& \hspace{0.3in} \forall \ol{K = (K^1, K^2, K^*) : RT_k}.\\
& \hspace{0.4in} ((\eta_1)_{\ol{K^1}},(\eta_2)_{\ol{K^2}}) \in
  (\relsem{\Gamma; \Phi \vdash G}\rho[\Phi := \ol{K}])^{\relsem{\Gamma;\Phi\vdash F}\rho[\Phi := \ol{K}]} \}\\
\end{align*}

Of course, we intend to interpret $\Lan$-types as actual left Kan
extensions. To get started, we define their set and relational
interpretations to be
\[
  \setsem{\Gamma;\Phi \vdash
    (\Lan^{\ol{\alpha}}_{\ol{K}}\,F)\ol{A}}\rho = (
  \textit{Lan}_{\lambda \ol{S}.\,\ol{\setsem{\Gamma; \ol\alpha \vdash K}
      \rho[\ol{\alpha := S}]}}\, \lambda \ol{S}.\,\setsem{\Gamma; \Phi, \ol\alpha
    \vdash F} \rho[\ol{\alpha := S}])\, \ol{\setsem{\Gamma; \Phi
      \vdash A}\rho}
\]
and
\[
  \relsem{\Gamma;\Phi \vdash
    (\Lan^{\ol{\alpha}}_{\ol{K}}\,F)\ol{A}}\rho = (
  \textit{Lan}_{\lambda \ol R.\,\ol{\relsem{\Gamma; \ol\alpha \vdash
        K}\rho[\ol{\alpha := R}]}}\, \lambda \ol R.\,\relsem{\Gamma;
    \Phi, \ol\alpha \vdash F}\rho[\ol{\alpha := R}])\,
  \ol{\relsem{\Gamma; \Phi \vdash A}\rho}
\]

\noindent
respectively. But in order to guarantee that the IEL, and thus
parametricity, continues to hold for the model we are constructing for
our extended calculus, we might think we have to cut down the set
interpretations of $\Lan$-types by restricting to the subcollection of
those interpretations that are well-behaved with respect to the
IEL\@. This would mirror the cutting down we have already seen in the
set interpretations of both our original $\Nat$-types and our of
generalized $\Nat$-types. However, as we will see in
Proposition~\ref{prop:full-lan} below, this will not be possible if we
want the set and relational interpretations of our $\Lan$-types to
include the set and relational interpretations of the terms given by
the $\Lan$-introduction rule. Since we do indeed want this, the upshot
is that the above interpretations of $\Lan$-types are, in fact, the
only possibility.

We must also define the functorial actions of the above set and
relational interpretations of $\Lan$-types on morphisms. If $f : \rho
\to \rho'$ is a morphism of set environments, then
\[
\setsem{\Gamma;\Phi \vdash (\Lan^{\ol{\alpha}}_{\ol{K}}\,F)\ol{A}}f
: \setsem{\Gamma;\Phi \vdash (\Lan^{\ol{\alpha}}_{\ol{K}}\,F)\ol{A}}\rho
\to \setsem{\Gamma;\Phi \vdash (\Lan^{\ol{\alpha}}_{\ol{K}}\,F)\ol{A}}\rho'
\]

\noindent
is defined to be
\begin{multline*}
\big(\textit{Lan}_{\lambda \ol{S}.\,\ol{\setsem{\Gamma; \ol\alpha
      \vdash K}\rho[\ol{\alpha := S}]}}\,
      \lambda \ol{S}.\,\setsem{\Gamma; \Phi, \ol\alpha \vdash
        F}f[\ol{\alpha := \id_S}]\big)\,
    \ol{\setsem{\Gamma; \Phi \vdash A}\rho'} \\
\circ\;
  \big(\textit{Lan}_{\lambda \ol{S}.\,\ol{\setsem{\Gamma; \ol\alpha
        \vdash K}\rho[\ol{\alpha := S}]}}\,
      \lambda \ol{S}.\,\setsem{\Gamma; \Phi, \ol\alpha \vdash
        F}\rho[\ol{\alpha := S}]\big)\,
    \ol{\setsem{\Gamma; \Phi \vdash A}f}
\end{multline*}
or, equivalently by naturality,
\begin{multline*}
\big(\textit{Lan}_{\lambda \ol{S}.\,\ol{\setsem{\Gamma; \ol\alpha
      \vdash K}\rho[\ol{\alpha := S}]}}\,
      \lambda \ol{S}.\,\setsem{\Gamma; \Phi, \ol\alpha \vdash
        F}\rho'[\ol{\alpha := S}]\big)\,
    \ol{\setsem{\Gamma; \Phi \vdash A}f} \\
\circ\;
  \big(\textit{Lan}_{\lambda \ol{S}.\,\ol{\setsem{\Gamma; \ol\alpha
        \vdash K}\rho[\ol{\alpha := S}]}}\,
     \lambda \ol{S}.\, \setsem{\Gamma; \Phi, \ol\alpha \vdash
       F}f[\ol{\alpha := \id_S}]\big)\,
    \ol{\setsem{\Gamma; \Phi \vdash A}\rho}
\end{multline*}

\noindent
The functorial action for the relational interpretations of
$\Lan$-types is defined analogously.

\subsection{Extending the Term Semantics}\label{sec:term-sem}

Recalling the notation for the colimit representation of left Kan
extensions from Section~\ref{sec:lke}, we can define the set and
relational interpretations of the terms from the introduction and
elimination rules for $\Lan$ types from Section~\ref{sec:ext-calc} by
\[
\setsem{\Gamma;\emptyset~|~\emptyset \vdash \int_{\ol K,F}
: \Nat^{\Phi, \ol\alpha}\, F\, (\Lan^{\ol\alpha}_{\ol{K}}\,  F) \ol{K}}\rho\,d
= \eta
\]
and
\[
\setsem{\Gamma;\emptyset~|~\Delta \vdash \partial^{G, \ol K}_F t
: \Nat^{\Phi, \ol\beta}\, (\Lan^{\ol\alpha}_{\ol K}\,F) \ol \beta\; G}\rho
= \mu
\]
respectively, where, for each $\ol{N : RT_k}$, $\eta_{\ol{N}}$
is the natural transformation associated with the left Kan extension

\[\begin{tikzcd} [column sep = large, row sep = large]
\set^{|\ol \alpha|}
\ar[rr, "\lambda \ol{A}.\,\setsem{\Gamma; \Phi, \ol\alpha
    \vdash F}\rho {[\Phi, \ol \alpha := \ol{N}, \ol{A} ]}"{name=F}]
\ar[dr, "\ol{\lambda \ol{A}.\,\setsem{\Gamma; \ol\alpha \vdash
      K}\rho {[\ol{\alpha := A} ]}}\;\;\;\;"{left}] && \set \\
&\set^{|\ol K|} \ar[ur, "\;\;\mathit{Lan}_{\ol{\lambda \ol A.\,
      \setsem{\Gamma; \ol\alpha \vdash K}\rho {[\ol{\alpha := A}
  ]}}}\, \lambda \ol A.\, \setsem{\Gamma; \Phi, \ol\alpha \vdash
    F}\rho {[\Phi, \ol \alpha :=  \ol N, \ol A ]}"{right}]
\ar[d, Rightarrow, shorten <= 2mm, shorten >= 1mm,
          from=F, to=2-2,  "\eta_{\ol{N}}"{}]
\end{tikzcd}\]
and, for all $\ol{N : RT_k}$ and $d : \setsem{\Gamma;\emptyset \vdash
  \Delta}$,
\[\begin{array}{l}
\mu_{\ol{N}\, d} :
\mathit{Lan}_{\ol{\lambda \ol A.\,\setsem{\Gamma; \ol\alpha \vdash K}\rho
  {[\ol{\alpha := A} ]}}} \lambda \ol A.\,\setsem{\Gamma; \Phi,
  \ol\alpha \vdash F}\rho {[\Phi, \ol \alpha := \ol N, \ol A ]}\\[2ex]
\hspace*{0.8in} \to \lambda \ol B.\, \setsem{\Gamma;
  \Phi, \ol\beta \vdash G} \rho[\Phi, \ol \beta := \ol N, \ol B]
\end{array}\]
is the unique natural transformation such that
\begin{equation}\label{eq:defn-of-mu-sub-N}
(\mu_{\ol{N}\, d} \, \ol{\lambda \ol A.\,\setsem{\Gamma; \ol\alpha \vdash K}
  \rho[\ol{\alpha := A}]}) \circ \eta_{\ol{N}}
= (\setsem{\Gamma; \emptyset \,|\, \Delta \vdash t : \Nat^{\Phi,
  \ol\alpha} \, F \,\, G[\ol{\beta := K}]}\rho \, d)_{\ol{N}}
\end{equation}
as given by the universal property of the left Kan extension.  The
relational interpretations\\ $\relsem{\Gamma;\emptyset~|~\emptyset
  \vdash \int_{\ol K,F} : \Nat^{\Phi, \ol\alpha}\, F\,
  (\Lan^{\ol\alpha}_{\ol{K}}\, F) \ol{K}}$ and
$\setsem{\Gamma;\emptyset~|~\Delta \vdash \partial^{G, \ol K}_F t :
  \Nat^{\Phi, \ol\beta}\, (\Lan^{\ol\alpha}_{\ol K}\,F) \ol \beta\;
  G}$ can then be defined entirely analogously.

\medskip
We conclude this subsection by showing, as promised above, that no
cutting down of set and relational interpretations of $\Lan$-types by
taking subsets is possible. We have:

\begin{prop}\label{prop:full-lan}
Let $F : \set^k \to \set$ and $\ol K : \set^k \to \set^h$ be functors.
If $L : \set^h \to \set$ is a functor such that
\begin{itemize}
\item $L \ol{A} \subseteq
  (\textit{Lan}_{\ol K} F)\ol{A}$ for all $\ol{A : \set}$,
\item $(\textit{Lan}_{\ol K} F)\ol{f}\,|_{L \ol A}\; x \in L
  \ol{B}$ for all $\ol{f : A \to B}$ and $x \in L \ol A$, and
\item $\eta_{\ol{A}}\, y \in L ( \ol{K \ol{A}} )$ for all $\ol{A :
  \set}$ and $y \in F \ol{A}$,
\end{itemize}
then $L = \textit{Lan}_{\ol K} F$.
\end{prop}
\begin{proof}
For all $\ol{A : \set}$ and $z \in (\textit{Lan}_{\ol K} F)\ol{A}$
there exist $\ol{S : \set_0}$, $\ol{f : K \ol{S} \to A}$, and $w \in F
\ol{S}$ such that $z = \iota_{\ol{S}, \ol{f}}\, w$.  Thus, $z =
\iota_{\ol{S}, \ol{f}}\, w = (\textit{Lan}_{\ol K} F) \ol{f}
(\eta_{\ol{S}} \,w)$ by Equation~\ref{eq:cocone-def}. Then by the
third assumption above we have that $\eta_{\ol{S}}\, w \in L (\ol{K
  \ol{S}})$, and by the second assumption above we have that $z \in L
\ol{A}$. This gives $(\textit{Lan}_{\ol K} F)\ol{A} \subseteq L \ol
A$. Finally, by the first assumption above, we therefore have that $L
\ol{A} = (\textit{Lan}_{\ol K} F)\ol{A}$.
\end{proof}
Thus, if $L$ were a restriction of $\lambda \ol B. \setsem{\Gamma; \Phi,\ol\beta
  \vdash (\Lan^{\ol\alpha}_{\ol K} F){\ol \beta}}\rho[\Phi := \ol
  N][\ol{\beta := B}]$,
if the functorial
action of $L$ were a restriction of that of $\lambda \ol B. \setsem{\Gamma;
  \Phi,\ol\beta \vdash (\Lan^{\ol\alpha}_{\ol K} F){\ol
    \beta}}\rho[\Phi := \ol N][\ol{\beta := B}]$, and if
$L(\ol{\setsem{\Gamma;\ol\alpha \vdash K}\rho[\ol{\alpha := A}]})$
contained $\setsem{\Gamma;\emptyset~|~\emptyset \vdash \!\int_{\ol
    K,F} : \Nat^{\Phi, \ol\alpha}\, F\, (\Lan^{\ol\alpha}_{\ol{K}}\,
  F) \ol{K}}\rho\,d\,\ol N\,\ol A\,y$ for all $\ol N$, $\ol A$, and
$y$, then $L$ would have to be the entirety of
 $\lambda \ol B. \setsem{\Gamma; \Phi,\ol\beta
  \vdash (\Lan^{\ol\alpha}_{\ol K} F){\ol \beta}}\rho[\Phi := \ol
  N][\ol{\beta := B}]$.
An analogous result holds for relational interpretations of
$\Lan$-types. This shows that the machinery of this section does not
give a well-defined relational semantics and, given the expected
properties of the term semantics, this cannot be
fixed by restricting the relational semantics of $\mathsf{Lan}$-types
from Section~\ref{sec:type-sem}.

\subsection{Parametricity and GADTs}

Having extended our calculus with both $\Lan$-types and terms of such
types, and having given sensible set and relational interpretations
for these types and terms, we now need to verify that these
interpretations give rise to a parametric model. The first step in
this process is to extend Lemma~\ref{lem:rel-transf-morph} to
$\Lan$-types by adding a clause for $\Lan$-types to the proof.
Unfortunately, however, Lemma~\ref{lem:rel-transf-morph} does not
extend to arbitrary $\Lan$-types, as the following example shows.

\begin{exa}\label{ex:no-rel-transf}
Consider the type $\emptyset; \alpha \vdash
(\Lan^{\emptyset}_{\onet}\, \onet) \alpha$. The analogue for $\rel$ of
Equation~\ref{eq:colim-form} gives
\[
\begin{array}{rl}
\relsem{ \emptyset; \alpha \vdash (\Lan^{\emptyset}_{\onet} \onet)
  \alpha } \rho [\alpha := (1, 0, 0)]
&= (\textit{Lan}_{\Eq_1} \Eq_1)(1, 0, 0) \\
&= \colim{f : \Eq_1 \to (1, 0, 0)}{\Eq_1} \\
&= (0, 0, 0)
\end{array}
\]
Here, the last equality holds because there are no morphisms in $\rel$
from $\Eq_1$ to $(1, 0, 0)$, and because $(0,0,0)$ is the initial
object in that category. On the other hand, for the set interpretation
with respect to the first projection of the relation environment
$\rho[\alpha := (1,0,0)]$, Equation~\ref{eq:colim-form} gives
\[
\begin{array}{rl}
\setsem{ \emptyset; \alpha \vdash (\Lan^{\emptyset}_{\onet} \onet)
  \alpha } ( \pi_1 ( \rho [\alpha := (1, 0, 0)] ) )
&= \setsem{ \emptyset; \alpha \vdash (\Lan^{\emptyset}_{\onet} \onet)
  \alpha } (\pi_1 \rho) [\alpha := 1] \\
&= (\textit{Lan}_{1} 1) 1 \\
&= \colim{1 \to 1}{1} \\
&= 1
\end{array}
\]
because $\id_1$ is the unique arrow from $1$ to $1$ in $\set$.
Lemma~\ref{lem:rel-transf-morph} therefore cannot hold.
\end{exa}

The problem in Example~\ref{ex:no-rel-transf} lies in the fact that
the type along whose interpretation we extend contains constants, i.e.,
subtypes constructed from $\onet$. Indeed, if $\ol{\emptyset;
  \overline{\alpha} \vdash K}$ consists only of polynomial (i.e.,
sum-of-products) types not containing constants (i.e., formed only
from $+$, $\times$, and the variables in $\alpha$), then
Lemma~\ref{lem:rel-transf-morph} actually does holds for $\Gamma; \Phi
\vdash (\Lan^{\alpha}_{\ol{K}} F)\ol{A}$. This is proved in the
following proposition, which covers the case of
Lemma~\ref{lem:rel-transf-morph} for such $\Lan$-types.

\begin{prop}
If $\ol{\emptyset; \overline{\alpha} \vdash K}$ consists only of
polynomial types not containing constants, then
\[(\setsem{ \Gamma; \Phi \vdash (\Lan^{\ol \alpha}_{\ol{K}} F)\ol{A}
}, \setsem{ \Gamma; \Phi \vdash (\Lan^{\ol \alpha}_{\ol{K}} F)\ol{A} },
\relsem{ \Gamma; \Phi \vdash (\Lan^{\ol \alpha}_{\ol{K}} F)\ol{A} })\]
is an $\omega$-cocontinuous environment transformer.
\end{prop}
\begin{proof}
We need only show that
\begin{equation}\label{eq:former}
\pi_i \,( \relsem{ \Gamma; \Phi \vdash (\Lan^{\ol \alpha}_{\ol{K}} F)\ol{A}
} \rho ) = \setsem{ \Gamma; \Phi \vdash (\Lan^{\ol \alpha}_{\ol{K}}
  F)\ol{A} } (\pi_i \rho)
\end{equation}
and
\begin{equation}\label{eq:latter}
\pi_i \,( \relsem{ \Gamma; \Phi \vdash (\Lan^{\ol \alpha}_{\ol{K}} F)\ol{A}
} f ) = \setsem{ \Gamma; \Phi \vdash (\Lan^{\ol \alpha}_{\ol{K}}
  F)\ol{A} } (\pi_i f)
\end{equation}
for all $\rho$, $\rho'$, $f : \rho \to \rho'$, and $i \in \{1, 2\}$.

To prove Equation~\ref{eq:former} we first observe that, by (the
analogue for $\rel$ of) Equation~\ref{eq:colim-form}, we have
\[
\relsem{ \Gamma; \Phi \vdash (\Lan^{\ol \alpha}_{\ol{K}} F)\ol{A} } \rho =
\colim{\ol{R : \rel_0},\, \ol{m : \relsem{ \emptyset; \ol\alpha \vdash K }
    \rho[\ol{\alpha := R}] \to \relsem{ \Gamma; \Phi \vdash A }
    \rho}}{ \relsem{ \Gamma; \Phi, \ol\alpha \vdash F }
  \rho[\ol{\alpha := R}] }
\]
with $\iota$ mapping the cocone into this colimit.  Since each
projection $\pi_i$ is cocontinuous, we therefore have that
\begin{align*}
&\pi_i\, ( \relsem{ \Gamma; \Phi \vdash (\Lan^{\ol \alpha}_{\ol{K}}
  F)\ol{A} } \rho ) \\
  &\quad {} = \colim{\ol{R : \rel_0}, \ol{m : \relsem{
      \emptyset; \ol\alpha \vdash K } \rho[\ol{\alpha := R}] \to
    \relsem{ \Gamma; \Phi \vdash A } \rho}}{ \setsem{ \Gamma; \Phi,
    \ol\alpha \vdash F } (\pi_i \rho) [\ol{\alpha := \pi_i R}]}
\end{align*}
with $\pi_i \iota$ mapping the cocone into this colimit.  On the
other hand, Equation~\ref{eq:colim-form} also gives
\begin{align*}
&\setsem{ \Gamma; \Phi \vdash (\Lan^{\ol \alpha}_{\ol{K}} F)\ol{A} } (\pi_i \rho) \\
&\quad {} = \colim{\ol{S : \set_0}, \ol{f : \setsem{ \emptyset;
      \ol\alpha \vdash K } (\pi_i \rho)[\ol{\alpha := S}] \to \setsem{
      \Gamma; \Phi \vdash A } (\pi_i \rho)}}{ \setsem{ \Gamma; \Phi,
    \ol\alpha \vdash F } (\pi_i \rho) [\ol{\alpha := S}]}
\end{align*}
with $j$ mapping the cocone into this colimit.

When $i = 1$ we prove that
\[
\pi_1 \,( \relsem{ \Gamma; \Phi \vdash (\Lan^{\ol \alpha}_{\ol{K}}
  F)\ol{A} } \rho ) = \setsem{ \Gamma; \Phi \vdash (\Lan^{\ol
    \alpha}_{\ol{K}} F)\ol{A} } (\pi_1 \rho)
\]
by
providing a pair of inverse functions; the proof when $i=2$ is
entirely analogous. To this end, we define
\[
h : \pi_1\, ( \relsem{ \Gamma; \Phi \vdash (\Lan^{\ol \alpha}_{\ol{K}}
  F)\ol{A} } \rho ) \to \setsem{ \Gamma; \Phi \vdash (\Lan^{\ol
    \alpha}_{\ol{K}} F)\ol{A} } (\pi_1 \rho)
\]
to be the unique function given by the universal property of the
colimit representation of $h$'s domain. That is, we take $h$ to be the
unique morphism such that, for any $\ol{R : \rel_0}$ and
\[\ol{m :
  \relsem{ \emptyset; \ol\alpha \vdash K } \rho[\ol{\alpha := R}] \to
  \relsem{ \Gamma; \Phi \vdash A } \rho}\] we have
\begin{equation}\label{eq:cocone-proj}
h \circ (\pi_1 \iota_{\ol{R},\, \ol{m}}) = j_{\ol{\pi_1 R},
  \,\ol{\pi_1 m}}
\end{equation}

Now, since $K$ does not contain constants, $\setsem{ \emptyset;
  \ol\alpha \vdash K } (\pi_2 \rho) [\ol{\alpha := 0}] = 0$. Thus, for
any $\ol{S : \set_0}$ and $\ol{f : \setsem{ \emptyset; \ol\alpha
    \vdash K } (\pi_1 \rho)[\ol{\alpha := S}] \to \setsem{ \Gamma;
    \Phi \vdash A } (\pi_1 \rho)} $ we have relations $\ol{(S, 0, 0)}$
and morphisms $\ol{(f, !) : \relsem{ \emptyset; \ol\alpha \vdash K }
  \rho[\ol{\alpha := (S, 0, 0)}] \to \relsem{ \Gamma; \Phi \vdash A }
  \rho}$, where, for any set $X$, we suppress the sub- and
superscripts and write $!$ for the unique morphism $!^0_X$ from $0$ to
$X$. We then define
\[
k : \setsem{ \Gamma; \Phi \vdash (\Lan^{\ol \alpha}_{\ol{K}} F)\ol{A}
} (\pi_1 \rho) \to \pi_1 \,( \relsem{ \Gamma; \Phi \vdash (\Lan^{\ol
    \alpha}_{\ol{K}} F)\ol{A} } \rho )
\]
to be the unique function given by the universal property of the
colimit representation of $k$'s domain. That is, we take $k$ to be the
unique morphism such that, for any $\ol{S : \set_0}$ and \[\ol{f :
  \setsem{ \emptyset; \ol\alpha \vdash K } (\pi_1 \rho)[\ol{\alpha :=
      S}] \to \setsem{ \Gamma; \Phi \vdash A } (\pi_1 \rho)}\] we have
$k \circ j_{\ol{S},\, \ol{f}} = \pi_1 \iota_{\ol{(S, 0, 0)},\, \ol{(f,
    !)}}$.

To see that $h$ and $k$ are mutually inverse we first observe that
\[
h \circ k \circ j_{\ol{S},\, \ol{f}}
= h \circ (\pi_1 \iota_{\ol{(S, 0, 0)}, \,\ol{(f, !)}})
= j_{\ol{\pi_1 (S, 0, 0)}, \,\ol{\pi_1 (f, !)}}
= j_{\ol{S},\, \ol{f}}
\]
for all $\ol{S : \set_0}$ and $\ol{f : \setsem{ \emptyset; \ol\alpha
    \vdash K } (\pi_1 \rho)[\ol{\alpha := S}] \to \setsem{ \Gamma;
    \Phi \vdash A } (\pi_1 \rho)}$, so that $h \circ k = \id$.
To see that  $k \circ h  = \id$ we first observe that
\begin{equation}\label{eq:k-comp-h}
k \circ h \circ (\pi_1 \iota_{\ol{R},\, \ol{m}}) = k \circ
j_{\ol{\pi_1 R}, \,\ol{\pi_1 m}} = \pi_1 \iota_{\ol{(\pi_1 R, 0,
    0)},\, \ol{(\pi_1 m, !)}}
\end{equation}
for all $\ol{R : \rel_0}$ and $\ol{m : \relsem{ \emptyset; \ol\alpha
    \vdash K } \rho[\ol{\alpha := R}] \to \relsem{ \Gamma; \Phi \vdash
    A } \rho}$. Then note that, for each sequence of morphisms
$\ol{(\id_{\pi_1 R}, !) : (\pi_1 R, 0, 0) \to R}$, we have that
\[\ol{m
\,\circ\, \relsem{ \emptyset; \ol\alpha \vdash K } \id_\rho[\ol{\alpha
    := (\id_{\pi_1 R}, !)}] = (\pi_1 m, !)}\]
in the indexing category for the colimit representation of the
codomain of $k$. This implies that
\[\iota_{\ol{R},\, \ol{m}} \,\circ\, \relsem{ \Gamma; \Phi,
  \ol\alpha \vdash F } \id_{\rho}[\ol{\alpha := (\id_{\pi_1 R}, !)}]
= \iota_{\ol{(\pi_1 R, 0, 0)},\, \ol{(\pi_1 m, !)}}\] Projecting the
first component thus gives
\[(\pi_1 \iota_{\ol{R},\,
  \ol{m}}) \, \circ\, \pi_1 ( \relsem{ \Gamma; \Phi, \ol\alpha \vdash
  F } \id_{\rho}[\ol{\alpha := (\id_{\pi_1 R}, !)}]) = \pi_1
\iota_{\ol{(\pi_1 R, 0, 0)},\, \ol{(\pi_1 m, !)}}\]
Now, the induction hypothesis on $\Gamma;\Phi,\ol\alpha
\vdash F$ gives that
\[\begin{array}{lll}
 &  & \pi_1 ( \relsem{ \Gamma; \Phi, \ol\alpha \vdash F}
\id_{\rho}[\ol{\alpha := (\id_{\pi_1 R}, !)}])\\
& = & \setsem{ \Gamma; \Phi, \ol\alpha \vdash F } \id_{\pi_1
  \rho}[\ol{\alpha := \id_{\pi_1 R}}]\\
& = & \id_{\setsem{ \Gamma; \Phi, \ol\alpha \vdash F } (\pi_1
  \rho)[\ol{\alpha := \pi_1 R}]}
\end{array}\]
so that, in fact, $\pi_1 \iota_{\ol R,\,\ol m} =
\pi_1\iota_{\ol{(\pi_1 R, 0, 0)},\, \ol{(\pi_1 m, !)}}$.  Finally, by
Equation~\ref{eq:k-comp-h} we have $k \circ h \circ (\pi_1
\iota_{\ol{R},\, \ol{m}}) = \pi_1 \iota_{\ol{(\pi_1 R, 0, 0)},\,
  \ol{(\pi_1 m, !)}}  = \pi_1 \iota_{\ol{R},\, \ol{m}}$ for all $\ol
R$ and $\ol m$, and thus $k \circ h = \id$.

To prove Equation~\ref{eq:latter}, first recall that
$\relsem{ \Gamma; \Phi \vdash (\Lan^{\ol \alpha}_{\ol{K}} F)\ol{A} }
f$ is defined to be
\begin{multline*}
\big(\textit{Lan}_{\ol{\lambda \ol{R}.\,\relsem{\Gamma; \ol\alpha
      \vdash K}\rho[\ol{\alpha := R}]}}\, \lambda
\ol{R}.\,\relsem{\Gamma; \Phi, \ol\alpha \vdash F}f[\ol{\alpha :=
    \id_R}]\big)\, \ol{\relsem{\Gamma; \Phi \vdash A}\rho'} \\ \circ\;
\big(\textit{Lan}_{\ol{\lambda \ol{R}.\,\relsem{\Gamma; \ol\alpha
      \vdash K}\rho[\ol{\alpha := R}]}}\, \lambda
\ol{R}.\,\relsem{\Gamma; \Phi, \ol\alpha \vdash F}\rho[\ol{\alpha :=
    R}]\big)\, \ol{\relsem{\Gamma; \Phi \vdash A}f}
\end{multline*}
and $\setsem{\Gamma; \Phi \vdash (\Lan^{\ol \alpha}_{\ol{K}} F)\ol{A}
} (\pi_i f)$ is defined to be
\begin{multline*}
\big(\textit{Lan}_{\ol{\lambda \ol{S}.\,\setsem{\Gamma; \ol\alpha
      \vdash K}(\pi_i \rho)[\ol{\alpha := S}]}}\, \lambda
\ol{S}.\,\setsem{\Gamma; \Phi, \ol\alpha \vdash F}(\pi_i f)[\ol{\alpha
    := \id_S}]\big)\, \ol{\setsem{\Gamma; \Phi \vdash A}(\pi_i
  \rho')}\\ \circ\; \big(\textit{Lan}_{\ol{\lambda
    \ol{S}.\,\setsem{\Gamma; \ol\alpha \vdash K}(\pi_i
    \rho)[\ol{\alpha := S}]}}\, \lambda \ol{S}.\,\setsem{\Gamma; \Phi,
  \ol\alpha \vdash F}(\pi_i \rho)[\ol{\alpha := S}]\big)\,
\ol{\relsem{\Gamma; \Phi \vdash A}(\pi_i f)}
\end{multline*}
It therefore suffices to show that
\begin{multline}\label{eq:part-one}
\pi_i \Big( \big(\textit{Lan}_{\ol{\lambda \ol{R}.\,\relsem{\Gamma; \ol\alpha
      \vdash K}\rho[\ol{\alpha := R}]}}\,
      \lambda \ol{R}.\,\relsem{\Gamma; \Phi, \ol\alpha \vdash F}f[\ol{\alpha := \id_R}]\big)\,
    \ol{\relsem{\Gamma; \Phi \vdash A}\rho'} \Big) \\
= \big(\textit{Lan}_{\ol{\lambda \ol{S}.\,\setsem{\Gamma; \ol\alpha
      \vdash K} (\pi_1 \rho) [\ol{\alpha := S}]}}\,
      \lambda \ol{S}.\,\setsem{\Gamma; \Phi, \ol\alpha \vdash F}(\pi_i f)[\ol{\alpha := \id_S}]\big)\,
    \ol{\setsem{\Gamma; \Phi \vdash A}(\pi_i \rho')}
\end{multline}
and that
\begin{multline}\label{eq:part-two}
\pi_i \Big( \big(\textit{Lan}_{\ol{\lambda \ol{R}.\,\relsem{\Gamma; \ol\alpha
        \vdash K}\rho[\ol{\alpha := R}]}}\,
      \lambda \ol{R}.\,\relsem{\Gamma; \Phi, \ol\alpha \vdash F}\rho[\ol{\alpha := R}]\big)\,
    \ol{\relsem{\Gamma; \Phi \vdash A}f} \Big) \\
= \big(\textit{Lan}_{\ol{\lambda \ol{S}.\,\setsem{\Gamma; \ol\alpha
        \vdash K} (\pi_i \rho) [\ol{\alpha := S}]}}\,
      \lambda \ol{S}.\,\setsem{\Gamma; \Phi, \ol\alpha \vdash F} (\pi_i \rho) [\ol{\alpha := S}]\big)\,
    \ol{\setsem{\Gamma; \Phi \vdash A} (\pi_i f)}
\end{multline}

\medskip
To prove Equation~\ref{eq:part-one} we first observe that the
morphisms $\iota'_{\ol{R}, \ol{m}} \,\circ\, \relsem{\Gamma; \Phi,
  \ol\alpha \vdash F}f[\ol{\alpha := \id_R}]$ form a cocone for
\[
\big(\textit{Lan}_{\ol{\lambda \ol{R}.\,\relsem{\Gamma; \ol\alpha
      \vdash K}\rho[\ol{\alpha := R}]}}\, \lambda
\ol{R}.\,\relsem{\Gamma; \Phi, \ol\alpha \vdash F}\rho[\ol{\alpha :=
    R}]\big)\, \ol{\relsem{\Gamma; \Phi \vdash A}\rho'}
\]
with vertex
\[
\big(\textit{Lan}_{\ol{\lambda \ol{R}.\,\relsem{\Gamma; \ol\alpha
      \vdash K}\rho[\ol{\alpha := R}]}}\, \lambda
\ol{R}.\,\relsem{\Gamma; \Phi, \ol\alpha \vdash F}\rho'[\ol{\alpha :=
    R}]\big)\, \ol{\relsem{\Gamma; \Phi \vdash A}\rho'}
\]
The universal property (Equation~\ref{eq:cocone-def}) of
\[
\big(\textit{Lan}_{\ol{\lambda \ol{R}.\,\relsem{\Gamma; \ol\alpha
      \vdash K}\rho[\ol{\alpha := R}]}}\, \lambda
\ol{R}.\,\relsem{\Gamma; \Phi, \ol\alpha \vdash F}f[\ol{\alpha :=
    \id_R}]\big)\, \ol{\relsem{\Gamma; \Phi \vdash A}\rho'}
\]
therefore gives that, for all $\ol{R : \rel_0}$ and $\ol{m :
  \relsem{\Gamma; \ol\alpha \vdash K}\rho[\ol{\alpha := R}] \to
  \relsem{\Gamma; \Phi \vdash A}\rho'}$,
\[\begin{array}{lll}
& & \big(\textit{Lan}_{\ol{\lambda \ol{R}.\,\relsem{\Gamma; \ol\alpha
      \vdash K}\rho[\ol{\alpha := R}]}}\, \lambda
\ol{R}.\,\relsem{\Gamma; \Phi, \ol\alpha \vdash F}f[\ol{\alpha :=
    \id_R}]\big)\, \ol{\relsem{\Gamma; \Phi \vdash A}\rho'} \,\circ\,
\iota_{\ol{R}, \ol{m}}\\[2ex]
& = & \iota'_{\ol{R}, \ol{m}} \,\circ\, \relsem{\Gamma; \Phi,
  \ol\alpha \vdash F}f[\ol{\alpha := \id_R}]
\end{array}\]
Here, $\iota$ is the morphism mapping the cocone into the colimit
\[
\big(\textit{Lan}_{\ol{\lambda \ol{R}.\,\relsem{\Gamma; \ol\alpha
      \vdash K}\rho[\ol{\alpha := R}]}}\, \lambda
\ol{R}.\,\relsem{\Gamma; \Phi, \ol\alpha \vdash F}\rho[\ol{\alpha :=
    R}]\big)\, \ol{\relsem{\Gamma; \Phi \vdash A}\rho'}
\]
and $\iota'$ is the morphism mapping the cocone into the colimit
\[
\big(\textit{Lan}_{\ol{\lambda \ol{R}.\,\relsem{\Gamma; \ol\alpha
      \vdash K}\rho[\ol{\alpha := R}]}}\, \lambda
\ol{R}.\,\relsem{\Gamma; \Phi, \ol\alpha \vdash F}\rho'[\ol{\alpha :=
    R}]\big)\, \ol{\relsem{\Gamma; \Phi \vdash A}\rho'}
\]
Projecting, together with the induction hypothesis for $\Gamma; \Phi,
\ol\alpha\vdash F$, thus gives
\[\begin{array}{lll}
 & & \pi_i \Big( \big(\textit{Lan}_{\ol{\lambda \ol{R}.\,\relsem{\Gamma;
      \ol\alpha \vdash K}\rho[\ol{\alpha := R}]}}\, \lambda
\ol{R}.\,\relsem{\Gamma; \Phi, \ol\alpha \vdash F}f[\ol{\alpha :=
    \id_R}]\big)\, \ol{\relsem{\Gamma; \Phi \vdash A}\rho'} \Big)
\,\circ\, \pi_i \iota_{\ol{R}, \ol{m}} \\[2ex]
& = & \pi_i \iota'_{\ol{R},
  \ol{m}} \,\circ\, \setsem{\Gamma; \Phi, \ol\alpha \vdash F} (\pi_i
f)[\ol{\alpha := \id_{\pi_i R}}]
\end{array}\]
Equation~\ref{eq:cocone-proj} and its analogue for
  $\iota'$ and $j'$ then give
\[\begin{array}{lll}
& & \pi_i \Big( \big(\textit{Lan}_{\ol{\lambda
    \ol{R}.\,\relsem{\Gamma; \ol\alpha \vdash K}\rho[\ol{\alpha :=
        R}]}}\, \lambda \ol{R}.\,\relsem{\Gamma; \Phi, \ol\alpha
  \vdash F}f[\ol{\alpha := \id_R}]\big)\, \ol{\relsem{\Gamma; \Phi
    \vdash A}\rho'} \Big) \,\circ \, k \,\circ\, j_{\ol{\pi_i R},
  \ol{\pi_i m}} \\[2ex]
&  = & \pi_i \Big( \big(\textit{Lan}_{\ol{\lambda \ol{R}.\,\relsem{\Gamma;
      \ol\alpha \vdash K}\rho[\ol{\alpha := R}]}}\, \lambda
\ol{R}.\,\relsem{\Gamma; \Phi, \ol\alpha \vdash F}f[\ol{\alpha :=
    \id_R}]\big)\, \ol{\relsem{\Gamma; \Phi \vdash A}\rho'} \Big)
\,\circ\, k \,\circ\, h \,\circ \,\pi_i \iota_{\ol{R}, \ol{m}} \\[2ex]
& = & \pi_i \Big( \big(\textit{Lan}_{\ol{\lambda \ol{R}.\,\relsem{\Gamma;
      \ol\alpha \vdash K}\rho[\ol{\alpha := R}]}}\, \lambda
\ol{R}.\,\relsem{\Gamma; \Phi, \ol\alpha \vdash F}f[\ol{\alpha :=
    \id_R}]\big)\, \ol{\relsem{\Gamma; \Phi \vdash A}\rho'} \Big)
\,\circ\, \pi_i \iota_{\ol{R}, \ol{m}} \\[2ex]
& = & \pi_i \iota'_{\ol{R},
  \ol{m}} \,\circ\, \setsem{\Gamma; \Phi, \ol\alpha \vdash F} (\pi_i
f)[\ol{\alpha := \id_{\pi_i R}}]\\[2ex]
& = & k' \, \circ \, h'\,\circ\,\pi_i \iota'_{\ol{R},
  \ol{m}} \,\circ\, \setsem{\Gamma; \Phi, \ol\alpha \vdash F} (\pi_i
f)[\ol{\alpha := \id_{\pi_i R}}]\\[2ex]
& = & k' \,\circ \, j'_{\ol{\pi_i R}, \ol{\pi_i m}} \,\circ\, \setsem{\Gamma; \Phi,
      \ol\alpha \vdash F} (\pi_i f)[\ol{\alpha := \id_{\pi_i R}}]
\end{array}\]
where $j$ is the morphism mapping the cocone into the colimit
\[
\big(\textit{Lan}_{\ol{\lambda \ol{S}.\,\setsem{\Gamma; \ol\alpha
      \vdash K}(\pi_i \rho)[\ol{\alpha := S}]}}\, \lambda
\ol{S}.\,\setsem{\Gamma; \Phi, \ol\alpha \vdash F}(\pi_i
\rho)[\ol{\alpha := S}]\big)\, \ol{\setsem{\Gamma; \Phi \vdash
    A}(\pi_i \rho')}
\]
and $j'$ is the morphism mapping the cocone into the colimit
\[
\big(\textit{Lan}_{\ol{\lambda \ol{S}.\,\setsem{\Gamma; \ol\alpha
      \vdash K}\rho[\ol{\alpha := S}]}}\, \lambda
\ol{R}.\,\setsem{\Gamma; \Phi, \ol\alpha \vdash F}(\pi_i
\rho')[\ol{\alpha := S}]\big)\, \ol{\setsem{\Gamma; \Phi \vdash
    A}(\pi_i \rho')}
\]
Since both $k$ and $k'$ are isomorphisms, we have that, up to
isomorphism,
\[\begin{array}{lll}
& & \pi_i \Big( \big(\textit{Lan}_{\ol{\lambda
    \ol{R}.\,\relsem{\Gamma; \ol\alpha \vdash K}\rho[\ol{\alpha :=
        R}]}}\, \lambda \ol{R}.\,\relsem{\Gamma; \Phi, \ol\alpha
  \vdash F}f[\ol{\alpha := \id_R}]\big)\, \ol{\relsem{\Gamma; \Phi
    \vdash A}\rho'} \Big) \circ\, j_{\ol{\pi_i R},
  \ol{\pi_i m}} \\[2ex]
& = & j'_{\ol{\pi_i R}, \ol{\pi_i m}} \,\circ\, \setsem{\Gamma; \Phi,
      \ol\alpha \vdash F} (\pi_i f)[\ol{\alpha := \id_{\pi_i R}}]
\end{array}\]
By the surjectivity of each $\pi_i : \rel \to \set$ we have that
\[\begin{array}{lll}
& & \pi_i \Big( \big(\textit{Lan}_{\ol{\lambda
    \ol{R}.\,\relsem{\Gamma;
      \ol\alpha \vdash K}\rho[\ol{\alpha := R}]}}\, \lambda
\ol{R}.\,\relsem{\Gamma; \Phi, \ol\alpha \vdash F}f[\ol{\alpha :=
    \id_R}]\big)\, \ol{\relsem{\Gamma; \Phi \vdash A}\rho'} \Big)
\,\circ\, j_{\ol{S}, \ol{g}} \\[2ex]
& = & j'_{\ol{S}, \ol{g}} \,\circ\, \setsem{\Gamma; \Phi, \ol\alpha
  \vdash F} (\pi_i f)[\ol{\alpha := \id_{S}}]
\end{array}\]
for all $\ol{S: \set_0}$ and $\ol{g : \setsem{\Gamma; \ol\alpha \vdash
    K}(\pi_i \rho)[\ol{\alpha := S}] \to \setsem{\Gamma; \Phi \vdash
    A}(\pi_i \rho')}$.  That is, \[\pi_i \Big(
\big(\textit{Lan}_{\ol{\lambda \ol{R}.\,\relsem{\Gamma; \ol\alpha
      \vdash K}\rho[\ol{\alpha := R}]}}\, \lambda
\ol{R}.\,\relsem{\Gamma; \Phi, \ol\alpha \vdash F}f[\ol{\alpha :=
    \id_R}]\big)\, \ol{\relsem{\Gamma; \Phi \vdash A}\rho'} \Big)\]
satisfies the universal property (Equation~\ref{eq:cocone-def}) of
\[
\big(\textit{Lan}_{\ol{\lambda \ol{S}.\,\setsem{\Gamma; \ol\alpha
      \vdash K} (\pi_1 \rho) [\ol{\alpha := S}]}}\, \lambda
\ol{S}.\,\setsem{\Gamma; \Phi, \ol\alpha \vdash F}(\pi_i f)[\ol{\alpha
    := \id_S}]\big)\, \ol{\setsem{\Gamma; \Phi \vdash A}(\pi_i \rho')}
\]
The two expressions must therefore be equal, and thus
Equation~\ref{eq:part-one} holds.

\medskip
To prove Equation~\ref{eq:part-two} we first observe that
  the universal property (Equation~\ref{eq:cocone-def}) of
\[
\big(\textit{Lan}_{\ol{\lambda \ol{R}.\,\relsem{\Gamma; \ol\alpha
      \vdash K}\rho[\ol{\alpha := R}]}}\, \lambda
\ol{R}.\,\relsem{\Gamma; \Phi, \ol\alpha \vdash F}\rho[\ol{\alpha :=
    R}]\big)\, \ol{\relsem{\Gamma; \Phi \vdash A}f}
\] gives that,
for every $\ol{R : \rel_0}$ and $\ol{m : \relsem{\Gamma; \ol\alpha
    \vdash K}\rho[\ol{\alpha := R}] \to \relsem{\Gamma; \Phi \vdash
    A}\rho}$,
\[\begin{array}{lll}
 & & \big(\textit{Lan}_{\ol{\lambda \ol{R}.\,\relsem{\Gamma; \ol\alpha
      \vdash K}\rho[\ol{\alpha := R}]}}\, \lambda
\ol{R}.\,\relsem{\Gamma; \Phi, \ol\alpha \vdash F}\rho[\ol{\alpha :=
    R}]\big)\, \ol{\relsem{\Gamma; \Phi \vdash A}f} \,\circ\,
\iota_{\ol{R}, \ol{m}}\\[2ex]
& = & \iota'_{\ol{R}, \ol{\relsem{\Gamma; \Phi \vdash A}f \,\circ\,
    m}}
\end{array}\]
where $\iota$ is the cocone into the colimit
\[
\big(\textit{Lan}_{\ol{\lambda \ol{R}.\,\relsem{\Gamma; \ol\alpha
      \vdash K}\rho[\ol{\alpha := R}]}}\, \lambda
\ol{R}.\,\relsem{\Gamma; \Phi, \ol\alpha \vdash F}\rho[\ol{\alpha :=
    R}]\big)\, \ol{\relsem{\Gamma; \Phi \vdash A}\rho}
\]
and $\iota'$ is the cocone into the colimit
\[\big(\textit{Lan}_{\ol{\lambda \ol{R}.\,\relsem{\Gamma; \ol\alpha
 \vdash K}\rho[\ol{\alpha := R}]}}\, \lambda \ol{R}.\,\relsem{\Gamma;
  \Phi, \ol\alpha \vdash F}\rho[\ol{\alpha := R}]\big)\,
  \ol{\relsem{\Gamma; \Phi \vdash A}\rho'}
\]
Projecting, together with the induction hypothesis for
$\Gamma;\Phi,\ol\alpha\vdash F$, thus gives
\[\begin{array}{lll}
 & & \pi_i \Big( \big(\textit{Lan}_{\ol{\lambda
    \ol{R}.\,\relsem{\Gamma; \ol\alpha
    \vdash K}\rho[\ol{\alpha := R}]}}\,
    \lambda \ol{R}.\,\relsem{\Gamma; \Phi, \ol\alpha \vdash
      F}\rho[\ol{\alpha := R}]\big)\,
    \ol{\relsem{\Gamma; \Phi \vdash A}f} \Big)
\,\circ\, \pi_i \iota_{\ol{R}, \ol{m}} \\[2ex]
& = & \pi_i \iota'_{\ol{R}, \ol{\relsem{\Gamma; \Phi \vdash A}f \,\circ\,
    m}}
\end{array}\]
Equation~\ref{eq:cocone-proj} and its analogue for
  $\iota'$ and $j'$ then give that, up to isomorphism,
\[\begin{array}{lll}
 & & \pi_i \Big( \big(\textit{Lan}_{\ol{\lambda \ol{R}.\,\relsem{\Gamma;
      \ol\alpha
    \vdash K}\rho[\ol{\alpha := R}]}}\,
    \lambda \ol{R}.\,\relsem{\Gamma; \Phi, \ol\alpha \vdash
      F}\rho[\ol{\alpha := R}]\big)\,
    \ol{\relsem{\Gamma; \Phi \vdash A}f} \Big)
\,\circ\, j_{\ol{\pi_i R}, \ol{\pi_i m}} \\[2ex]
& = & j'_{\ol{\pi_i R}, \ol{\relsem{\Gamma; \Phi \vdash A}(\pi_i f)
    \,\circ\, \pi_i m}}
\end{array}\]
where $j$ is the cocone into the colimit
\[
\big(\textit{Lan}_{\ol{\lambda \ol{S}.\,\setsem{\Gamma; \ol\alpha
      \vdash K}(\pi_i \rho)[\ol{\alpha := S}]}}\, \lambda
\ol{S}.\,\setsem{\Gamma; \Phi, \ol\alpha \vdash F}(\pi_i
\rho)[\ol{\alpha := S}]\big)\, \ol{\setsem{\Gamma; \Phi \vdash
    A}(\pi_i \rho)}
\]
and $j'$ is the cocone into the colimit
\[
\big(\textit{Lan}_{\ol{\lambda \ol{S}.\,\setsem{\Gamma; \ol\alpha
      \vdash K}(\pi_i \rho)[\ol{\alpha := S}]}}\, \lambda
\ol{S}.\,\setsem{\Gamma; \Phi, \ol\alpha \vdash F}(\pi_i
\rho)[\ol{\alpha := S}]\big)\, \ol{\relsem{\Gamma; \Phi \vdash
    A}(\pi_i \rho')}
\]
By the surjectivity of each $\pi_i : \rel \to \set$ we have that
\[\begin{array}{lll}
 & & \pi_i \Big( \big(\textit{Lan}_{\ol{\lambda
    \ol{R}.\,\relsem{\Gamma; \ol\alpha \vdash K}\rho[\ol{\alpha :=
        R}]}}\, \lambda \ol{R}.\,\relsem{\Gamma; \Phi, \ol\alpha
  \vdash F}\rho[\ol{\alpha := R}]\big)\, \ol{\relsem{\Gamma; \Phi
    \vdash A}f} \Big) \,\circ\, j_{\ol{S}, \ol{g}} \\[2ex]
& = & j'_{\ol{S}, \ol{\relsem{\Gamma; \Phi \vdash A}(\pi_i f) \,\circ\, g}}
\end{array}\]
for every $\ol{S: \set_0}$ and $\ol{g : \setsem{\Gamma; \ol\alpha
    \vdash K}(\pi_i \rho)[\ol{\alpha := S}] \to \setsem{\Gamma; \Phi
    \vdash A}(\pi_i \rho)}$. That is,
\[\pi_i \Big( \big(\textit{Lan}_{\ol{\lambda
    \ol{R}.\,\relsem{\Gamma; \ol\alpha \vdash K}\rho[\ol{\alpha :=
        R}]}}\, \lambda \ol{R}.\,\relsem{\Gamma; \Phi, \ol\alpha
  \vdash F}\rho[\ol{\alpha := R}]\big)\, \ol{\relsem{\Gamma; \Phi
    \vdash A}f} \Big)\]
satisfies the universal property (Equation~\ref{eq:cocone-def}) of
\[
\big(\textit{Lan}_{\ol{\lambda \ol{S}.\,\setsem{\Gamma; \ol\alpha
      \vdash K} (\pi_i \rho) [\ol{\alpha := S}]}}\, \lambda
\ol{S}.\,\setsem{\Gamma; \Phi, \ol\alpha \vdash F} (\pi_i \rho)
   [\ol{\alpha := S}]\big)\, \ol{\setsem{\Gamma; \Phi \vdash A} (\pi_i
     f)}
\]
The two expressions must therefore be equal, and thus
Equation~\ref{eq:part-two} holds.
\end{proof}

The requirement that the types in $\ol{K}$ be constant-free
polynomials is quite restrictive. In fact, it precludes the expression
of many GADTs commonly used in practice, such as the following GADT
$\mathtt{Expr}$ of typed expressions:
\[
\begin{array}{l}
\mathtt{data\; Expr \;(A : Set)\;:\;Set\;where}\\
\hspace*{0.4in}\mathtt{const\;:\; Int \to Expr\; Int}\\
\hspace*{0.4in}\mathtt{is\_zero\;:\; Expr\; Int \rightarrow Expr\; Bool}\\
\hspace*{0.4in}\mathtt{if\;:\; Expr\; Bool \to Expr\; A \to Expr\; A \to Expr\; A}
\end{array}
\]

Even when extending along constant-free polynomials, it is unclear how
to prove the Identity Extension Lemma for $\Lan$-types or, indeed,
whether it holds at all in their presence. To prove the IEL for
$\Nat$-types, for example, it is necessary to cut down the semantic
interpretation by requiring that the natural transformations in the
set interpretation preserve equalities. The impossibility of
restricting the semantic interpretations of $\Lan$-types, as shown in
Proposition~\ref{prop:full-lan}, suggests that the IEL may fail for
them. As noted above, this would derail not just parametricity, but
the well-definedness of the term semantics as well.  We have been able
neither to prove nor disprove the IEL thus far, but our substantial
and lengthy investigations into the issue lead us to suspect that it
may not actually be possible to define a functorial semantics for
$\Lan$-types that gives rise to parametric models for languages
supporting even constant-free polynomial primitive GADTs.

\section{Conclusion and Directions for Future Work}\label{sec:conclusion}

We have constructed a parametric model for a calculus providing
primitives for constructing nested types directly as fixpoints, rather
than representing them via their Church encodings. We have also used
the Abstraction Theorem for this model to derive free theorems for
nested types. This was not possible before~\cite{jp19} because such
types were not previously known to have well-defined interpretations
in locally finitely presentable categories (here, $\set$ and $\rel$),
and, to our knowledge, no term calculus for them existed either. The
key to obtaining our parametric model is the delicate threading of
functoriality and its accompanying naturality conditions throughout
our model construction.

We were surprised to find that, although GADTs were shown
in~\cite{jp19} to have appropriately cocontinuous functorial semantics
in terms of left Kan extensions, our model construction does not
extend to give a parametric model for them. It may be possible to
modify the categories of relations and relation transformers so that
if $(A,B,R)$ is a relation then $\pi_1 : R \to A$ and $\pi_2 : R \to
B$ are always surjective; this would prohibit situations like that in
Example~\ref{ex:no-rel-transf}, and might therefore allow us to
recover our IEL and, ultimately, an Abstraction Theorem appropriate to
GADTs.  If this turns out to be possible, then generalizing the
resulting model construction to locally $\lambda$-presentable
categories for $\lambda > \omega$ would make it possible to handle
broader classes of GADTs. (As shown in~\cite{jp19}, $\lambda >
\omega_1$ is required to interpret even common GADTs.) We could even
attempt to carry out our construction in locally $\lambda$-presentable
cartesian closed categories (lpcccs) $\cal C$ whose categories of
(abstract) relations, obtained by pullback as in~\cite{jac99}, are
also lpcccs and are appropriately fibred over $\cal C$. This would
give a framework for constructing parametric models for calculi with
primitive GADTs that is based on locally $\lambda$-presentable
fibrations, for some appropriate definition thereof.

The expressivity of folds for nested types has long been a vexing
issue (see, e.g.,~\cite{bm98}), and this is naturally inherited by the
calculus presented here. Since it codes all recursion using standard
folds, and since folds for nested types must return natural
transformations, many standard functions over nested types cannot be
represented in this calculus. Another important direction for future
work is therefore to improve the expressivity of our calculus by
adding, say, generalized folds~\cite{bp99}, or Mendler
iterators~\cite{amu05}, or otherwise extending standard folds to
express computations whose results are not natural transformations. In
particular, we may wish to add term-level fixpoints as, e.g.,
in~\cite{pit00}. This would require the categories interpreting types
to be not just locally $\lambda$-presentable, but also to support some
kind of domain structure. At the moment it seems that such an endeavor
will have to wait for a generalization of the results presented here
to at least locally $\omega_1$-presentable categories, however:
$\omega$-CPO, the most natural category of domains to replace $\set$,
is not locally finitely presentable.

\medskip
\noindent
{\small \sc Acknowledgment} We thank Daniel Jeffries for his
collaboration on~\cite{jgj21}, for his contributions to a preliminary
version of this paper, and for supporting publication of the present
version in his absence. We also thank the anonymous reviewers for
their detailed and thorough reviews, which have greatly improved the
paper.

\bibliographystyle{alpha}
\bibliography{references}

\newcommand{\etalchar}[1]{$^{#1}$}
\begin{thebibliography}{GJF{\etalchar{+}}15}

\bibitem[AMU05]{amu05}
A.~Abel, R.~Matthes, and T.~Uustalu.
\newblock Iteration and coiteration schemes for higher-order and nested
  datatypes.
\newblock {\em Theoretical Computer Science}, 333:3--66, 2005.

\bibitem[AR94]{ar94}
J.~Ad\'{a}mek and J.~Rosick\'{y}.
\newblock {\em Locally Presentable and Accessible Categories}.
\newblock Cambridge University Press, 1994.

\bibitem[Atk12]{atk12}
R.~Atkey.
\newblock Relational parametricity for higher kinds.
\newblock In {\em Computer Science Logic}, pages 46--61, 2012.

\bibitem[BFSS90]{bfss90}
E.~S. Bainbridge, P.~J. Freyd, A.~Scedrov, and P.~J. Scott.
\newblock Functorial polymorphism.
\newblock {\em Theoretical Computer Science}, 70:35--64, 1990.

\bibitem[BM98]{bm98}
R.~Bird and L.~Meertens.
\newblock Nested datatypes.
\newblock In {\em Mathematics of Program Construction}, pages 52--67, 1998.

\bibitem[BM05]{bm05}
L.~Birkedal and R.~E. M{\o}gelberg.
\newblock Categorical models for {A}badi and {P}lotkin's logic for
  parametricity.
\newblock {\em Mathematical Structures in Computer Science}, 15:709--772, 2005.

\bibitem[BP99]{bp99}
R.~Bird and R.~Paterson.
\newblock Generalised folds for nested datatypes.
\newblock {\em Formal Aspects of Computing}, 11:200--222, 1999.

\bibitem[Car84]{car97}
L.~Cardelli.
\newblock Type systems.
\newblock In {\em CRC Handbook of Computer Science and Engineering}, pages
  2208--2236. CRC Press, 1984.

\bibitem[CH03]{ch03}
J.~Cheney and R.~Hinze.
\newblock First-class phantom types.
\newblock Technical report, Cornell University, 2003.

\bibitem[DR04]{dr04}
B.~Dunphy and U.~Reddy.
\newblock Parametric limits.
\newblock In {\em Logic in Computer Science}, pages 242--251, 2004.

\bibitem[Gir72]{gir72}
J.-Y. Girard.
\newblock {\em Interpr\'etation fonctionnelle et \'elimination des coupures de
  l'arithmétique d'ordre sup\'erieur}.
\newblock PhD thesis, University of Paris, 1972.

\bibitem[GJF{\etalchar{+}}15]{gjfor15}
N.~Ghani, P.~Johann, F.~Nordvall Forsberg, F.~Orsanigo, and T.~Revell.
\newblock Bifibrational functorial semantics for parametric polymorphism.
\newblock In {\em Mathematical Foundations of Program Semantics}, pages
  165--181, 2015.

\bibitem[GLP93]{glp93}
A.\ Gill, J.\ Launchbury, and S.L.\ {Peyton Jones}.
\newblock A short cut to deforestation.
\newblock In {\em Functional Programming Languages and Computer Architecture},
  pages 223--232, 1993.

\bibitem[Has94]{has94}
R.~Hasegawa.
\newblock Categorical data types in parametric polymorphism.
\newblock {\em Mathematical Structures in Computer Science}, 4:71--109, 1994.

\bibitem[Jac99]{jac99}
B.~Jacobs.
\newblock {\em Categorical Logic and Type Theory}.
\newblock Elsevier, 1999.

\bibitem[JG08]{jg08}
P.~Johann and N.~Ghani.
\newblock Foundations for structured programming with {GADT}s.
\newblock In {\em Principles of Programming Languages}, pages 297--308, 2008.

\bibitem[JG10]{jg10}
P.~Johann and N.~Ghani.
\newblock Haskell programming with nested types: A principled approach.
\newblock {\em Higher-Order and Symbolic Computation}, 22(2):155--189, 2010.

\bibitem[JGJ21]{jgj21}
P.~Johann, E.~Ghiorzi, and D.~Jeffries.
\newblock Parametricity for primitive nested types.
\newblock In {\em Foundations of Software Science and Computation Structures},
  pages 297--308, 2021.

\bibitem[Joh02]{joh02}
P.\ Johann.
\newblock A generalization of short-cut fusion and its correctness proof.
\newblock {\em Higher-Order and Symbolic Computation}, 15:273--300, 2002.

\bibitem[JP19]{jp19}
P.~Johann and A.~Polonsky.
\newblock Higher-kinded data types: Syntax and semantics.
\newblock In {\em Logic in Computer Science}, pages 1--13, 2019.

\bibitem[JP20]{jp20}
P.~Johann and A.~Polonsky.
\newblock Deep induction: Induction rules for (truly) nested types.
\newblock In {\em Foundations of Software Science and Computation Structures},
  2020.

\bibitem[Mat11]{mat11}
R.~Matthes.
\newblock Map fusion for nested datatypes in intensional type theory.
\newblock {\em Science of Computer Programming}, 76(3):204--224, 2011.

\bibitem[MG01]{mg01}
C.~Martin and J.~Gibbons.
\newblock On the semantics of nested datatypes.
\newblock {\em Information Processing Letters}, 80(5):233--238, 2001.

\bibitem[MR92]{mr92}
Q.~Ma and J.~C. Reynolds.
\newblock Types, abstractions, and parametric polymorphism, part 2.
\newblock In {\em Mathematical Foundations of Program Semantics}, pages 1--40,
  1992.

\bibitem[Pit98]{pit98}
A.~Pitts.
\newblock Parametric polymorphism, recursive types, and operational
  equivalence.
\newblock Unpublished manuscript, 1998.

\bibitem[Pit00]{pit00}
A.M.\ Pitts.
\newblock Parametric polymorphism and operational equivalence.
\newblock {\em Mathematical Structures in Computer Science}, 10:321--359, 2000.

\bibitem[Rey83]{rey83}
J.~C. Reynolds.
\newblock Types, abstraction, and parametric polymorphism.
\newblock {\em Information Processing}, 83(1):513--523, 1983.

\bibitem[Rey84]{rey84}
J.~C. Reynolds.
\newblock Polymorphism is not set-theoretic.
\newblock {\em Semantics of Data Types}, pages 145--156, 1984.

\bibitem[Rie16]{rie16}
E.~Riehl.
\newblock {\em Category Theory in Context}.
\newblock Aurora, 2016.

\bibitem[RR94]{rr94}
E.~Robinson and G.~Rosolini.
\newblock Reflexive graphs and parametric polymorphism.
\newblock In {\em Logic in Computer Science}, pages 364--371, 1994.

\bibitem[Wad89]{wad89}
P.\ Wadler.
\newblock Theorems for free!
\newblock In {\em Functional Programming Languages and Computer Architecture},
  pages 347--359, 1989.

\bibitem[XCC03]{xcc03}
H.~Xi, C.~Chen, and G.~Chen.
\newblock Guarded recursive datatype constructors.
\newblock In {\em Principles of Programming Languages}, pages 224--235, 2003.

\end{thebibliography}

\end{document}